\documentclass[final,12pt]{colt2025} % Include author names

\usepackage{amsfonts,mathtools,mathrsfs,color}
\usepackage{caption}
\usepackage{subcaption}
\usepackage{bm}
\usepackage{enumerate}
\usepackage{soul}
\usepackage{tcolorbox}
\usepackage{framed}

\newtheorem{assumption}{Assumption}

\renewcommand{\P}{\mathcal{P}}
\newcommand{\R}{\mathbb{R}} 
\newcommand{\N}{\mathcal{N}}

\newcommand{\E}{\mathbb{E}}
\newcommand{\F}{\mathcal{F}}
\newcommand{\X}{\mathcal{X}}
\newcommand{\Y}{\mathcal{Y}}

\newcommand{\bX}{\mathbf{X}}
\newcommand{\bY}{\mathbf{Y}}
\newcommand{\bZ}{\mathbf{Z}}
\newcommand{\bx}{\mathbf{x}}
\newcommand{\by}{\mathbf{y}}
\newcommand{\bz}{\mathbf{z}}

\newcommand{\KL}{\mathsf{KL}}
\newcommand{\FI}{\mathsf{FI}}
\newcommand{\DG}{\mathsf{DG}}
\newcommand{\Tr}{\mathsf{Tr}}
\newcommand{\Cov}{\mathsf{Cov}}
\newcommand{\Var}{\mathsf{Var}}
\newcommand{\op}{\mathsf{op}}
\newcommand{\HS}{\mathsf{HS}}
\newcommand{\GD}{\mathsf{GD}}
\newcommand{\reg}{\tau}

\newcommand{\grad}{\mathrm{grad}}
\newcommand{\error}{\varepsilon}
\newcommand{\Unif}{\mathrm{Unif}}
\newcommand{\avg}{\mathrm{avg}}
\newcommand{\sym}{\mathrm{sym}}
\newcommand{\Err}{\mathsf{Err}}
\renewcommand{\part}[2]{\frac{\partial #1}{\partial #2}}

\newcommand{\red}[1]{{\textcolor{red}{#1}}}
\newcommand{\blue}[1]{{\textcolor{blue}{#1}}}
\newcommand{\purple}[1]{{\color{purple}{#1}}}

\title[On the Convergence of Min-Max Langevin Dynamics and Algorithm]{On the Convergence of 
Min-Max Langevin Dynamics
and Algorithm}
\usepackage{times}

\coltauthor{\Name{Yang Cai} \Email{yang.cai@yale.edu}\\
 \Name{Siddharth Mitra} \Email{siddharth.mitra@yale.edu}\\
 \Name{Xiuyuan Wang} \Email{xiuyuan.wang@yale.edu}\\
 \Name{Andre Wibisono} \Email{andre.wibisono@yale.edu}\\
 \addr Department of Computer Science, Yale University, New Haven, CT, USA}

\begin{document}

\maketitle

\begin{abstract}%
  We study zero-sum games in the space of probability distributions over the Euclidean space $\R^d$ with entropy regularization, in the setting when the interaction function between the players is smooth and strongly convex-strongly concave.
  We prove an exponential convergence guarantee for the mean-field min-max Langevin dynamics to compute the equilibrium distribution of the zero-sum game.
  We also study the finite-particle approximation of the mean-field min-max Langevin dynamics, both in continuous and discrete times.
  We prove biased convergence guarantees for the continuous-time finite-particle min-max Langevin dynamics to the stationary mean-field equilibrium distribution with an explicit bias term which does not scale with the number of particles.
  We also prove biased convergence guarantees for the discrete-time finite-particle min-max Langevin algorithm to the stationary mean-field equilibrium distribution with an additional bias term which scales with the step size and the number of particles.
  This provides an explicit iteration complexity for the average particle along the finite-particle algorithm to approximately compute the equilibrium distribution of the zero-sum game.
\end{abstract}


%%%%%%%%%%%%%%%%
\section{Introduction}
\label{Sec:Introduction}

Many tasks in computer science, economics, and machine learning can be formulated as \textit{games} in which two or more agents compete to optimize their own objective functions. Examples include Generative Adversarial Networks (GANs)~\citep{goodfellow2020generative}, adversarial learning~\citep{madry2018towards}, reinforcement learning~\citep{busoniu2008comprehensive}, and robust optimization~\citep{rahimian2022frameworks}. The solutions to these problems correspond to the Nash equilibria of these~games.

Two-player zero-sum games, first studied by \citet{borel1921theorie} and \citet{v1928theorie}, are arguably the most fundamental and well-studied class of games in game theory. Unlike classical game theory, where players are typically assumed to choose actions from a finite set, modern applications -- including those mentioned above -- require players to select actions from a continuous set $\X$, introducing substantial new challenges in both the existence and tractability of Nash equilibria.

Given a zero-sum game over action sets $\X$ and $\Y$ with payoff function $V \colon \X \times \Y \to \R$:
\begin{align}\label{Eq:Game00}
    \min_{x \in \X} \max_{y \in \Y} V(x,y),
\end{align}
the celebrated Minimax Theorem by \citet{v1928theorie} guarantees the existence of a pure (Nash) equilibrium $(x^*,y^*)$ of this game if $\X$ and $\Y$ are compact convex sets, and $V(x,y)$ is convex in $x$ and concave in $y$.\footnote{The finite-action setting is captured by choosing  $\X$  and  $\Y$  as simplices and  $V(x, y)$ as a bilinear function.} Unfortunately, when the payoff function is general, such a pure equilibrium need not exist. This motivates us to consider the more general solution concept of a mixed (Nash) equilibrium, which is a probability distribution over actions, and which is guaranteed to exist even for general payoff functions~\citep{glicksberg1952further}.
The problem of finding mixed equilibria in game~\eqref{Eq:Game00} can be recast as finding pure equilibria in the lifted game over the spaces of probability distributions \( \P(\X) \) and \( \P(\Y) \), where the payoff function is given by the expectation of the base payoff function:
\begin{align}\label{Eq:Game0}
    \min_{\rho^X \in \P(\X)} \max_{\rho^Y \in \P(\Y)} \E_{\rho^X \otimes \rho^Y}[V(X,Y)].
\end{align}
This has been studied in many works, including in~\citep{hsieh2019finding} for GANs.
The actions in the lifted game~\eqref{Eq:Game0} are probability distributions $\rho^X$ and $\rho^Y$ over the base sets $\X$ and $\Y$, so a pure equilibrium for the lifted game~\eqref{Eq:Game0} corresponds to a mixed equilibrium of the base game~\eqref{Eq:Game00};
we refer to a pure equilibrium of the lifted game~\eqref{Eq:Game0} as an \textit{equilibrium distribution}.

Despite its universality, finding an equilibrium distribution for the game~\eqref{Eq:Game0} may still be computationally challenging, depending on the base function $V$. 
Recent works, including~\citep{hsieh2019finding,domingo2020mean,ma2022provably,lu2023two,ding2024papal,kim2024symmetric}, have proposed regularizing the game by adding an entropy term to the payoff function with some regularization parameter $\reg > 0$;
this results in a game~\eqref{Eq:Game} that we study below.
The hope is that the entropy regularization makes the equilibrium distribution easier to compute, similar to what happens in the finite-dimensional problem.

In the game theory literature, the equilibrium distribution of the regularized game~\eqref{Eq:Game} with entropy regularization is known as the \textit{quantal response equilibrium} (QRE) defined by~\citet{mckelvey1995quantal}.\footnote{QRE  is originally defined for finite action games. We study its natural extension to continuous games in this paper.}
When the regularization parameter $\reg > 0$ is small, the QRE provides a good approximation to the equilibrium distribution of the original game~\eqref{Eq:Game0}.
We also observe that the QRE is the solution to one step of the proximal method with entropy regularization to compute the equilibrium distribution of the original game~\eqref{Eq:Game0}; therefore, if we can compute the QRE efficiently, then we can hope to run the conceptual prox method~\citep{nemirovski2004prox} which may have a good convergence property to the true equilibrium distribution of the game~\eqref{Eq:Game0}.

In this paper, we study this question in the setting when the payoff function $V$ is strongly convex-strongly concave. 
Two remarks are in order. First, in this setting, a pure equilibrium, i.e., an equilibrium distribution that is a point mass, of the game~\eqref{Eq:Game0} exists and is tractable~\citep{facchinei2003finite}.
However, these algorithms apply only to finite-dimensional settings, making them unsuitable for finding the QRE, an infinite-dimensional object.  
Second, in the space of probability distributions under the Wasserstein metric, the payoff function of the regularized game~\eqref{Eq:Game} %\( \F_\reg \) in~\eqref{Eq:PayoffDef} 
is geodesically strongly convex-strongly concave. 
In finite dimensions, for a strongly convex-strongly concave min-max optimization problem, the min-max gradient flow and its straightforward discretization -- the gradient descent-ascent algorithm -- converge at an exponential rate~\citep[Chapter~12.4.2]{facchinei2003finite}; see also Section~\ref{Sec:ReviewDeterministic} for a review of the deterministic game.
In our setting, the min-max Wasserstein gradient flow in the space of distributions corresponds to the mean-field min-max Langevin dynamics, arguably the most natural dynamics for solving the game.
Unlike in the finite-dimensional case, the convergence properties of this natural dynamics, along with its particle and time discretizations, remain unknown and were listed as an open question in~\citep{wang24open}.
All existing results for finding the QRE are for different settings or require modifications to the algorithms -- see review below.
% ~\cite{ma2022provably,lu2023two} study the compact manifold setting with two-timescale dynamics;~\cite{kim2024symmetric} study zero-sum games with entropy regularization with a convex-concave interaction functional, and analyze a mean-field Langevin averaged-gradient dynamics;~\cite{ding2024papal} consider the interaction function to be a bounded perturbation of a quadratic function; see Section~\ref{Sec:Related} for further discussion. 
This gap in our understanding motivates us to investigate the following question:

\begin{center}
    \itshape
   Does the mean-field min-max Langevin dynamics 
   converge for the regularized game~\eqref{Eq:Game}? If so, what are the convergence rates of its particle approximation and its discrete-time implementation?
\end{center}

%%%%%%%%%%%%%%%%
\paragraph{Related work}
Zero-sum games in the space of probability distributions have been studied in many recent works, including~\citep{hsieh2019finding,domingo2020mean,cen2021fast,wang2022exponentially,ma2022provably,lu2023two,kim2024symmetric,lascu2023entropic,lascu2024a,lascu2024mirror,ding2024papal,conger2024coupled,an2025convergence,lu2025convergencetimeaveragedmeanfield}.
We mention a few works here, and refer the reader to Section~\ref{Sec:Related} for further discussion. 
When the domains are compact Riemannian manifolds,
\cite{domingo2020mean} show that if the mean-field dynamics converges, then it must converge to the equilibrium distribution.
\cite{ma2022provably} and \cite{lu2023two} show the convergence of the continuous-time mean-field dynamics under timescale separation.
\cite{conger2024coupled} study the more general setting of min-max and cooperative games in the space of distributions. 
Notably,~\cite[Theorem~3.4]{conger2024coupled} show that for zero-sum games under strong convexity, the mean-field min-max Langevin dynamics has exponential convergence to the equilibrium distribution; thus, they have answered the open question by~\citep{wang24open} in the continuous-time mean-field setting under strong convexity.
In this work, we complement the results by providing guarantees for the finite-particle dynamics and the discrete-time algorithm.
When the domains are Euclidean spaces, 
\cite{kim2024symmetric} study zero-sum games with entropy regularization with a more general convex-concave interaction functional using the mean-field Langevin dynamics with a modified drift term replaced by the time average of the gradients; they prove a continuous-time convergence rate, as well as a convergence analysis of the finite-particle discrete-time algorithm.
\cite{ding2024papal} study a finite-particle discrete-time algorithm that implements the mirror-prox primal-dual algorithm in the space of distributions, 
% which requires an inner loop running a sampling algorithm to implement each iteration, 
and show explicit convergence guarantees of the resulting algorithm, under the assumption that the base payoff function is a bounded perturbation of a quadratic function.

\subsection{Problem Setting}

In this paper, we study zero-sum games in the space of probability distributions with interaction term which is an expectation over a base interaction function, with entropy regularization.
% In this paper, w
We work on the unconstrained state space $\X = \Y = \R^d$, for some dimension $d \ge 1$. 
Our techniques and results also generalize to $\X = \R^{d_X}$ and $\Y = \R^{d_Y}$ where $d_X \neq d_Y$, but here we set $d_X = d_Y = d$ for simplicity.
Let $\P(\R^d)$ denote the space of probability distributions over $\R^d$ which are absolutely continuous with respect to the Lebesgue measure and have a finite second moment.
Let $H(\rho) = -\E_\rho[\log \rho]$ be the entropy functional.

Let $V \colon \R^d \times \R^d \to \R$ be a given base payoff function.
We study the following min-max (or zero-sum) game on the product space of probability distributions:
\begin{align}\label{Eq:Game}
    \min_{\rho^X \in \P(\R^d)} \, \max_{\rho^Y \in \P(\R^d)} \; \F_\reg(\rho^X,\rho^Y)
\end{align}
where $\reg > 0$ is a regularization parameter, and the function $\F_\reg \colon \P(\R^d) \times \P(\R^d) \to \R$ is defined~as:
\begin{align}\label{Eq:PayoffDef}
    \F_\reg(\rho^X,\rho^Y)
    := \E_{\rho^X \otimes \rho^Y}[V] - \reg H(\rho^X) + \reg H(\rho^Y).
\end{align}

We say that a pair of probability distributions $(\bar \nu^X, \bar \nu^Y) \in \P(\R^d) \times \P(\R^d)$ is an \textit{equilibrium distribution} for the min-max game~\eqref{Eq:Game} if the following holds for all $(\rho^X, \rho^Y) \in \P(\R^d) \times \P(\R^d)$:
\begin{align}\label{Eq:Equilibrium}
    \F_\reg(\bar \nu^X, \rho^Y) \le \F_\reg(\bar \nu^X, \bar \nu^Y) \le \F_\reg(\rho^X, \bar \nu^Y).
\end{align}

The \textit{duality gap} $\DG \colon \P(\R^d) \times \P(\R^d) \to \R$ of the min-max game~\eqref{Eq:Game} is defined by:
\begin{align}\label{Eq:DG}
    \DG(\rho^X,\rho^Y) := \max_{\tilde \rho^Y \in \P(\R^d)} \, \F_\reg(\rho^X, \tilde \rho^Y) - \min_{\tilde \rho^X \in \P(\R^d)} \, \F_\reg(\tilde \rho^X, \rho^Y).
\end{align}
In the literature, this is also known as the Nikaido-Isoda error~\citep{nikaido1955note}.
Note that by construction, $\DG(\rho^X,\rho^Y) \ge 0$ for all $\rho^X, \rho^Y \in \P(\R^d)$, and furthermore, $\DG(\rho^X,\rho^Y) = 0$ if and only if $(\rho^X, \rho^Y)$ is an equilibrium distribution for the game~\eqref{Eq:Game}.

We are interested in characterizing the existence and uniqueness of the equilibrium distribution, as well as algorithms for approximately computing the equilibrium distribution in practice.
We study the case when the payoff function $V(x,y)$ is strongly convex in $x$, strongly concave in $y$, and has bounded second derivatives, see Assumption~\ref{As:SCSmooth}.
In this case, there exists a unique equilibrium distribution $(\bar \nu^X, \bar \nu^Y)$; see Section~\ref{Sec:Equilibrium}.

We now provide details for the dynamics and algorithms we study.

\paragraph{(1) The mean-field dynamics:}
This is a pair of stochastic processes $(\bar X_t)_{t \ge 0}$ and $(\bar Y_t)_{t \ge 0}$ in $\R^{d}$ which evolve via the following \textit{mean-field min-max Langevin dynamics} for all $t \ge 0$:
\begin{subequations}\label{Eq:MFSystem}
\begin{align}
    d\bar X_t &= -\E_{\bar \rho_t^Y}[\nabla_x V(\bar X_t, \bar Y_t)] \, dt + \sqrt{2\reg} \, dW_t^X \\
    d\bar Y_t &= \E_{\bar \rho_t^X}[\nabla_y V(\bar X_t, \bar Y_t)] \, dt + \sqrt{2\reg} \, dW_t^Y
\end{align}
\end{subequations}
where $\bar X_t \sim \bar \rho_t^X$ and $\bar Y_t \sim \bar \rho_t^Y$, and
where $(W_t^X)_{t \ge 0}$ and $(W_t^Y)_{t \ge 0}$ are independent standard Brownian motion in $\R^d$.
This comes from running the gradient flow dynamics in the space of probability distributions with the Wasserstein metric; see Section~\ref{Sec:MFDerivation} for derivation.
The dynamics~\eqref{Eq:MFSystem} is a \textit{mean-field} system because the evolution of $\bar X_t$ depends on the distribution $\bar \rho_t^Y$, and similarly, the evolution of $\bar Y_t$ depends on $\bar \rho_t^X$.
However, note the dependence is only via their expectation, so in the mean-field system above, $\bar X_t$ and $\bar Y_t$ evolve independently.
Therefore, if we initialize from independent $(\bar X_0, \bar Y_0) \sim \bar \rho_0^X \otimes \bar \rho_0^Y$, then $(\bar X_t, \bar Y_t) \sim \bar \rho_t^X \otimes \bar \rho_t^Y$ remains independent for all $t \ge 0$.

\vspace{-.05in}
\begin{framed}
    \noindent
    \textbf{Main Result 1:}
    When the payoff function $V$ is strongly concave-strongly convex, we show the mean-field dynamics $(\bar \rho_t^X, \bar \rho_t^Y)$ converges to the equilibrium distribution $(\bar \nu^X, \bar \nu^Y)$ exponentially fast in terms of duality gap, KL divergence, and Wasserstein distance; see Theorem~\ref{Thm:ConvMF}.
\end{framed}
\vspace{-.05in}

However, the mean-field system is an idealized algorithm that we cannot exactly implement due to its mean-field dependence and continuous-time nature. Thus, we study its approximations.

% \vspace{-.05in}

\paragraph{(2) The particle dynamics:}
This is an approximation of the mean-field dynamics~\eqref{Eq:MFSystem} by replacing $\bar X_t \sim \bar \rho_t^X$ and $\bar Y_t \sim \bar \rho_t^Y$ by $N \ge 1$ particles $\bX_t = (X_t^1,\dots, X_t^N) \sim \rho_t^{\bX}$ and $\bY_t = (Y_t^1,\dots, Y_t^N) \sim \rho_t^{\bY}$ in $\R^{dN}$, which jointly evolve via the following, for all $t \ge 0$ and for all $i \in [N] := \{1,\dots,N\}$:
\begin{subequations}\label{Eq:ParticleSystem}
\begin{align}
    dX_t^i &= -\frac{1}{N} \sum_{j \in [N]} \nabla_x V(X_t^i, Y_t^j) \, dt + \sqrt{2\reg} \, dW_t^{X,i} \\
    dY_t^i &= \frac{1}{N} \sum_{j \in [N]} \nabla_y V(X_t^j, Y_t^i) \, dt + \sqrt{2 \reg} \, dW_t^{Y,i}
\end{align}
\end{subequations}
where $(W_t^{X,i})_{t \ge 0}$ and $(W_t^{Y,i})_{t \ge 0}$ are independent standard Brownian motions in $\R^{d}$, for $i \in [N]$.
Note the above is not a mean-field system, but a standard system of interacting stochastic processes.
The distribution of $\bZ_t := (\bX_t, \bY_t) \sim \rho_t^{\bZ}$ in $\R^{2dN}$ is in general not independent: $\rho_t^{\bZ} \neq \rho_t^{\bX} \otimes \rho_t^{\bY}$.
Note that in the particle dynamics~\eqref{Eq:ParticleSystem}, we use the empirical mean from the $N$ particles to approximate the true expectation in the mean-field dynamics.
Thus, as $N \to \infty$, we expect the particle dynamics~\eqref{Eq:ParticleSystem} to become closer to the mean-field dynamics. 
As $t \to \infty$, the particle dynamics~\eqref{Eq:ParticleSystem} converges to a stationary distribution $(\bX_\infty, \bY_\infty) \sim \rho_\infty^{\bZ}$, which we expect to be close to the independent product of the stationary mean-field distribution $\bar \nu^{\bZ} = (\bar \nu^X)^{\otimes N} \otimes (\bar \nu^Y)^{\otimes N}$.

\vspace{-.05in}
\begin{framed}
    \noindent\textbf{Main Result 2:}
    For strongly convex-strongly concave and smooth $V$, 
    we prove a biased convergence guarantee of $\rho_t^{\bZ}$ along the finite-particle dynamics~\eqref{Eq:ParticleSystem} to the stationary mean-field distribution $\bar \nu^{\bZ}$, with a bias that is independent of the number of particles $N$; see Theorem~\ref{Thm:ConvergenceDynToMF}. 
\end{framed}
\vspace{-.1in}

\paragraph{(3) The particle algorithm:}
This is a time discretization of the particle dynamics~\eqref{Eq:ParticleSystem}, which maintains a collection of particles 
$\bx_k = (x_k^1,\dots,x_k^N) \sim \rho_k^{\bx,\eta}$ and $\by_k = (y_k^1, \dots, y_k^N) \sim \rho_k^{\by,\eta}$ in $\R^{dN}$ and evolves them via the following discrete-time update, for all $k \ge 0$ and for all $i \in [N]$:
\begin{subequations}\label{Eq:ParticleAlgorithm}
\begin{align}
    x_{k+1}^i &= x_k^i - \eta \cdot \frac{1}{N} \sum_{j \in [N]} \nabla_x V(x_k^i, y_k^j) + \sqrt{2\reg \eta} \, \zeta_k^{x,i} \\
    y_{k+1}^i &= y_k^i + \eta \cdot \frac{1}{N} \sum_{j \in [N]} \nabla_y V(x_k^j, y_k^i) + \sqrt{2 \reg \eta} \, \zeta_k^{y,i}
\end{align}
\end{subequations}
where $\eta > 0$ is a fixed step size, and $\zeta_k^{x,1}, \dots, \zeta_k^{x,N}, \zeta_k^{y,1}, \dots, \zeta_k^{x,N} \sim \N(0,I)$ are independent standard Gaussian random variables in $\R^d$.
Similar to standard discretization of the Langevin dynamics, this discrete-time algorithm is biased, i.e., as $k \to \infty$, the distribution of $\bz_k = (\bx_k, \by_k) \sim \rho_k^{\bz,\eta}$ converges to some distribution $\rho_\infty^{\bz,\eta}$ which is different from the stationary distribution $\rho_\infty^{\bZ}$ of the continuous-time particle dynamics~\eqref{Eq:ParticleSystem}.

\begin{framed}
    \noindent\textbf{Main Result 3:}
    For strongly convex-strongly concave and smooth $V$, we prove a biased convergence guarantee for $\rho_k^{\bz,\eta}$ along the discrete-time algorithm~\eqref{Eq:ParticleAlgorithm} to the stationary mean-field distribution $\bar \nu^{\bZ}$, with
    a bias that scales with step size $\eta$ and the number of particles $N$; see Theorem~\ref{Thm:ConvergenceAlgToMF}.
    This provides an explicit iteration complexity guarantee for the average particle of the 
    algorithm~\eqref{Eq:ParticleAlgorithm} to approximate
    the equilibrium distribution $\bar \nu^{Z}$ for the game~\eqref{Eq:Game}; see Corollary~\ref{Cor:ComplexityAlgToMF}.
\end{framed}

\paragraph{Organization:}
We provide definitions and assumptions in Section~\ref{Sec:Prelim}.
We discuss the properties of the equilibrium distribution and the duality gap in Section~\ref{Sec:Equilibrium}.
We discuss the convergence analysis of the mean-field min-max Langevin dynamics in Section~\ref{Sec:MF}.
We discuss the convergence analysis of the continuous-time finite-particle min-max Langevin dynamics in Section~\ref{Sec:Particle}, and the discrete-time finite-particle min-max Langevin algorithm in Section~\ref{Sec:Algorithm}.
We conclude with discussion in Section~\ref{Sec:Discussion}.
We provide additional details and proofs in the appendix.

%%%%%%%%%%%%%%%%%%%
\section{Preliminaries}
\label{Sec:Prelim}

%%%%%%%%%%%%%%%%%%%%
\subsection{Notation and Definitions}

Let $\P(\R^d)$ denote the space of probability distributions $\rho$ over $\R^d$ which are absolutely continuous with respect to the Lebesgue measure and which have finite second moment: $\E_{\rho}[\|X\|^2] < \infty$.
We identify a probability distribution $\rho \in \P(\R^d)$ with its probability density function (Radon-Nikodym derivative) $\rho \colon \R^d \to (0,\infty)$ with respect to the Lebesgue measure.

Let $H \colon \P(\R^d) \to \R$ be the \textit{entropy} functional:
$$H(\rho) = -\E_\rho[\log \rho] = -\int_{\R^d} \rho(x) \, \log \rho(x) \, dx.$$
 
The Wasserstein $W_2$ distance between probability distributions $\rho, \nu \in \P(\R^d)$ is defined by:
$$W_2(\rho,\nu) = \inf_{\gamma \in \Pi(\rho,\nu)} \E\left[\|X-Y\|^2\right]^{1/2},$$
where the infimum is over all couplings between $\rho$ and $\nu$, i.e., joint distributions of $(X,Y) \sim \gamma$ with the correct marginal distributions $X \sim \rho$ and $Y \sim \nu$.

For probability distributions $\rho, \nu \in \P(\R^d)$ with $\rho \ll \nu$ (i.e., if $\nu(x) = 0$, then $\rho(x) = 0$), the \textit{Kullback-Leibler (KL) divergence} or the \textit{relative entropy} of $\rho$ with respect to $\nu$ is defined by:
$$\KL(\rho\,\|\,\nu) = \E_\rho\left[\log \frac{\rho}{\nu}\right] = \int_{\R^d} \rho(x) \log \frac{\rho(x)}{\nu(x)} \, dx.$$

If $\rho, \nu \in \P(\R^d)$ have differentiable density functions, then the {\em relative Fisher information} of $\rho$ with respect to $\nu$ is defined by:
$$\FI(\rho\,\|\,\nu) = \E_\rho\left[\left\|\nabla \log \frac{\rho}{\nu}\right\|^2\right] = \int_{\R^d} \rho(x) \, \left\|\nabla \log \frac{\rho(x)}{\nu(x)}\right\|^2 \, dx.$$

We say a probability distribution $\nu$ is {\em $\alpha$-strongly log-concave} ($\alpha$-SLC) for some $\alpha > 0$ if the negative log-density $-\log \nu \colon \R^d \to \R$ is $\alpha$-strongly convex; if $\log \nu$ is twice differentiable, then this is equivalent to $-\nabla^2 \log \nu(x) \succeq \alpha I$ for all $x \in \R^d$.

We say $\nu$ satisfies {\em $\alpha$-log-Sobolev inequality} ($\alpha$-LSI) for some $\alpha > 0$ if for all probability distribution $\rho$, the following inequality holds:
$$\FI(\rho\,\|\,\nu) \ge 2\alpha \, \KL(\rho\,\|\,\nu).$$

We say $\nu$ satisfies {\em $\alpha$-Talagrand inequality} ($\alpha$-TI) for some $\alpha > 0$ if for all probability distribution $\rho$, the following inequality holds:
$$\KL(\rho\,\|\,\nu) \ge \frac{\alpha}{2} \, W_2(\rho, \nu)^2.$$

We recall that if $\nu$ is $\alpha$-SLC, then $\nu$ also satisfies $\alpha$-LSI;
furthermore, if $\nu$ satisfies $\alpha$-LSI, then $\nu$ also satisfies $\alpha$-TI, see~\citep{otto2000generalization}.

%%%%%%%%%%%%%%%%%%%%
\subsection{Assumptions}

Throughout this paper, we make the following assumption.

\begin{assumption}\label{As:SCSmooth}
    The base interaction function $V \colon \R^d \times \R^d \to \R$ is four-times continuously differentiable, $\alpha$-strongly convex in the first argument, $\alpha$-strongly concave in the second argument, 
    and is $L$-smooth of the second order, for some $0 < \alpha \le L < \infty$.
    That is, for all $x,y \in \R^d$:
    \begin{align*}
        \nabla^2_{xx} V(x,y) &\succeq \alpha I,
        \qquad
        -\nabla^2_{yy} V(x,y) \succeq \alpha I, 
        \qquad \text{ and } \qquad
        \|\nabla^2 V(x,y)\|_{\op} \le L.
    \end{align*}
    In particular, $(x,y) \mapsto \nabla V(x,y)$ is $L$-Lipschitz for all $x,y \in \R^d$. 
\end{assumption}

We note the assumption that $V$ be four-times differentiable is to ensure that our computations below are justified, in particular when using integration by parts such as in the proof of Lemma~\ref{Lem:ddtFI_MF}.

%%%%%%%%%%%%%%%%%%%
\section{Properties of the Equilibrium Distribution and Duality Gap}
\label{Sec:Equilibrium}

We define the \textit{best-response maps} $\Phi^X \colon \P(\Y) \to \P(\X)$ and $\Phi^Y \colon \P(\X) \to \P(\Y)$ by:
\begin{align*}
    \Phi^X(\rho^Y) = \nu^X
    \qquad \text{ and } \qquad
    \Phi^Y(\rho^X) = \nu^Y
\end{align*}
for all $\rho^X, \rho^Y \in \P(\R^d)$, where $\nu^X, \nu^Y \in \P(\R^d)$ are distributions with density given by:
\begin{subequations}\label{Eq:BestResponse}
\begin{align}
    \nu^X(x) &\propto \exp\left( -\reg^{-1}\E_{\rho^Y}[V(x,Y)] \right), \\
    \nu^Y(y) &\propto \exp\left( \reg^{-1} \E_{\rho^X}[V(X,y)] \right).    
\end{align}
\end{subequations}
Note that if we assume $V$ is strongly convex-strongly concave, then the maps above are well-defined, since the right-hand sides are integrable over $\R^d$.
More generally, we need some growth condition on $V$ to ensure the above are well-defined.
The term best-response is justified by the following property.
We provide the proof of Lemma~\ref{Lem:BestResponse} in Section~\ref{Sec:BestResponseProof}.

\begin{lemma}\label{Lem:BestResponse}
    For all $\rho^X, \rho^Y \in \P(\R^d)$, with $\nu^X = \Phi^X(\rho^Y)$ and $\nu^Y = \Phi^Y(\rho^X)$, we have:
    \begin{align*}
        \nu^X &= \arg\min_{\tilde \rho^X \in \P(\R^d)} \, \F_\reg(\tilde \rho^X, \rho^Y),  \\
        % \qquad\qquad
        \nu^Y &= \arg\max_{\tilde \rho^Y \in \P(\R^d)} \, \F_\reg(\rho^X, \tilde \rho^Y).
    \end{align*}
    Furthermore, the duality gap is given by:
    \begin{align*}
        \DG(\rho^X,\rho^Y) = \reg \, \KL(\rho^X \,\|\, \nu^X) + \reg \, \KL(\rho^Y \,\|\, \nu^Y).
    \end{align*}
\end{lemma}

Therefore, the equilibrium distribution $(\bar \nu^X, \bar \nu^Y)$, which minimizes the duality gap, is a fixed point of the best-response maps:
\begin{align*}
    \Phi^X(\bar \nu^Y) = \bar \nu^X
    \qquad \text{ and } \qquad
    \Phi^Y(\bar \nu^X) = \bar \nu^Y.
\end{align*}
We have the following property, which guarantees that if we can minimize the duality gap, then we also have convergence of the iterates to the equilibrium distribution, both in KL divergence and in Wasserstein distance.
The uniqueness of the equilibrium distribution follows from our convergence guarantee of the mean-field dynamics presented in Theorem~\ref{Thm:ConvMF}.
The bound~\eqref{Eq:DGBound} below is the same as stated in~\cite[Eq.~(4)]{wang24open} and~\cite[Lemma~3.5]{kim2024symmetric};
the first inequality in~\eqref{Eq:DGBound} follows from Talagrand's inequality, and the second inequality follows from the optimality property of the duality gap.
We provide the proof of Lemma~\ref{Lem:Equilibrium} in Section~\ref{Sec:EquilibriumProof}.

\begin{lemma}\label{Lem:Equilibrium}
    Assume Assumption~\ref{As:SCSmooth}.
    Then there exists a unique equilibrium distribution $(\bar \nu^X, \bar \nu^Y) \in \P(\R^d) \times \P(\R^d)$ for the game~\eqref{Eq:Game}, and the distributions $\bar \nu^X$ and $\bar \nu^Y$ are $(\alpha/\reg)$-strongly log-concave.
    Furthermore, for all $(\rho^X, \rho^Y) \in \P(\R^d) \times \P(\R^d)$:
    \begin{align}\label{Eq:DGBound}
        \frac{\alpha}{2} \left( W_2(\rho^X, \bar \nu^X)^2 + W_2(\rho^Y, \bar \nu^Y)^2 \right)
        \le \reg \left( \KL(\rho^X \,\|\, \bar \nu^X) + \KL(\rho^Y \,\|\, \bar \nu^Y) \right)
        \le \DG(\rho^X,\rho^Y).
    \end{align}
\end{lemma}

%%%%%%%%%%%%%%%%%%%
\section{Analysis of the Mean-Field Min-Max Langevin Dynamics}
\label{Sec:MF}

Given a pair of random variables $\bar X_t \sim \bar \rho_t^X$ and $\bar Y_t \sim \bar \rho_t^Y$ in $\R^d$, we can define the joint random variable $\bar Z_t = (\bar X_t, \bar Y_t) \in \R^{2d}$ which has an independent joint distribution $\bar Z_t \sim \bar \rho_t^Z := \bar \rho_t^X \otimes \bar \rho_t^Y$.
Define the independent product of the best-response distributions 
$$\bar \nu_t^Z := \bar \nu_t^X \otimes \bar \nu_t^Y$$ 
where $\bar \nu_t^X = \Phi^X(\bar \rho_t^Y)$ and $\bar \nu_t^Y = \Phi^Y(\bar \rho_t^X)$.
We observe that if $\bar X_t \sim \bar \rho_t^X$ and $\bar Y_t \sim \bar \rho_t^Y$ evolve following the mean-field min-max Langevin dynamics~\eqref{Eq:MFSystem},
then the joint variable $\bar Z_t \sim \bar \rho_t^Z$ evolves following the Langevin dynamics targeting its best-response distribution $\bar \nu_t^Z$:
\begin{align}\label{Eq:MFJoint}
    d\bar Z_t = \reg \nabla \log \bar \nu_t^Z(\bar Z_t) \, dt + \sqrt{2\reg} \, dW_t^Z
\end{align}
where $W_t^Z = (W_t^X, W_t^Y)$ is the standard Brownian motion in $\R^{2d}$.
This follows from the definition of the best-response distributions~\eqref{Eq:BestResponse} and computing their score functions.

We have the following convergence results for the mean-field min-max Langevin dynamics.
We show that the relative Fisher information between the iterates and their best-response distributions converges exponentially fast along the mean-field Langevin dynamics (Theorem~\ref{Thm:ConvMF} part (1)).
This is analogous to the result in the deterministic setting where the squared norm of the velocity vector field converges exponentially fast along the min-max gradient flow (see Theorem~\ref{Thm:DetMinMaxGF} in Section~\ref{Sec:ReviewDeterministic} for a review).
For the mean-field min-max Langevin dynamics, which is the min-max Wasserstein gradient flow, the squared Wasserstein norm of the velocity of $\bar \rho_t^Z$ is precisely the relative Fisher information between $\bar \rho_t^Z$ and its best-response distribution $\bar \nu_t^Z$, which motivates our result below.
The convergence in relative Fisher information also implies the convergence in the duality gap or KL divergence, by the log-Sobolev inequality (Theorem~\ref{Thm:ConvMF} part (2)).

Below, we write $\DG(\bar \rho_t^Z) \equiv \DG(\bar \rho_t^X, \bar \rho_t^Y)$.
We provide the proof of Theorem~\ref{Thm:ConvMF} in Section~\ref{Sec:ConvMFProof}.

\begin{theorem}\label{Thm:ConvMF}
    Assume Assumption~\ref{As:SCSmooth}.
    Suppose $\bar Z_t \sim \bar \rho_t^Z$ evolves following the mean-field dynamics~\eqref{Eq:MFJoint} in the joint space $\R^{2d}$ from $\bar Z_0 \sim \bar \rho_0^Z \in \P(\R^{2d})$, and let $\bar \nu_t^Z$ be the best-response distribution as defined above. 
    Then for all $t \ge 0$, we have the following properties:
    \begin{enumerate}
        \item Convergence in relative Fisher information:~
        $$\FI(\bar \rho_t^Z \,\|\, \bar \nu_t^Z) \le e^{-2\alpha t} \, \FI(\bar \rho_0^Z \,\|\, \bar \nu_0^Z).$$
        \item Convergence in duality gap:~
        $$\DG(\bar \rho_t^Z) = \reg \, \KL(\bar \rho_t^Z \,\|\, \bar \nu_t^Z) \le e^{-2\alpha t} \, \frac{\reg^2}{2\alpha} \, \FI(\bar \rho_0^Z \,\|\, \bar \nu_0^Z).$$
    \end{enumerate}
\end{theorem}

Let $\bar \nu^Z := \bar \nu^X \otimes \bar \nu^Y$ be the stationary distribution for the mean-field dynamics in the joint space $\R^{2d}$.
Note that by the bound~\eqref{Eq:DGBound} from Lemma~\ref{Lem:Equilibrium}, the convergence rate in duality gap above also implies the following convergence guarantees between the iterate $\bar \rho_t^Z$ of the mean-field dynamics and the stationary mean-field distribution $\bar \nu^Z$:
\begin{align}\label{Eq:ConvMFImplication}
    \frac{\alpha}{2} \, W_2(\bar \rho_t^Z, \bar \nu^Z)^2
    \,\le\, \reg \, \KL(\bar \rho_t^Z \,\|\, \bar \nu^Z)
    \,\le\, \DG(\bar \rho_t^Z) 
    \,\le\, e^{-2\alpha t} \, \frac{\reg^2}{2\alpha} \, \FI(\bar \rho_0^Z \,\|\, \bar \nu_0^Z).
\end{align}
The bound above is in terms of the initial relative Fisher information of $\bar \rho_0^Z$ to its best-response distribution.
When we choose a Gaussian initial distribution $\bar \rho_0^Z = \N(0, \frac{\reg^2}{L^2} I)$, we can bound $\FI(\bar \rho_0^Z \,\|\, \bar \nu_0^Z) = O(\frac{dL^2}{\reg^2})$; see Lemma \ref{Lem:InitialGaussianFI} in Section~\ref{Sec:BoundFIBestResponse}.

We note that under Assumption~\ref{As:SCSmooth}, one can show that the mean-field dynamics~\eqref{Eq:MFJoint} is in fact an exponential contraction in the $W_2$ distance, which implies $W_2(\bar \rho_t^Z, \bar \nu^Z)^2 \le e^{-2\alpha t} \, W_2(\bar \rho_0^Z, \bar \nu^Z)^2$, see~\cite[Proposition~7.3]{conger2024coupled}.
We note that~\cite[Lemma~7.2]{conger2024coupled} also shows the exponential convergence of the dissipation functional, which is the same as the relative Fisher information above in our setting.
We provide a self-contained proof of Theorem~\ref{Thm:ConvMF} in Section~\ref{Sec:ConvMFProof}; in particular, we derive an identity of the time derivative of the relative Fisher information along the mean-field Langevin dynamics, see Lemma~\ref{Lem:ddtFI_MF} in Section~\ref{Sec:ddtFI_MF}.

%%%%%%%%%%%%%%%%%%%
\section{Analysis of the Finite-Particle Min-Max Langevin Dynamics}
\label{Sec:Particle}

Here we study a finite-particle approximation of the mean-field min-max Langevin dynamics~\eqref{Eq:MFSystem} where we replace each $\bar X_t \in \R^d$ and $\bar Y_t \in \R^d$ with $N \ge 1$ particles $X_t^1,\dots,X_t^N \in \R^d$ and $Y_t^1,\dots,Y_t^N \in \R^d$, which evolve following the finite-particle system of dynamics~\eqref{Eq:ParticleSystem} where we use the empirical mean from the particles in place of the true expectation in the drift terms of~\eqref{Eq:MFSystem}.

We can write the finite-particle dynamics~\eqref{Eq:ParticleSystem} in terms of the joint vectors $\bX_t = (X_t^1,\dots,X_t^N) \in \R^{dN}$ and $\bY_t = (Y_t^1, \dots, Y_t^N) \in \R^{dN}$ as:
\begin{subequations}\label{Eq:ParticleSystem2}
\begin{align}
    d\bX_t &= b^{\bX}(\bX_t, \bY_t) \, dt + \sqrt{2\reg} \, dW_t^{\bX} \\
    d\bY_t &= b^{\bY}(\bX_t, \bY_t) \, dt + \sqrt{2 \reg} \, dW_t^{\bY}
\end{align}
\end{subequations}
where $(W_t^{\bX})_{t \ge 0}$ and $(W_t^{\bY})_{t \ge 0}$ are independent standard Brownian motions in $\R^{dN}$.
In the above, we have defined the vector fields $b^{\bX} \colon \R^{dN} \times \R^{dN} \to \R^{dN}$ and $b^{\bY} \colon \R^{dN} \times \R^{dN} \to \R^{dN}$ by, for all $\bx = (x^1,\dots,x^N) \in \R^{dN}$ and $\by = (y^1,\dots,y^N) \in \R^{dN}$:
\begin{subequations}\label{Eq:Defbxby}
\begin{align}
    b^{\bX}(\bx,\by) &= 
    \begin{pmatrix}
        b^X(x^1, \by) \\
        \cdots \\
        b^X(x^N, \by)
    \end{pmatrix}
    := \begin{pmatrix}
        -\frac{1}{N} \sum_{j \in [N]} \nabla_x V(x^1, y^j) \\
        \cdots \\
        -\frac{1}{N} \sum_{j \in [N]} \nabla_x V(x^N, y^j)
    \end{pmatrix}, \\
    b^{\bY}(\bx,\by) &= 
    \begin{pmatrix}
        b^Y(\bx, y^1) \\
        \cdots \\
        b^Y(\bx, y^N)
    \end{pmatrix}
    := \begin{pmatrix}
        \frac{1}{N} \sum_{j \in [N]} \nabla_y V(x^j, y^1) \\
        \cdots \\
        \frac{1}{N} \sum_{j \in [N]} \nabla_y V(x^j, y^N)
    \end{pmatrix}.
\end{align}
\end{subequations}

We can further write the finite-particle dynamics~\eqref{Eq:ParticleSystem2} in terms of the joint random variable $\bZ_t = (\bX_t, \bY_t) \in \R^{2dN}$ as:
\begin{align}\label{Eq:ParticleSystem3}
    d\bZ_t = b^{\bZ}(\bZ_t) \, dt + \sqrt{2\reg} \, dW_t^{\bZ}
\end{align}
where $W_t^{\bZ} = (W_t^{\bX}, W_t^{\bY})$ is the standard Brownian motion in $\R^{2dN}$, and we have defined the vector field $b^{\bZ} \colon \R^{2dN} \to \R^{2dN}$ by, for all $\bz = (\bx,\by) \in \R^{2dN}$:
\begin{align}\label{Eq:Defbz}
    b^{\bZ}(\bz) = \begin{pmatrix}
        b^{\bX}(\bx,\by) \\
        b^{\bY}(\bx,\by)
    \end{pmatrix}.
\end{align}

%%%%%%%%%%%%%%%%%%%
\subsection{Biased Convergence of Finite-Particle Dynamics to Stationary Mean-Field Distribution}

Recall $\bar \nu^Z = \bar \nu^X \otimes \bar \nu^Y \in \P(\R^{2d})$ is the stationary distribution for the base mean-field dynamics~\eqref{Eq:MFSystem}.
We define the tensorized stationary mean-field distribution:
\begin{align}\label{Eq:TensorMF}
    \bar \nu^{\bZ} := (\bar \nu^X)^{\otimes N} \otimes (\bar \nu^Y)^{\otimes N} \in \P(\R^{2dN}).
\end{align}

We can show the biased convergence guarantee of the finite-particle dynamics~\eqref{Eq:ParticleSystem3} to the tensorized mean-field stationary distribution~\eqref{Eq:TensorMF} in Theorem~\ref{Thm:ConvergenceDynToMF}.
A key observation is that the drift term in the finite-particle dynamics~\eqref{Eq:ParticleSystem3} is dissipative (see Lemma~\ref{Lem:LipschitzMonotone_bz} in Section~\ref{Sec:PropertiesVectorFields}); this allows us to perform a synchronous coupling analysis against the stationary mean-field dynamics to show a biased convergence guarantee in $W_2$ distance, which we leverage to show a biased convergence guarantee in KL divergence via a time derivative calculation along simultaneous diffusion processes.
A careful calculation shows that we can control the bias without dependence on the number of particles $N$ (see Lemma~\ref{Lem:ComparisonArbitrary} in Section~\ref{Sec:ComparisonVectorFields}).
We provide the proof of Theorem~\ref{Thm:ConvergenceDynToMF} in Section~\ref{Sec:ConvergenceDynToMFProof}.

\begin{theorem}\label{Thm:ConvergenceDynToMF}
    Assume Assumption~\ref{As:SCSmooth}.
    Suppose $\bZ_t \sim \rho_t^{\bZ}$ evolves following the finite-particle dynamics~\eqref{Eq:ParticleSystem3} in $\R^{2dN}$ from $\bar Z_0 \sim \rho_0^{\bZ} \in \P(\R^{2dN})$.
    Then for all $t \ge 0$:
    \begin{align*}
        W_2(\rho_t^{\bZ}, \bar \nu^{\bZ})^2 &\le e^{-\frac{3}{2}\alpha t} \, W_2(\rho_0^{\bZ}, \bar \nu^{\bZ})^2 + \frac{2L^2}{\alpha^2} \, \Var_{\bar \nu^{Z}}(\bar Z) \\
        \KL(\rho_t^{\bZ} \,\|\, \bar \nu^{\bZ}) &\le e^{-\alpha t} \left( \KL\left(\rho_0^{\bZ} \,\|\, \bar \nu^{\bZ} \right) + 
        \frac{2L^2}{\alpha \reg} \, W_2(\rho_0^{\bZ}, \bar \nu^{\bZ})^2\right) + \frac{4L^4}{\alpha^3 \reg} \, \Var_{\bar \nu^{Z}}(\bar Z).
    \end{align*}
\end{theorem}

By taking $t \to \infty$, we obtain the following estimates on the bias of the stationary distribution $\rho_\infty^{\bZ}$ of the finite-particle dynamics~\eqref{Eq:ParticleSystem3} from the tensorized stationary mean-field distribution $\bar \nu^{\bZ}$:
\begin{align*}
    W_2(\rho_\infty^{\bZ}, \bar \nu^{\bZ})^2 &\le \frac{2L^2}{\alpha^2} \, \Var_{\bar \nu^Z}(\bar Z), \\
    % \qquad\qquad
    \KL(\rho_\infty^{\bZ} \,\|\, \bar \nu^{\bZ}) &\le \frac{4L^4}{\alpha^3 \reg} \, \Var_{\bar \nu^Z}(\bar Z).
\end{align*}
We can bound the variance by $\Var_{\bar \nu^Z}(\bar Z) \le (2\reg d)/\alpha$, see Lemma~\ref{Lem:BoundVar}.  Therefore, note that the bias terms above do not scale with the number of particles $N$.
When we consider the average particle below, this implies a bias of order $O(1/N)$, see Corollary~\ref{Cor:AvgParticleSystem}.

We can also show the unbiased exponential convergence guarantees of the iterate $\rho_t^{\bZ}$ of the finite-particle dynamics~\eqref{Eq:ParticleSystem3} to its stationary distribution $\rho_\infty^{\bZ}$, see Theorem~\ref{Thm:DynamicsToBiasedLimit} in Section~\ref{Sec:DynamicsToBiasedLimit}.

%%%%%%%%%%%%%%%%%%%
\subsection{Biased Convergence of the Average Particle of the Finite-Particle Dynamics}

Suppose $\bZ_t = (\bX_t,\bY_t) = (X_t^1,\dots,X_t^N,Y_t^1,\dots,Y_t^N) \sim \rho_t^{\bZ}$ evolves following the finite-particle dynamics~\eqref{Eq:ParticleSystem3} in $\R^{2dN}$.
We define the \textit{average particle} to be the random variable $Z_t^I = (X_t^I,Y_t^I) \in \R^{2d}$, where $I \sim \Unif([N])$ is a uniformly chosen random index.
Note the distribution of the average particle $Z_t^I \sim \rho_t^{Z,\avg}$ (including the randomization over indices) is:
$$\rho_t^{Z,\avg} = \frac{1}{N} \sum_{i \in [N]} \rho_t^{Z,i}$$
where $\rho_t^{Z,i}$ is the marginal distribution of $Z_t^i = (X_t^i,Y_t^i) \in \R^{2d}$ from the joint vector $\bZ_t \sim \rho_t^{\bZ}$.

We have the following biased convergence guarantee of the average particle along the finite-particle dynamics~\eqref{Eq:ParticleSystem3}.
We provide the proof of Corollary~\ref{Cor:AvgParticleSystem} in Section~\ref{Sec:AvgParticleSystemProof}.

\begin{corollary}\label{Cor:AvgParticleSystem}
    Assume Assumption~\ref{As:SCSmooth}.
    Suppose $\bZ_t \sim \rho_t^{\bZ}$ evolves following the finite-particle dynamics~\eqref{Eq:ParticleSystem3} in $\R^{2dN}$, and let $Z_t^I \sim \rho_t^{Z,\avg}$ be the average particle as defined above.
    For all $t \ge 0$:
    \begin{align*}
        \KL(\rho_t^{Z,\avg} \,\|\, \bar \nu^{Z}) \le \frac{1}{N} e^{-\alpha t} \left( \KL\left(\rho_0^{\bZ} \,\|\, \bar \nu^{\bZ} \right) + 
        \frac{2L^2}{\alpha \reg} \, W_2(\rho_0^{\bZ}, \bar \nu^{\bZ})^2\right) + \frac{4L^4}{\alpha^3 \reg N} \, \Var_{\bar \nu^{Z}}(\bar Z).
    \end{align*}
\end{corollary}

By taking $t \to \infty$, we see that the limiting distribution $\rho_\infty^{Z,\avg}$ of the average particle satisfies:
\begin{align*}
    \KL(\rho_\infty^{Z,\avg} \,\|\, \bar \nu^{Z}) \le \frac{4L^4}{\alpha^3 \reg N} \, \Var_{\bar \nu^{Z}}(\bar Z).
\end{align*}
This shows that as $N \to \infty$ and $t \to \infty$, the average particle of the finite-particle dynamics~\eqref{Eq:ParticleSystem3} indeed converges to the stationary mean-field distribution $\bar \nu^Z$, which is the equilibrium distribution for the game~\eqref{Eq:Game} that we wish to compute.
However, this is still in the idealized continuous-time setting.
To obtain a meaningful practical guarantee, we study the discrete-time algorithm in Section~\ref{Sec:Algorithm}.

%%%%%%%%%%%%%%%%%%%
\section{Analysis of the Discrete-Time Finite-Particle Min-Max Langevin Algorithm}
\label{Sec:Algorithm}

We study a time discretization of the finite-particle dynamics~\eqref{Eq:ParticleSystem} into the discrete-time algorithm~\eqref{Eq:ParticleAlgorithm}.
We can write the algorithm~\eqref{Eq:ParticleAlgorithm} in terms of the joint vectors $\bx_k = (x_k^1,\dots,x_k^N) \in \R^{dN}$ and $\by_k = (y_k^1,\dots,y_k^N) \in \R^{dN}$ which follow the update:
\begin{subequations}\label{Eq:ParticleAlgorithm2}
    \begin{align}
        \bx_{k+1} &= \bx_k + \eta b^{\bX}(\bx_k,\by_k) + \sqrt{2\reg\eta} \, \bm{\zeta}_k^{\bx} \\
        \by_{k+1} &= \by_k + \eta b^{\bY}(\bx_k,\by_k) + \sqrt{2\reg\eta} \, \bm{\zeta}_k^{\by}
    \end{align}
\end{subequations}
where $\eta > 0$ is step size, $\bm{\zeta}_k^{\bx}, \bm{\zeta}_k^{\by} \sim \N(0,I)$ are independent standard Gaussian random variables in $\R^{dN}$, and where we have used the same vector fields $b^{\bX}$ and $b^{\bY}$ as defined in~\eqref{Eq:Defbxby}.

We can further write the algorithm~\eqref{Eq:ParticleAlgorithm2} in terms of the joint random variable $\bz_k = (\bx_k,\by_k) \in \R^{2dN}$ which follows the update:
\begin{align}\label{Eq:ParticleAlgorithm3}
    \bz_{k+1} &= \bz_k + \eta b^{\bZ}(\bz_k) + \sqrt{2\reg\eta} \, \bm{\zeta}_k^{\bz}
\end{align}
where $\eta > 0$ is step size, $\bm{\zeta}_k^{\bz} \sim \N(0,I)$ is an independent standard Gaussian random variable in $\R^{2dN}$, and where we have used the same vector field $b^{\bZ}$ as defined in~\eqref{Eq:Defbz}.

%%%%%%%%%%%%%%%%%%%
\subsection{Biased Convergence of Finite-Particle Algorithm to Stationary Mean-Field Distribution}

We can prove the following biased convergence guarantees of the finite-particle algorithm~\eqref{Eq:ParticleAlgorithm3} to the tensorized stationary mean-field distribution $\bar \nu^{\bZ} \in \P(\R^{2dN})$ defined in~\eqref{Eq:TensorMF}.
Our technique is to show a one-step biased contraction guarantee in $W_2$ distance along the discrete-time algorithm~\eqref{Eq:ParticleAlgorithm3}, using a continuous-time interpolation of each step of the algorithm~\eqref{Eq:ParticleAlgorithm3} as the solution to a stochastic process, and performing a synchronous coupling analysis against the stationary mean-field dynamics; see Lemma~\ref{Lem:OneStepW2} in Section~\ref{Sec:OneStepRecurrenceW2}.
We leverage this to show a one-step biased contraction guarantee in KL divergence along the algorithm~\eqref{Eq:ParticleAlgorithm3}, see Lemma~\ref{Lem:OneStepKL} in Section~\ref{Sec:OneStepRecurrenceKL}, and solve the recursions to conclude the biased convergence guarantees stated in Theorem~\ref{Thm:ConvergenceAlgToMF}.
We provide the proof of Theorem~\ref{Thm:ConvergenceAlgToMF} in Section~\ref{Sec:ConvergenceAlgToMFProof}.

Note that the bias terms in Theorem~\ref{Thm:ConvergenceAlgToMF} below match the continuous-time bias from Theorem~\ref{Thm:ConvergenceDynToMF}, with an additional bias term that scales with the step size $\eta$ as well as the number of particles $N$, which comes from the dimension of the space $\R^{2dN}$ where the algorithm operates.
This additional bias term that scales with $\eta$ is consistent with the analysis of simple discretization of the Langevin dynamics such as the Unadjusted Langevin Algorithm, see e.g.,~\citep{VW19}. 

\begin{theorem}\label{Thm:ConvergenceAlgToMF}
    Assume Assumption~\ref{As:SCSmooth}.
    Suppose $\bz_k \sim \rho_k^{\bz,\eta}$ evolves via the finite-particle algorithm~\eqref{Eq:ParticleAlgorithm3} in $\R^{2dN}$ with step size $0 < \eta \le \frac{\alpha}{64L^2}$ from $\bz_0 \sim \rho_0^{\bz,\eta} \in \P(\R^{2dN})$.
    Then for all $k \ge 0$:
    \begin{align*}
        W_2(\rho_k^{\bz,\eta}, \bar \nu^{\bZ})^2 &\le e^{-\frac{3}{2}\alpha \eta k} \, W_2(\rho_0^{\bz,\eta}, \bar \nu^{\bZ})^2 + \frac{8L^2}{\alpha^2}  \left(\Var_{\bar \nu^Z}(\bar Z) + 64 \, \reg \eta d N \right) \\
        \KL(\rho_k^{\bz,\eta} \,\|\, \bar \nu^{\bZ}) 
        &\le e^{-\alpha \eta k} \left(\KL(\rho_0^{\bz,\eta} \,\|\, \bar \nu^{\bZ}) + \frac{9 L^2}{\alpha \reg} W_2(\rho_0^{\bz,\eta}, \bar \nu^{\bZ})^2 \right) + \frac{45 L^4}{\alpha^3 \reg} \left(\Var_{\bar \nu^{Z}}(\bar Z) 
        + 55 \, \eta \reg d N \right).
    \end{align*}        
\end{theorem}

As the finite-particle algorithm~\eqref{Eq:ParticleAlgorithm3} is practically implementable, this provides a concrete algorithm to approximately compute the stationary mean-field distribution, with an explicit iteration complexity which we characterize in Corollary~\ref{Cor:ComplexityAlgToMF} below.

We can also show the unbiased exponential convergence guarantees of the iterate $\rho_k^{\bz,\eta}$ of the finite-particle algorithm~\eqref{Eq:ParticleAlgorithm3} to its stationary distribution $\rho_\infty^{\bz,\eta}$, see Theorem~\ref{Thm:AlgorithmToBiasedLimit} in Section~\ref{Sec:AlgorithmToBiasedLimit}.

%%%%%%%%%%%%%%%%%%%
\subsection{Biased Convergence of the Average Particle of the Finite-Particle Algorithm}

Suppose $\bz_k = (\bx_k,\by_k) = (x_k^1,\dots,x_k^N,y_k^1,\dots,y_k^N) \sim \rho_k^{\bz,\eta}$ evolves following the finite-particle algorithm~\eqref{Eq:ParticleAlgorithm3} in $\R^{2dN}$.
As before, we define the \textit{average particle} to be the random variable $z_k^I = (x_k^I,y_k^I) \in \R^{2d}$, where $I \sim \Unif([N])$ is a uniformly chosen random index.
The distribution of the average particle $z_k^I \sim \rho_k^{z,\eta,\avg}$ (including the randomization over indices) is:
$$\rho_k^{z,\eta,\avg} = \frac{1}{N} \sum_{i \in [N]} \rho_k^{z,\eta,i}$$
where $\rho_k^{z,\eta,i}$ is the marginal distribution of $z_k^i = (x_k^i,y_k^i) \in \R^{2d}$ from the joint vector $\bz_k \sim \rho_k^{\bz,\eta}$.

We have the following biased convergence guarantee of the average particle along the finite-particle algorithm~\eqref{Eq:ParticleAlgorithm3}.
We provide the proof of Corollary~\ref{Cor:AvgParticleAlgorithm} in Section~\ref{Sec:AvgParticleAlgorithmProof}.
\begin{corollary}\label{Cor:AvgParticleAlgorithm}
    Assume Assumption~\ref{As:SCSmooth}.
    Suppose $\bz_k \sim \rho_k^{\bz,\eta}$ evolves following the finite-particle algorithm~\eqref{Eq:ParticleAlgorithm3} in $\R^{2dN}$ with step size $0 < \eta \le \frac{\alpha}{64L^2}$, and let $z_k^I \sim \rho_k^{z,\eta,\avg}$ be the average particle as defined above.
    Then for all $k \ge 0$:
    \begin{align*}
        \KL(\rho_k^{z,\eta,\avg} \,\|\, \bar \nu^{Z}) \le \frac{e^{-\alpha \eta k}}{N} \left(\KL(\rho_0^{\bz,\eta} \,\|\, \bar \nu^{\bZ}) + \frac{9 L^2}{\alpha \reg} W_2(\rho_0^{\bz,\eta}, \bar \nu^{\bZ})^2 \right) 
        + \frac{45 L^4}{\alpha^3 \reg N} \Var_{\bar \nu^{Z}}(\bar Z)
        + 2475 \frac{\eta d L^4}{\alpha^3}.
    \end{align*}
\end{corollary}

By taking $k \to \infty$, we see that the limiting distribution $\rho_\infty^{z,\eta,\avg}$ of the average particle satisfies:
\begin{align*}
    \KL(\rho_\infty^{z,\eta,\avg} \,\|\, \bar \nu^{Z}) \le 45 \frac{L^4}{\alpha^3 \reg N} \Var_{\bar \nu^{Z}}(\bar Z)  
    + 2475 \frac{\eta d L^4}{\alpha^3}.
\end{align*}
This shows that the average particle of the $N$-particle algorithm approximately converges to the stationary mean-field distribution $\bar \nu^Z$, up to a bias term which scales with step size $\eta$ and inversely with the number of particles $1/N$.
Thus, we can make the bias arbitrarily small by choosing a sufficiently large number of particles $N$ and a sufficiently small step size $\eta$. 

The result above implies the following iteration complexity guarantee for the average particle of the finite-particle algorithm~\eqref{Eq:ParticleAlgorithm3} to reach an arbitrary precision in KL divergence to the stationary mean-field distribution $\bar \nu^Z = \bar \nu^X \otimes \bar \nu^Y$, which is the equilibrium distribution for the game~\eqref{Eq:Game} that we wish to compute.
We provide the proof of Corollary~\ref{Cor:ComplexityAlgToMF} in Section~\ref{Sec:ComplexityAlgToMFProof}.

\begin{corollary}\label{Cor:ComplexityAlgToMF}
    Assume Assumption~\ref{As:SCSmooth}.
    For any regularization parameter $\reg > 0$, and given any small error threshold $\error > 0$, suppose we do the following:
    \begin{enumerate}
        \item Run the min-max gradient descent algorithm~\eqref{Eq:MinMaxGD} from $\tilde z_0 = (0,0) \in \R^{2d}$ with step size $\eta_{\GD} = \frac{\alpha}{4L^2}$ for 
        $$k_{\GD} \ge \frac{4L^2}{\alpha^2} \log \frac{\alpha^3 \|z^*\|^2}{\reg dL^2}$$ 
        iterations, to obtain a final point $m^Z = \tilde z_{k_\GD} \in \R^{2d}$. 
        \item Define the Gaussian distribution $\gamma^Z = \N(m^Z, \frac{\reg}{L} I)$ on $\R^{2d}$, and initialize the algorithm~\eqref{Eq:ParticleAlgorithm3} from $\bz_0 = (z_0^{1}, \dots, z_0^{N}) \in \R^{2dN}$
        where $z_0^{1}, \dots, z_0^{N} \sim \gamma^Z$ are i.i.d., so $\bz_0 \sim \rho_0^{\bz,\eta} = (\gamma^Z)^{\otimes N}$.
        %
        % \item Run the algorithm~\eqref{Eq:ParticleAlgorithm3} with step size $\eta = \frac{\error \, \alpha^3}{7500 \, d L^4}$ and number of particles $N \ge \frac{270 \, dL^4}{\error \, \alpha^4}$.
        \item Run the finite-particle algorithm~\eqref{Eq:ParticleAlgorithm3} with step size and number of particles:
        $$\eta = \frac{\error \, \alpha^3}{7500 \, d L^4} \ , 
        \qquad\qquad
        N \ge \frac{270 \, dL^4}{\error \, \alpha^4}.$$
    \end{enumerate}
    Then the average particle $z_k^{I} \sim \rho_k^{z,\eta,\avg}$ of the algorithm~\eqref{Eq:ParticleAlgorithm3} satisfies $\KL(\rho_k^{z,\eta,\avg} \,\|\, \bar \nu^{Z}) \le \error$
    whenever the number of iterations $k$ satisfies:
    $$k \,\ge\, \frac{7500 \, d L^4}{\error \, \alpha^4} \log \frac{684 \, dL^6}{\error \, \alpha^6}.$$
\end{corollary}

%%%%%%%%%%%%%%
\section{Discussion}
\label{Sec:Discussion}

In this paper, we study zero-sum games in the space of probability distributions over $\R^d$ with entropy regularization and a base interaction function which is smooth and strongly convex-strongly concave.
We show the exponential convergence guarantee of the mean-field min-max Langevin dynamics to the equilibrium distribution in continuous time.
We also show the biased convergence guarantees for the finite-particle dynamics in continuous time, and the finite-particle algorithm in discrete time, to the equilibrium (stationary mean-field) distribution. 
We also provide an explicit iteration complexity for the average particle of the finite-particle algorithm to approximately compute the equilibrium distribution. 
We use standard tools from the analysis of stochastic processes and their time discretization, which have been used for analyzing sampling algorithms.

Our results answer a special case of the open question posed by~\citep{wang24open} in the simple setting of Euclidean case under strong convexity assumption.
There are many directions one can study 
toward the more general open question.
It would be interesting to study how to extend our results to weaken the strong convexity assumption, for example to allow weak convexity or local non-convexity of the interaction function, or only assuming isoperimetry of the equilibrium distribution.
In our analysis, the strong convexity assumption is crucial to show the exponential convergence guarantee of the mean-field dynamics and to provide biased convergence guarantees for the finite-particle systems 
in $W_2$ distance, which we leverage to obtain biased convergence guarantees in KL divergence.
It would also be interesting to study the more general setting of constrained domains or the manifold setting without convexity assumptions.
For constrained domains, we may have to add a reflection or projection step to the Langevin dynamics; or we can try to adapt the idea of the Proximal Sampler from sampling, see e.g.,~\citep{LST21,chen2022improved}.

%%%%%%%%
% For arXiv version
\acks{The authors thank Guillaume Wang, Daniel Lacker, and Manuel Arnese for valuable discussions.
A.W. thanks Tinsel W. for guidance. A.W. is supported by NSF Award CCF-2403391 and CCF-2443097. 
Y.C. is supported by the NSF Award CCF-2342642.
}

%%%%%%%%%
% \newpage
\addcontentsline{toc}{section}{References}
% \bibliography{refs.bib}
\bibliography{arxiv-min-max-v3.bbl}

%%%%%%%%%
\newpage
\appendix

%%%%%%%%%%%%%%
\setcounter{tocdepth}{2}
\tableofcontents

%%%%%%%%%%%
%%%%%%%%%%%%%%%%%%%
\section{Additional Related Work}
\label{Sec:Related}

Zero-sum games in the space of probability distributions have been studied in many recent works, including~\citep{hsieh2019finding,domingo2020mean,cen2021fast,wang2022exponentially,ma2022provably,lu2023two,kim2024symmetric,lascu2023entropic,lascu2024a,lascu2024mirror,ding2024papal,conger2024coupled,an2025convergence,lu2025convergencetimeaveragedmeanfield}.
Our main motivation in this work comes from the open problem by~\cite{wang24open}, who pose the question of studying the convergence guarantees of the simple mean-field min-max Langevin dynamics and its particle approximations for zero-sum games in the space of distributions, in particular without using two timescales or modifying the dynamics or algorithm.

The question posed by~\cite{wang24open} is for a general setting on a manifold, without convexity assumption on the interaction function.
In this work, we contribute an answer to this question for the simple case on the unconstrained Euclidean space with a strong convexity assumption on the base interaction function. 
From the perspective of optimization on the space of probability distributions under the Wasserstein metric, this makes the payoff function $\F_\reg$ in~\eqref{Eq:PayoffDef} a strongly convex-strongly concave function, so it is natural to expect that the min-max gradient flow in the space of distributions---which is exactly the mean-field min-max Langevin dynamics~\eqref{Eq:MFSystem}---to converge exponentially fast.
However, precise results in this simple setting taking into account particle approximation and time discretization seem to be previously unknown.
We establish these results in this paper to provide a baseline toward understanding the more general question posed by~\cite{wang24open}. 

We review a few works that are the most related to our work; see also~\citep{wang24open} and the references therein.
\cite{domingo2020mean} study the setting when the domains are compact Riemannian manifolds, and show a conditional convergence result that if the mean-field dynamics converges, then it converges to the equilibrium distribution of the game.
\cite{ma2022provably} and \cite{lu2023two} also study the compact manifold setting with two-timescale dynamics. 
\cite{ma2022provably} study the quasistatic regime where one player is always at optimality, and establish the asymptotic convergence of the mean-field dynamics to the equilibrium distribution; they also study the finite-particle and discrete-time approximation, without convergence analysis.
\cite{lu2023two} study the two-timescale dynamics with finite timescale separation, and establish the exponential convergence guarantees of the mean-field dynamics to the equilibrium distribution.
In the compact manifold setting, only smoothness assumptions on the interaction function are needed.

When the domains are Euclidean spaces, some convexity or integrability assumptions are needed.
\cite{kim2024symmetric} consider a more general setting of zero-sum games with entropy regularization where the interaction functional on the space of distributions is convex-concave (whereas in our setting~\eqref{Eq:Game} the interaction functional is bilinear since it is an expectation over a base function);
they study a modified mean-field Langevin averaged-gradient dynamics where the drift term uses a time average of the gradients over the iterates, and show a continuous-time convergence rate, as well as a convergence analysis of the finite-particle discrete-time algorithm.
\cite{ding2024papal} consider the setting when the base interaction function is a bounded perturbation of a quadratic function;
they study a finite-particle discrete-time algorithm that implements the mirror-prox primal-dual algorithm in the space of distributions, which requires an inner loop running a sampling algorithm to implement each iteration, and show explicit convergence guarantees of the resulting algorithm.

The work of~\cite{conger2024coupled} study a very general setting of coupled Wasserstein gradient flows for min-max and cooperative games in the space of distributions, with payoff functionals which include additional interaction energy terms, which are not present in our setting.
They show convergence guarantees of the mean-field dynamics in continuous time under convexity assumptions, utilizing the convexity structure in the Wasserstein space of distributions.
A special setting of their result~\cite[Theorem~3.4]{conger2024coupled}, for zero-sum games with a bilinear interaction functional as in~\eqref{Eq:Game} with a strongly convex-strongly concave base interaction function, already shows the exponential convergence guarantee of the mean-field dynamics to the equilibrium distribution, as in our Theorem~\ref{Thm:ConvMF}, and thus they have answered the open question by~\cite{wang24open} for the mean-field setting with strongly convex interaction.
However,~\cite{conger2024coupled} focuses on the continuous-time mean-field analysis, and in this work we complement the results by considering finite-particle approximation and discrete-time analysis.

%%%%%%%%%%%
\section{Proofs for the Properties of Equilibrium Distribution and Duality Gap}

%%%%%%%%%%%
\subsection{Proof of Lemma~\ref{Lem:BestResponse} (Properties of the Best-Response Distribution)}
\label{Sec:BestResponseProof}

\begin{proof}
    We note the KL divergence between two distributions $\rho, \nu \in \P(\R^d)$ can be decomposed as:
    \begin{align*}
        \KL(\rho \,\|\, \nu)
        % &= \int_{\R^d} \rho \, \log \frac{\rho}{\nu} \, dx \\
        &= \int_{\R^d} \rho \, \log \rho \, dx - \int_{\R^d} \rho \, \log \nu \, dx 
        = -H(\rho) + \E_\rho[-\log \nu]
    \end{align*}
    where $H$ is the negative entropy functional.
    For all $\rho^X, \rho^Y \in \P(\R^d)$, from the definitions~\eqref{Eq:BestResponse} of the 
    distributions $\nu^X = \Phi^X(\rho^Y)$ and $\nu^Y = \Phi^Y(\rho^X)$, we have:
    \begin{subequations}\label{Eq:BestResponseExplicit}
    \begin{align}
        -\log \nu^X(x) &= \reg^{-1} \E_{\rho^Y}[V(x,Y)] + \log C^Y(\rho^Y), \\
        -\log \nu^Y(y) &= -\reg^{-1} \E_{\rho^X}[V(X,y)] + \log C^X(\rho^X)
    \end{align}
    \end{subequations}
    where $C^Y(\rho^Y)$ and $C^X(\rho^X)$ are the normalizing constants:
    \begin{subequations}\label{Eq:NormConst}     
    \begin{align}
        C^Y(\rho^Y) &:= \int_{\R^d} \exp\left(-\reg^{-1} \E_{\rho^Y}[V(x,Y)]\right) \, dx, \\
        C^X(\rho^X) &:= \int_{\R^d} \exp\left(\reg^{-1} \E_{\rho^X}[V(X,y)]\right) \, dy.
    \end{align}
    \end{subequations}
    Then from the definition~\eqref{Eq:PayoffDef} of the payoff function $\F_\reg$, we can write:
    \begin{align}
        \F_\reg(\rho^X, \rho^Y) 
        &= \E_{\rho^X \otimes \rho^Y}[V(X,Y)] - \reg H(\rho^X) + \reg H(\rho^Y) \notag \\
        &= \reg \left(- H(\rho^X) + \E_{\rho^X}\left[ \reg^{-1} \E_{\rho^Y}\left[V(X,Y)\right]\right] \right) + \reg H(\rho^Y) \notag \\
        &= \reg \left(- H(\rho^X) + \E_{\rho^X}\left[ -\log \nu^X \right] - \log C^Y(\rho^Y) \right) + \reg H(\rho^Y) \notag \\
        &= \reg \, \KL(\rho^X \,\|\, \nu^X) - \reg \log C^Y(\rho^Y) + \reg H(\rho^Y). \label{Eq:DGExp1}
    \end{align}
    Then as a function of $\rho^X$, we see that $\rho^X \mapsto \F_\reg(\rho^X, \rho^Y)$ is minimized when we set $\rho^X = \nu^X$:
    $$\nu^X = \arg\min_{\tilde \rho^X \in \P(\R^d)} \F_\reg(\tilde \rho^X, \rho^Y).$$
    Similarly, we can also write:
    \begin{align}\label{Eq:DGExp2}
        \F_\reg(\rho^X, \rho^Y) 
        &= -\reg \, \KL(\rho^Y \,\|\, \nu^Y) + \reg \log C^X(\rho^X) - \reg H(\rho^X).
    \end{align} 
    Then as a function of $\rho^Y$, we see that $\rho^Y \mapsto \F_\reg(\rho^X, \rho^Y)$ is maximized when we set $\rho^Y = \nu^Y$:
    $$\nu^Y = \arg\max_{\tilde \rho^Y \in \P(\R^d)} \F_\reg(\rho^X, \tilde \rho^Y).$$
    
    The above computation also gives us:
    \begin{align*}
        \min_{\tilde \rho^X \in \P(\R^d)} \F_\reg(\tilde \rho^X, \rho^Y)
        % &= \F_\reg(\nu^X, \rho^Y) 
        &= - \reg \log C^Y(\rho^Y) + \reg H(\rho^Y) = \F_\reg(\rho^X, \rho^Y) - \reg \, \KL(\rho^X \,\|\, \nu^X) \\
        \max_{\tilde \rho^Y \in \P(\R^d)} \F_\reg(\rho^X, \tilde \rho^Y)
        % &= \F_\reg(\rho^X, \nu^Y) 
        &= \; \reg \log C^X(\rho^X) - \reg H(\rho^X) \;= \F_\reg(\rho^X, \rho^Y) + \reg \, \KL(\rho^Y \,\|\, \nu^Y).
    \end{align*}
    Therefore, we can write the duality gap as:
    \begin{align*}
        \DG(\rho^X, \rho^Y) = 
        \max_{\tilde \rho^Y \in \P(\R^d)} \F_\reg(\rho^X, \tilde \rho^Y)
        -
        \min_{\tilde \rho^X \in \P(\R^d)} \F_\reg(\tilde \rho^X, \rho^Y)
        = \reg \left( \KL(\rho^X \,\|\, \nu^X) + \KL(\rho^Y \,\|\, \nu^Y) \right)
    \end{align*}
    as desired.
\end{proof}

%%%%%%%%%%%
\subsection{Proof of Lemma~\ref{Lem:Equilibrium} (Existence, Uniqueness of Equilibrium and Bound on Duality Gap)}
\label{Sec:EquilibriumProof}

\begin{proof}[Proof of Lemma~\ref{Lem:Equilibrium}]
\textbf{(1)~Existence of equilibrium distribution:} 
Suppose $\bar Z_t \sim \bar \rho_t^Z = \bar \rho_t^X \otimes \bar \rho_t^Y$ in $\R^{2d}$ evolves following the mean-field dynamics~\eqref{Eq:MFJoint} from $\bar Z_0 \sim \bar \rho_0^Z = \bar \rho_0^X \otimes \bar \rho_0^Y \in \P(\R^{2d})$, and let $\bar \nu_t^Z = \bar \nu_t^X \otimes \bar \nu_t^Y$ where $\bar \nu_t^X = \Phi^X(\bar \rho_t^Y)$ and $\bar \nu_t^Y = \Phi^Y(\bar \rho_t^X)$ are the best-response distributions.
We show in Lemma~\ref{Lem:BoundMomentParticleMF} that $\bar \rho_t^Z \in \P(\R^{2d})$ for all $t \ge 0$.
In particular, $\bar \rho_t^Z$ is also in $\P_2(\R^{2d})$, the space of probability distributions over $\R^{2d}$ with finite second moment (without requiring absolute continuity with respect to the Lebesgue measure).

Furthermore, we show in Theorem~\ref{Thm:ConvMF} that for all $t \ge 0$:
$$\FI(\bar \rho_t^Z \,\|\, \bar \nu_t^Z) \le e^{-2\alpha t} \, \FI(\bar \rho_0^Z \,\|\, \bar \nu_0^Z).$$
We claim this implies $(\bar \rho_t^Z)_{t \ge 0}$ is a (continuous-time) Cauchy sequence in $\P_2(\R^{2d})$ with the $W_2$ metric.
Since $\bar Z_t$ evolves following the dynamics~\eqref{Eq:MFJoint},
$\bar \rho_t^Z$ evolves following the Fokker-Planck equation: 
$\part{\bar \rho_t^Z}{t} = \reg \nabla \cdot \left( \bar \rho_t^Z \nabla \log \frac{\bar \rho_t^Z}{\bar \nu_t^Z} \right)$,
so its Wasserstein speed (norm of the velocity) is: 
\begin{align*}
    \left\|\part{\bar \rho_t^Z}{t} \right\|_{\bar \rho_t^Z} = \reg \, \E_{\bar \rho_t^Z}\left[\left\|\nabla \log \frac{\bar \rho_t^Z}{\bar \nu_t^Z}\right\|^2\right]^{1/2} = \reg \, \FI(\bar \rho_t^Z \,\|\, \bar \nu_t^Z)^{1/2}
    \le e^{-\alpha t} \, \reg \, \FI(\bar \rho_0^Z \,\|\, \bar \nu_0^Z)^{1/2}.
\end{align*}
Then using the definition of the Wasserstein distance as the shortest path length, we can estimate:
\begin{align*}
    W_2(\bar \rho_{T_0}^Z, \bar \rho_{T_1}^Z) 
    \,\le\, \int_{T_0}^{T_1} \left\|\part{\bar \rho_t^Z}{t} \right\|_{\bar \rho_t^Z} \, dt 
    \,\le\, \reg \, \FI(\bar \rho_0^Z \,\|\, \bar \nu_0^Z)^{1/2} \int_{T_0}^{T_1} e^{-\alpha t} \, dt  
    \,\le\, e^{-\alpha T_0} \, \frac{\reg}{\alpha} \, \FI(\bar \rho_0^Z \,\|\, \bar \nu_0^Z)^{1/2}.
\end{align*}
Therefore, for any $T_1 > T_0$, we have $W_2(\bar \rho_{T_0}^Z, \bar \rho_{T_1}^Z) \to 0$ as $T_0 \to \infty$.
This shows that $(\bar \rho_t^Z)_{t \ge 0}$ is a Cauchy sequence in $\P_2(\R^{2d})$ with the $W_2$ metric.
Since $\P_2(\R^{2d})$ with the $W_2$ metric is a complete metric space~\citep[Theorem~6.18]{villani2009optimal}, this implies $(\bar \rho_t^Z)_{t \ge 0}$ must converge to a limit: $\bar \rho_\infty^Z = \lim_{t \to \infty} \bar \rho_t^Z$ which is also in $\P_2(\R^{2d})$.
Since each $\bar \rho_t^Z = \bar \rho_t^X \otimes \bar \rho_t^Y$ is a product distribution, the limit $\bar \rho_\infty^Z = \bar \rho_\infty^X \otimes \bar \rho_\infty^Y$ must also be a product distribution.

Furthermore, since $\lim_{t \to \infty} \FI(\bar \rho_t^Z \,\|\, \bar \nu_t^Z) = 0$, the limiting distribution $\bar \rho_\infty^Z$ is a fixed point for the best-response map, i.e., $\bar \rho_\infty^Z = \bar \rho_\infty^X \otimes \bar \rho_\infty^Y$ satisfies:
$$\bar \rho_\infty^X = \Phi^X(\bar \rho_\infty^Y),
\qquad\qquad
\bar \rho_\infty^Y = \Phi^Y(\bar \rho_\infty^X).$$
Therefore, $\bar \rho_\infty^Z$ minimizes the duality gap: $\DG(\bar \rho_\infty^X, \bar \rho_\infty^Y) = 0$, and thus, $(\bar \rho_\infty^X, \bar \rho_\infty^Y)$ is an equilibrium distribution for the game~\eqref{Eq:Game}, as defined in~\eqref{Eq:Equilibrium}.
Furthermore, from the definition of the best-response maps~\eqref{Eq:BestResponse}, we see that $\bar \rho_\infty^X = \Phi^X(\rho_\infty^Y)$ is $(\alpha/\reg)$-SLC, since we assume $V(x,y)$ is $\alpha$-strongly convex in $x$.
Similarly, $\bar \rho_\infty^Y = \Phi^Y(\rho_\infty^X)$ is $(\alpha/\reg)$-SLC, since we assume $V(x,y)$ is $\alpha$-strongly concave in $y$.

\medskip
\noindent
\textbf{(2)~Bound on duality gap:}
Let $(\bar \nu^X, \bar \nu^Y) \in \P(\R^d) \times \P(\R^d)$ be an equilibrium distribution for the game~\eqref{Eq:Game}, which exists by the argument above (which we call $(\bar \rho_\infty^X, \bar \rho_\infty^Y)$ above).
Since $\bar \nu^X$ is $(\alpha/\reg)$-SLC, it also satisfies $(\alpha/\reg)$-Talagrand inequality, which means for all $\rho^X \in \P(\R^d)$:
$$\KL(\rho^X \,\|\, \bar \nu^X) \ge \frac{\reg}{2\alpha} \, W_2(\rho^X, \bar \nu^X)^2.$$
Similarly, since $\bar \nu^Y$ is $(\alpha/\reg)$-SLC, it satisfies $(\alpha/\reg)$-Talagrand inequality, so for all $\rho^Y \in \P(\R^d)$:
$$\KL(\rho^Y \,\|\, \bar \nu^Y) \ge \frac{\reg}{2\alpha} \, W_2(\rho^Y, \bar \nu^Y)^2.$$
Adding the two bounds above gives the first inequality claimed in~\eqref{Eq:DGBound}.

Next, since $(\bar \nu^X, \bar \nu^Y)$ is a fixed point of the best-response map, the identities from~\eqref{Eq:BestResponseExplicit} become:
\begin{align*}
        -\log \bar \nu^X(x) &= \reg^{-1} \E_{\bar \nu^Y}[V(x,Y)] + \log C^Y(\bar \nu^Y), \\
        -\log \bar \nu^Y(y) &= -\reg^{-1} \E_{\bar \nu^X}[V(X,y)] + \log C^X(\bar \nu^X)   
\end{align*}
where $C^X, C^Y$ are the normalizing constants defined in~\eqref{Eq:NormConst}.
Then we can compute:
\begin{align*}
    H(\bar \nu^X) 
    = \E_{\bar \nu^X}\left[-\log \bar \nu^X\right]
    &= \reg^{-1} \E_{\bar \nu^X \otimes \bar \nu^Y}[V(X,Y)] + \log C^Y(\bar \nu^Y) \\
    H(\bar \nu^Y) 
    = \E_{\bar \nu^Y}\left[-\log \bar \nu^Y\right]
    &= -\reg^{-1} \E_{\bar \nu^X \otimes \bar \nu^Y}[V(X,Y)] + \log C^X(\bar \nu^X).
\end{align*}
Adding the two equations, we obtain:
\begin{align}\label{Eq:EqEntIdentity}
    H(\bar \nu^X) + H(\bar \nu^Y) = \log C^Y(\bar \nu^Y) + \log C^X(\bar \nu^X).
\end{align}
Therefore, for all $\rho^X, \rho^Y \in \P(\R^d)$, from the definition of the duality gap, and using the relations~\eqref{Eq:DGExp1} and~\eqref{Eq:DGExp2}, the identities $\bar \nu^X = \Phi^X(\bar \nu^Y)$, $\bar \nu^Y = \Phi^Y(\bar \nu^X)$, and the identity~\eqref{Eq:EqEntIdentity}, we have:
\begin{align*}
    \DG(\rho^X, \rho^Y)
    &= 
    \max_{\tilde \rho^Y \in \P(\R^d)} \F_\reg(\rho^X, \tilde \rho^Y)
    - \min_{\tilde \rho^X \in \P(\R^d)} \F_\reg(\tilde \rho^X, \rho^Y) \\
    &\ge \F_\reg(\rho^X, \bar \nu^Y) - \F_\reg(\bar \nu^X, \rho^Y) \\
    &= \reg \, \KL(\rho^X \,\|\, \bar \nu^X) - \reg \log C^Y(\bar \nu^Y) + \reg H(\bar \nu^Y) \\
    &\qquad - \left(-\reg \, \KL(\rho^Y \,\|\, \bar \nu^Y) + \reg \log C^X(\bar \nu^X) - \reg H(\bar \nu^X)\right) \\
    &= \reg \, \KL(\rho^X \,\|\, \bar \nu^X) + \reg \, \KL(\rho^Y \,\|\, \bar \nu^Y)
\end{align*}
which is the second inequality claimed in~\eqref{Eq:DGBound}.

\medskip
\noindent
\textbf{(3)~Uniqueness of equilibrium distribution:}
Suppose the contrary that we have two equilibrium distributions $(\bar \nu^X, \bar \nu^Y)$ and $(\bar \mu^X, \bar \mu^Y)$ of the game~\eqref{Eq:Game}.
Then by using the second inequality in the bound~\eqref{Eq:DGBound} with $(\rho^X, \rho^Y) = (\bar \mu^X, \bar \mu^Y)$, we have:
\begin{align*}
    \reg \, \KL(\bar \mu^X \,\|\, \bar \nu^X) + \reg \, \KL(\bar \mu^Y \,\|\, \bar \nu^Y)
    \le \DG(\bar \mu^X, \bar \mu^Y)
    = 0
\end{align*}
where the last equality follows because $(\bar \mu^X, \bar \mu^Y)$ is an equilibrium distribution so it minimizes the duality gap.
Therefore, we must have $\KL(\bar \mu^X \,\|\, \bar \nu^X) = 0$, so $\bar \mu^X = \bar \nu^X$.
Similarly, we must have $\KL(\bar \mu^Y \,\|\, \bar \nu^Y) = 0$, so $\bar \mu^Y = \bar \nu^Y$.
Therefore, the equilibrium distribution of the game~\eqref{Eq:Game} is unique.
\end{proof}

%%%%%%%%%%%
\section{Proofs for the Mean-Field Min-Max Langevin Dynamics}
\label{Sec:MFAppendix}

%%%%%%%%%%%%%%%
\subsection{Derivation of the Mean-Field Dynamics}
\label{Sec:MFDerivation}

At an idealized continuous-time level, a natural strategy to compute the equilibrium distribution of the game~\eqref{Eq:Game} is for each player to run the gradient flow dynamics in the space of probability distributions to minimize their own objective function.
Suppose we endow the space of probability distribution $\P(\R^d)$ with the Wasserstein $W_2$ metric.
Then each player maintains a continuous-time curve of probability distributions $(\bar \rho^X_t)_{t \ge 0}$ and $(\bar \rho^Y_t)_{t \ge 0}$, which they evolve via:
\begin{align*}
    \part{\bar \rho^X_t}{t} &= -\grad_{\rho^X} \, \F_\reg( \bar \rho^X_t, \bar \rho^Y_t),
    \qquad\qquad
    \part{\bar \rho^Y_t}{t} = \grad_{\rho^Y} \, \F_\reg(\bar  \rho^X_t, \bar \rho^Y_t).
\end{align*}
In the above, $\grad_{\rho^X} \F_\reg(\rho^X,\rho^Y)$ denotes the Wasserstein gradient of $\F_\reg(\rho^X,\rho^Y)$ with respect to the first argument $\rho^X$, with the second argument $\rho^Y$ fixed.
Similarly, $\grad_{\rho^Y} \F_\reg(\rho^X,\rho^Y)$ denotes the Wasserstein gradient with respect to the second argument $\rho^Y$, with the first argument $\rho^X$ fixed. 

Following the computation rule for the Wasserstein gradient~\citep{villani2009optimal}, the above system of gradient flow dynamics corresponds to the following system of Fokker-Planck equations:
\begin{align*}
    \part{\bar \rho^X_t}{t} &= \nabla \cdot \left( \bar \rho^X_t \, \E_{\bar \rho^Y_t}[\nabla_x V(\cdot, \bar Y_t)] \right) + \reg \Delta \bar \rho^X_t \\
    \part{\bar \rho^Y_t}{t} &= -\nabla \cdot \left( \bar \rho^Y_t \, \E_{\bar \rho^X_t}[\nabla_y V(\bar X_t, \cdot)] \right) + \reg \Delta \bar \rho^Y_t.
\end{align*}
In the above, $\nabla_x V(x,y)$ and $\nabla_y V(x,y)$ denote the gradient with respect to the first and second argument, respectively, while keeping the other argument fixed.
Here $\nabla \cdot$ is the divergence (trace of the Jacobian) of a vector field, and $\Delta$ is the Laplacian (trace of the Hessian) of a function.

The above system of Fokker-Planck equations can be realized as the continuity equations for a pair of stochastic processes $(\bar X_t)_{t \ge 0}$ and $(\bar Y_t)_{t \ge 0}$ in $\R^d$ which evolve following the \textit{mean-field min-max Langevin dynamics}:
\begin{align*}
    d\bar X_t &= -\E_{\bar \rho_t^Y}[\nabla_x V(\bar X_t, \bar Y_t)] \, dt + \sqrt{2\reg} \, dW_t^X \\
    d\bar Y_t &= \E_{\bar \rho_t^X}[\nabla_y V(\bar X_t, \bar Y_t)] \, dt + \sqrt{2\reg} \, dW_t^Y
\end{align*}
where $(W_t^X)_{t \ge 0}$ and $(W_t^Y)_{t \ge 0}$ are independent standard Brownian motion in $\R^d$.
The above is a mean-field system because the evolution of $\bar X_t$ depends on the distribution $\bar \rho_t^Y$ of the other player, while the evolution of $\bar Y_t$ depends on the distribution $\bar \rho_t^X$.
However, note the dependence is only via their expectations, so in the mean-field system above, $\bar X_t$ and $\bar Y_t$ evolve independently, i.e., they do not interact at the level of random variables, only at the level of distributions.
In particular, it does not matter that we use independent Brownian motions for $\bar X_t$ and $\bar Y_t$, we could also use the same Brownian motion; we leave it as the above for convenience when writing the dynamics in the joint form in Section~\ref{Sec:MF}.

%%%%%%%%%%%%
\subsection{Time Derivative of Relative Fisher Information along Mean-Field Langevin Dynamics}
\label{Sec:ddtFI_MF}

We show the following identity on the time derivative of the relative Fisher information between the iterate of the mean-field min-max Langevin dynamics and its best-response distribution.
In Lemma~\ref{Lem:ddtFI_MF} below, the first term on the right-hand side is the second-order relative Fisher information between $\bar \rho_t^Z$ and $\bar \nu_t^Z$, and the second term is a weighted relative Fisher information.

We note the identity in Lemma~\ref{Lem:ddtFI_MF} is formally identical to the time derivative identity of the relative Fisher information $\FI(\rho_t \,\|\, \nu)$ (or the second derivative of the KL divergence $\KL(\rho_t \,\|\, \nu)$) when $\rho_t$ evolves along the Langevin dynamics to the target distribution $\nu$, using the Otto calculus formula for the Hessian of KL divergence; see~\citep[Formula~15.7]{villani2009optimal} or~\citep[Eq.~(10) in Appendix~A]{WJ18b}.
Since the mean-field min-max Langevin dynamics is the min-max Wasserstein gradient flow in the space of distributions, the identity in Lemma~\ref{Lem:ddtFI_MF} can also be formally derived via Otto calculus computation.
We provide a proof of Lemma~\ref{Lem:ddtFI_MF} via an explicit computation below.
See~\citep[Lemma~7.2]{conger2024coupled} for an alternative proof for a similar result, and see also the proof of Theorem~\ref{Thm:DetMinMaxGF} part (4) in Section~\ref{Sec:ReviewDetGF} for the finite-dimensional version of this identity.

We denote $\|u\|^2_A := u^\top A u$ and $\|A\|^2_{\HS} := \Tr(A^\top A)$ for any $u \in \R^{2d}$ and $A \in \R^{2d \times 2d}$.

\begin{lemma}\label{Lem:ddtFI_MF}
    Suppose $\bar Z_t = (\bar X_t,\bar Y_t) \sim \bar \rho_t^Z = \bar \rho_t^X \otimes \bar \rho_t^Y$ evolves following the mean-field min-max Langevin dynamics~\eqref{Eq:MFJoint} in the joint space $\R^{2d}$.
    Define its best-response distribution $\bar \nu_t^Z = \bar \nu_t^X \otimes \bar \nu_t^Y$ where $\bar \nu_t^X = \Phi^X(\bar \rho_t^Y)$ and $\bar \nu_t^Y = \Phi^Y(\bar \rho_t^X)$ as defined in~\eqref{Eq:BestResponse}.
    Then for all $t \ge 0$, we have the identity:
    \begin{align*}
        \frac{d}{dt} \FI(\bar \rho_t^Z \,\|\, \bar \nu_t^Z)
        &= -2\reg \, \E_{\bar \rho_t^Z}\left[\left\|\nabla^2 \log \frac{\bar \rho_t^Z}{\bar \nu_t^Z} \right\|^2_{\HS} \right] - 2\reg \, \E_{\bar \rho_t^Z} \left[\left\|\nabla \log \frac{\bar \rho_t^Z}{\bar \nu_t^Z} \right\|^2_{(-\nabla^2 \log \bar \nu_t^Z)}\right].
    \end{align*}
\end{lemma}
\begin{proof}
    For ease of notation, in this proof we write $\rho_t^Z \equiv \bar \rho_t^Z$ and $\nu_t^Z \equiv \bar \nu_t^Z$, i.e., we drop the superscript bar notation (which denotes distributions along mean-field dynamics).
    Similarly, we write $\rho_t^X \equiv \bar \rho_t^X$, $\rho_t^Y \equiv \bar \rho_t^Y$, and $\nu_t^X \equiv \bar \nu_t^X$, $\nu_t^Y \equiv \bar \nu_t^Y$.
    
    We define $f_t = -\log \rho_t^Z$ and $g_t = -\log \nu_t^Z$, so $f_t, g_t \colon \R^{2d} \to \R$ are separable functions, in particular:
    $$g_t(z) = g_t(x,y) = g_t^X(x) + g_t^Y(y)$$
    where $g_t^X = -\log \nu_t^X$ and $g_t^Y = -\log \nu_t^Y$.
    Recall from the definition~\eqref{Eq:BestResponse} for $\nu_t^X = \Phi^X(\rho_t^Y)$ and $\nu_t^Y = \Phi^Y(\rho_t^X)$, we have:
    \begin{subequations}\label{Eq:BRScore}
    \begin{align}
        \nabla g_t^X(x) &= \reg^{-1} \E_{\rho_t^Y}[\nabla_x V(x, Y)] \label{Eq:BRScorex} \\
        \nabla g_t^Y(y) &= -\reg^{-1} \E_{\rho_t^X}[\nabla_y V(X, y)]. \label{Eq:BRScorey}
    \end{align}
    \end{subequations}
    This also implies:
    \begin{subequations}\label{Eq:BRLaplacian}
    \begin{align}
        \Delta g_t^X(x) &= \reg^{-1} \E_{\rho_t^Y}\left[\Delta_x V(x, Y)\right] \label{Eq:BRLaplacianx} \\
        \Delta g_t^Y(y) &= -\reg^{-1} \E_{\rho_t^X}\left[\Delta_y V(X, y)\right]. \label{Eq:BRLaplaciany}
    \end{align}        
    \end{subequations}
        
    The mean-field min-max Langevin dynamics~\eqref{Eq:MFJoint} can be written as the SDE:
    \begin{align*}
        d\bar Z_t 
        \,=\, -\reg \nabla g_t(\bar Z_t) \, dt + \sqrt{2\reg} \, dW_t
    \end{align*}
    Then $\rho_t^Z = \exp(-f_t)$ evolves following the Fokker-Planck equation:
    \begin{align}
        \part{\rho_t^Z}{t}
        &= \reg \nabla \cdot \left(\rho_t^Z \nabla g_t \right) + \reg \Delta \rho_t^Z. \label{Eq:FPrhot}
    \end{align}
    In particular, for the components $\rho_t^Z = \rho_t^X \otimes \rho_t^Y$, we also have:
    \begin{subequations}\label{Eq:FPrhotXY}
    \begin{align}
        \part{\rho_t^X}{t}
        &= \reg \nabla \cdot \left(\rho_t^X \nabla g_t^X \right) + \reg \Delta \rho_t^X \label{Eq:FPrhotX} \\
        \part{\rho_t^Y}{t}
        &= \reg \nabla \cdot \left(\rho_t^Y \nabla g_t^Y \right) + \reg \Delta \rho_t^Y. \label{Eq:FPrhotY}
    \end{align}
    \end{subequations}
    
    We will use the following relations for $f_t = -\log \rho_t^Z$: 
    \begin{subequations}
        \begin{align}
            \rho_t^Z \, \nabla f_t &= -\nabla \rho_t^Z \label{Eq:Rel1} \\
            \rho_t^Z \, \nabla^2 f_t &= -\nabla^2 \rho_t^Z + \rho_t^Z (\nabla f_t) (\nabla f_t)^\top  \label{Eq:Rel2} \\
            \Delta \rho_t^Z &= - \rho_t^Z \Delta f_t + \rho_t^Z \|\nabla f_t\|^2. \label{Eq:Rel3}
        \end{align}
    \end{subequations}
    Then we can compute that $f_t = -\log \rho_t^Z$ evolves following the equation:
    \begin{align}
        \partial_t f_t = -\frac{1}{\rho_t^Z} \part{\rho_t^Z}{t} 
        &=  -\frac{1}{\rho_t^Z} \left( \reg \langle \nabla \rho_t^Z, \nabla g_t \rangle +\reg \rho_t^Z \Delta g_t + \reg \Delta \rho_t^Z \right)
        \notag \\
        &= \reg \langle \nabla f_t, \nabla g_t \rangle - \reg \Delta g_t + \reg \Delta f_t - \reg \|\nabla f_t\|^2 \label{Eq:ft}
    \end{align}

    We will also use the \textit{Bochner's formula}, which states that for smooth $u \colon \R^{2d} \to \R$,
    \begin{align}\label{Eq:Bochner}
        \langle \nabla u, \nabla \Delta u \rangle = \frac{1}{2} \Delta \|\nabla u\|^2 - \|\nabla^2 u\|^2_{\HS}.
    \end{align}

    By integration by parts, we can write the relative Fisher information as:
    \begin{align}
        \FI(\rho_t^Z \,\|\, \nu_t^Z)
        \,&=\, \E_{\rho_t^Z}\left[\|\nabla f_t - \nabla g_t\|^2 \right] \notag \\
        &= \E_{\rho_t^Z}\left[\|\nabla f_t\|^2 + \|\nabla g_t\|^2 - 2 \langle \nabla f_t, \nabla g_t \rangle \right] \notag \\
        &= \blue{\E_{\rho_t^Z}\left[\|\nabla f_t\|^2 \right]}
        \red{\,+\, \E_{\rho_t^Z}\left[\|\nabla g_t\|^2\right]} 
        \purple{\,-\, 2 \E_{\rho_t^Z}\left[\Delta g_t \right]}. \label{Eq:FIt}
    \end{align}
    We compute the time derivative of the three terms above separately, with explicit computation below, color-coded for clarity for when we combine them.
    
    \paragraph{(I) First term:}
    The first term in~\eqref{Eq:FIt} is the Fisher information of $\rho_t^Z$.    
    We compute its time derivative using the Fokker-Planck equation~\eqref{Eq:FPrhot} and the formula~\eqref{Eq:ft} to get:
    \begin{subequations}
    \begin{align}
        \frac{d}{dt} \E_{\rho_t^Z}\left[\|\nabla f_t\|^2 \right] 
        &= \int_{\R^{2d}} (\partial_t \rho_t^Z) \,  \|\nabla f_t\|^2 \, dz + 2\int_{\R^{2d}} \rho_t^Z \, \langle \nabla f_t, \nabla \partial_t f_t \rangle \, dz \notag \\
        &= \reg \int_{\R^{2d}} \nabla \cdot (\rho_t^Z \nabla g_t) \,  \|\nabla f_t\|^2 \, dz + \reg \int_{\R^{2d}} (\Delta \rho_t^Z) \,  \|\nabla f_t\|^2 \, dz  \label{Eq:Calc1a}  \\
        &\qquad
        + 2\reg \int_{\R^{2d}} \rho_t^Z \, \langle \nabla f_t, \nabla (\langle \nabla f_t, \nabla g_t \rangle - \Delta g_t) \rangle \, dz \label{Eq:Calc1b} \\
        &\qquad 
        + 2\reg \int_{\R^{2d}} \rho_t^Z \, \langle \nabla f_t, \nabla (\Delta f_t) \rangle \, dz 
        - 2\reg \int_{\R^{2d}} \rho_t^Z \, \langle \nabla f_t, \nabla \|\nabla f_t\|^2 \rangle \, dz. \label{Eq:Calc1c} 
    \end{align}
    \end{subequations}
    We calculate the terms above one by one.
    \begin{enumerate}
        \item The first term in~\eqref{Eq:Calc1a} is, by integration by parts:
        \begin{align*}
            \reg\int_{\R^{2d}} \nabla \cdot (\rho_t^Z \nabla g_t) \,  \|\nabla f_t\|^2 \, dz
            &= -\reg \int_{\R^{2d}} \rho_t^Z \langle \nabla g_t, \, \nabla (\|\nabla f_t\|^2) \rangle \, dz  \\
            &= -2\reg\int_{\R^{2d}} \rho_t^Z \langle \nabla g_t, \, (\nabla^2 f_t) \, \nabla f_t \rangle \, dz.            
        \end{align*}
        \item The second term in~\eqref{Eq:Calc1a} is, by integration by parts:
        \begin{align*}
            \reg \int_{\R^{2d}} (\Delta \rho_t^Z) \,  \|\nabla f_t\|^2 \, dz 
            = \reg \int_{\R^{2d}} \rho_t^Z \, \Delta  \|\nabla f_t\|^2 \, dz.
        \end{align*}
        \item The term in~\eqref{Eq:Calc1b} is, by distributing the gradient,
        \begin{align*}
            &2\reg \int_{\R^{2d}} \rho_t^Z \, \langle \nabla f_t, \nabla ( \langle \nabla f_t, \nabla g_t \rangle - \Delta g_t) \rangle \, dz \\
            &= 2\reg \int_{\R^{2d}} \rho_t^Z \left( \langle \nabla f_t, (\nabla^2 f_t) \, \nabla g_t \rangle 
            + \langle \nabla f_t, (\nabla^2 g_t) \, \nabla f_t \rangle 
            - \langle \nabla f_t, \, \nabla \Delta g_t \rangle \right) dz.
        \end{align*}
        \item The first term in~\eqref{Eq:Calc1c} is, by Bochner's formula~\eqref{Eq:Bochner}:
        \begin{align*}
            2\reg \int_{\R^{2d}} \rho_t^Z \, \langle \nabla f_t, \nabla (\Delta f_t) \rangle \, dz
            &= \reg \int_{\R^{2d}} \rho_t^Z \, \Delta \|\nabla f_t\|^2 \, dz
            - 2\reg \int_{\R^{2d}} \rho_t^Z \, \|\nabla^2 f_t\|^2_{\HS} \, dz.
        \end{align*}
        \item The second term in~\eqref{Eq:Calc1c} is, using relation~\eqref{Eq:Rel2} and integration by parts:
        \begin{align*}
            - 2\reg \int_{\R^{2d}} \rho_t^Z \, \langle \nabla f_t, \nabla \|\nabla f_t\|^2 \rangle \, dz
            &= 2\reg \int_{\R^{2d}} \langle \nabla \rho_t^Z, \nabla \|\nabla f_t\|^2 \rangle \, dz 
            = - 2\reg \int_{\R^{2d}} \rho_t^Z \, \Delta \|\nabla f_t\|^2 dz.
        \end{align*}
    \end{enumerate}
    Combining the above, we see that the terms involving $\langle \nabla g_t, \, (\nabla^2 f_t) \, \nabla f_t \rangle$ and $\Delta \|\nabla f_t\|^2$ vanish, so we are left with:
    \begin{align}
        \blue{\frac{d}{dt} \E_{\rho_t^Z}\left[\|\nabla f_t\|^2 \right]
        = -2\reg \, \E_{\rho_t^Z}\left[\left\|\nabla^2 f_t \right\|^2_{\HS} \right] 
        +2\reg \E_{\rho_t^Z}\left[\langle \nabla f_t, (\nabla^2 g_t) \, \nabla f_t \rangle  \right] 
        - 2\reg \E_{\rho_t^Z}\left[\langle \nabla f_t, \, \nabla \Delta g_t \rangle  \right].} 
        \label{Eq:FirstTerm-1}
    \end{align}
    
    %%%%%%%%%%%
    \paragraph{(II) Second term:} For the second term in~\eqref{Eq:FIt}, using the Fokker-Planck equation~\eqref{Eq:FPrhot}, we have:
    \begin{subequations}        
    \begin{align}
        \frac{d}{dt} \E_{\rho_t^Z}\left[\|\nabla g_t\|^2 \right]
        &= \int_{\R^{2d}} (\partial_t \rho_t^Z) \|\nabla g_t\|^2 \, dz +  2 \int_{\R^{2d}} \rho_t^Z \, \langle \nabla g_t, \partial_t (\nabla g_t) \rangle \, dz \notag \\
        &= \reg \int_{\R^{2d}}  \nabla \cdot (\rho_t^Z \nabla g_t) \, \|\nabla g_t\|^2 \, dz \,+\, \reg \int_{\R^{2d}}  \Delta \rho_t^Z \, \|\nabla g_t\|^2 \, dz \label{Eq:Calc2a} \\
        &\qquad + 2 \int_{\R^{d}} \rho_t^X \, \langle \nabla g_t^X, \partial_t (\nabla g_t^X) \rangle \, dx
        + 2 \int_{\R^{d}} \rho_t^Y \, \langle \nabla g_t^Y, \partial_t (\nabla g_t^Y) \rangle \, dy \label{Eq:Calc2b}
    \end{align}
    \end{subequations}
    where in the second line above we have used the fact that $g_t(x,y) = g_t^X(x) + g_t^Y(y)$ is separable, and $\rho_t^Z = \rho_t^X \otimes \rho_t^Y$ is independent.
    We calculate the terms above one by one.
    \begin{enumerate}
        \item The first term in~\eqref{Eq:Calc2a} is, by integration by parts:
        \begin{align*}
            \reg \int_{\R^{2d}}  \nabla \cdot (\rho_t^Z \nabla g_t) \, \|\nabla g_t\|^2 \, dz
            &= -\reg \int_{\R^{2d}} \rho_t^Z \langle \nabla g_t, \nabla (\|\nabla g_t\|^2) \rangle \, dz \\
            &= -2\reg\int_{\R^{2d}} \rho_t^Z \langle \nabla g_t, (\nabla^2 g_t) \, \nabla g_t \rangle \, dz.
        \end{align*}
        \item The second term in~\eqref{Eq:Calc2a} is, by integration by parts and using relation~\eqref{Eq:Rel1}:
        \begin{align*}
            \reg \int_{\R^{2d}} \Delta \rho_t^Z \, \|\nabla g_t\|^2 \, dz
            &= \reg \int_{\R^{2d}} \rho_t^Z \langle \nabla f_t, \, \nabla \|\nabla g_t\|^2 \rangle \, dz \\
            &= 2\reg \int_{\R^{2d}} \rho_t^Z \, \langle \nabla f_t, \, (\nabla^2 g_t) \, \nabla g_t \rangle \, dz.
        \end{align*}
        \item For the first term in~\eqref{Eq:Calc2b}, we can compute using~\eqref{Eq:BRScorex} and~\eqref{Eq:FPrhotY}, and integration by parts:
        \begin{align*}
            \partial_t (\nabla g_t^X(x))
            &= \reg^{-1} \int_{\R^d} \partial_t \rho_t^Y(y) \nabla_x V(x,y) \, dy \\
            &= \int_{\R^d} \left(\nabla \cdot \left(\rho_t^Y \nabla g_t^Y \right)(y) + \Delta \rho_t^Y(y)\right) \nabla_x V(x,y) \, dy \\
            &= -\int_{\R^d} \rho_t^Y(y) \, \nabla^2_{yx} V(x,y) \, \nabla g_t^Y(y) \, dy + \int_{\R^d} \rho_t^Y(y) \Delta_y \nabla_x V(x,y) \, dy.
        \end{align*}
        Then, since $\rho_t^Z = \rho_t^X \otimes \rho_t^Y$, we can compute:
        \begin{align*}
            &2 \int_{\R^{d}} \rho_t^X \, \langle \nabla g_t^X, \partial_t (\nabla g_t^X) \rangle \, dx \\
            &= 2 \int_{\R^{d}} \rho_t^X(x) \, \left\langle \nabla g_t^X(x), \, \int_{\R^d} \rho_t^Y(y) \left(- \nabla^2_{yx} V(x,y) \, \nabla g_t^Y(y) + \Delta_y \nabla_x V(x,y) \right) dy \right \rangle \, dx \\
            &= 2 \int_{\R^{2d}} \rho_t^Z(x,y) \, \left\langle \nabla g_t^X(x), \, - \nabla^2_{yx} V(x,y) \, \nabla g_t^Y(y) + \Delta_y \nabla_x V(x,y) \right \rangle \, dx \,dy \\
            &= 2 \int_{\R^{2d}} \rho_t^Z \, \left\langle \nabla g_t^X, \, - (\nabla^2_{yx} V ) \nabla g_t^Y + \Delta_y \nabla_x V \right \rangle \, dz.
        \end{align*}
        \item For the second term in~\eqref{Eq:Calc2b}, we can compute using~\eqref{Eq:BRScorey} and~\eqref{Eq:FPrhotX}, and integration by parts:
        \begin{align*}
            \partial_t (\nabla g_t^Y(y))
            &= -\reg^{-1} \int_{\R^d} \partial_t \rho_t^X(x) \nabla_y V(x,y) \, dx \\
            &= -\int_{\R^d} \left(\nabla \cdot \left(\rho_t^X \nabla g_t^X \right)(x) + \Delta \rho_t^X(x)\right) \nabla_y V(x,y) \, dx \\
            &= \int_{\R^d} \rho_t^X(x) \, \nabla^2_{xy} V(x,y) \, \nabla g_t^X(x) \, dx - \int_{\R^d} \rho_t^X(x) \Delta_x \nabla_y V(x,y) \, dx.
        \end{align*}
        Then, since $\rho_t^Z = \rho_t^X \otimes \rho_t^Y$, we can compute:
        \begin{align*}
            &2 \int_{\R^{d}} \rho_t^Y \, \langle \nabla g_t^Y, \partial_t (\nabla g_t^Y) \rangle \, dy \\
            &= 2 \int_{\R^{d}} \rho_t^Y(y) \, \left\langle \nabla g_t^Y(y), \, \int_{\R^d} \rho_t^X(x) \left( \nabla^2_{xy} V(x,y) \, \nabla g_t^X(x) - \Delta_x \nabla_y V(x,y) \right) dx \right \rangle \, dy \\
            &= 2 \int_{\R^{2d}} \rho_t^Z(x,y) \, \left\langle \nabla g_t^Y(y), \, \nabla^2_{xy} V(x,y) \, \nabla g_t^X(x) - \Delta_x \nabla_y V(x,y) \right \rangle \, dx \,dy \\
            &= 2 \int_{\R^{2d}} \rho_t^Z \, \left\langle \nabla g_t^Y, \, (\nabla^2_{xy} V ) \nabla g_t^X - \Delta_x \nabla_y V \right \rangle \, dz.
        \end{align*}    
    \end{enumerate}
    Since $\nabla^2_{xy} V = (\nabla^2_{yx} V)^\top$, when we combine the calculations above, we see that the terms involving $\langle \nabla g_t^Y, \nabla^2_{xy} V \; \nabla g_t^X \rangle$ cancel, so we get:
    \begin{align}\label{Eq:SecondTerm-1}
        \red{\frac{d}{dt} \E_{\rho_t^Z}\left[\|\nabla g_t\|^2 \right]}
        &\red{\,= \E_{\rho_t^Z}\left[2\reg \, \langle \nabla f_t-\nabla g_t, (\nabla^2 g_t) \nabla g_t \rangle + 2\langle \nabla g_t^X, \Delta_y \nabla_x V \rangle - 2\langle \nabla g_t^Y, \Delta_x \nabla_y V \rangle \right].}
    \end{align}

    %%%%%%%%%%%
    \paragraph{(III) Third term:} For the third term in~\eqref{Eq:FIt}, using the Fokker-Planck equation~\eqref{Eq:FPrhot}, we have: 
    \begin{subequations}        
    \begin{align}
        \frac{d}{dt} \left(-2 \E_{\rho_t^Z}\left[\Delta g_t \right] \right)
        &= -2 \int_{\R^{2d}} (\partial_t \rho_t^Z) \, \Delta g_t \, dz - 2 \int_{\R^{2d}} \rho_t^Z \, \partial_t (\Delta g_t) \, dz \notag \\
        &= -2\reg \int_{\R^{2d}} \nabla \cdot (\rho_t^Z \nabla g_t) \, (\Delta g_t)  \, dz
        - 2\reg \int_{\R^{2d}} (\Delta \rho_t^Z) \, (\Delta g_t) \, dz \label{Eq:Calc3a} \\
        &\qquad -2 \int_{\R^{d}} \rho_t^X \, \partial_t (\Delta_x g_t^X) \, dx
        -2 \int_{\R^{d}} \rho_t^Y \, \partial_t (\Delta_y g_t^Y) \, dy \label{Eq:Calc3b}
    \end{align}
    \end{subequations}
    where in the second line above we have used the fact that $g_t(x,y) = g_t^X(x) + g_t^Y(y)$ is separable, and $\rho_t^Z = \rho_t^X \otimes \rho_t^Y$ is independent.
    We calculate the terms above one by one.
    \begin{enumerate}
        \item The first term in~\eqref{Eq:Calc3a} is, by integration by parts:
        \begin{align*}
            -2\reg \int_{\R^{2d}} \nabla \cdot (\rho_t^Z \nabla g_t) \, (\Delta g_t)  \, dz
            &= 2\reg \int_{\R^{2d}} \rho_t^Z \langle \nabla g_t, \, \nabla \Delta g_t \rangle \, dz.
        \end{align*}
        \item The second term in~\eqref{Eq:Calc3a} is, by integration by parts and denoting $z = (z_1,\dots,z_{2d})$:
        {\allowdisplaybreaks
        \begin{align*}
            -2\reg \int_{\R^{2d}} (\Delta \rho_t^Z) \, (\Delta g_t) \, dz 
            &= -2\reg \int_{\R^{2d}} \left(\sum_{i=1}^{2d} \frac{\partial^2}{\partial z_i^2} \rho_t^Z(z) \right) \, \left(\sum_{j=1}^{2d} \frac{\partial^2}{\partial z_j^2} g_t(z) \right) \, dz \notag \\
            &= -2\reg \int_{\R^{2d}} \sum_{i,j=1}^{2d}  \left(\frac{\partial^2}{\partial z_i^2} \rho_t^Z(z) \right) \, \left(\frac{\partial^2}{\partial z_j^2} g_t(z) \right) \, dz \notag \\
            &= -2\reg \int_{\R^{2d}} \sum_{i,j=1}^{2d}  \rho_t^Z(z) \, \left(\frac{\partial^4}{\partial z_i^2 \, \partial z_j^2} g_t(z) \right) \, dz \notag \\
            &= -2\reg \int_{\R^{2d}} \sum_{i,j=1}^{2d} \left(\frac{\partial^2}{\partial z_i \, \partial z_j} \rho_t^Z(z) \right) \left(\frac{\partial^2}{\partial z_i \, \partial z_j} g_t(z) \right) \, dz \notag \\
            &= -2\reg \int_{\R^{2d}} \langle \nabla^2 \rho_t^Z, \nabla^2 g_t \rangle_{\HS} \, dz \notag \\
            &= 2\reg \, \int_{\R^{2d}} \rho_t^Z \, \left( \langle \nabla^2 f_t, \nabla^2 g_t \rangle_{\HS} - \left\langle \nabla f_t, \, (\nabla^2 g_t) \, \nabla f_t \right\rangle \right) \, dz
        \end{align*}
        }
        where in the last step we have used relation~\eqref{Eq:Rel2}.
        \item For the first term in~\eqref{Eq:Calc3b}, we can compute using~\eqref{Eq:BRLaplacianx} and~\eqref{Eq:FPrhotY}, and integration by parts:
        \begin{align*}
            \partial_t (\Delta_x g_t^X(x))
            &= \reg^{-1} \int_{\R^d} \partial_t \rho_t^Y(y) \, \Delta_x V(x, y) \, dy \\
            &= \int_{\R^d} \left(\nabla \cdot \left(\rho_t^Y \nabla g_t^Y \right)(y) + \Delta \rho_t^Y(y) \right) \, \Delta_x V(x, y) \, dy \\
            &= \int_{\R^d} \rho_t^Y(y) \left(-\langle \nabla g_t^Y(y), \nabla_y \Delta_x V(x, y) \rangle \,+\, \Delta_y \Delta_x V(x, y) \right) dy.
        \end{align*}
        Therefore, since $\rho_t^Z = \rho_t^X \otimes \rho_t^Y$:
        \begin{align*}
            &-2 \int_{\R^{d}} \rho_t^X \, \partial_t (\Delta_x g_t^X) \, dx \\
            &= 2 \int_{\R^{d}} \rho_t^X(x) \, \int_{\R^d} \rho_t^Y(y) \left(\langle \nabla g_t^Y(y), \nabla_y \Delta_x V(x, y) \rangle \,-\, \Delta_y \Delta_x V(x, y) \right) dy  \, dx \\
            &= 2 \int_{\R^{2d}} \rho_t^Z \left(\langle \nabla g_t^Y, \nabla_y \Delta_x V \rangle \,-\, \Delta_y \Delta_x V \right) dz.
        \end{align*}
        \item For the second term in~\eqref{Eq:Calc3b}, we can compute using~\eqref{Eq:BRLaplaciany} and~\eqref{Eq:FPrhotX}, and integration by parts:
        \begin{align*}
            \partial_t (\Delta g_t^Y(y)) 
            &= -\reg^{-1} \int_{\R^d} \partial_t \rho_t^X(x) \, \Delta_y V(x, y) \, dx \\
            &= -\int_{\R^d} \left(\nabla \cdot \left(\rho_t^X \nabla g_t^X \right)(x) + \Delta \rho_t^X(x) \right) \, \Delta_y V(x, y) \, dx \\
            &= \int_{\R^d} \rho_t^X(x) \left(\langle \nabla g_t^X(x), \nabla_x \Delta_y V(x, y) \rangle \,-\, \Delta_x \Delta_y V(x, y) \right) dx.
        \end{align*}
        Therefore, since $\rho_t^Z = \rho_t^X \otimes \rho_t^Y$:
        \begin{align*}
            &-2 \int_{\R^{d}} \rho_t^Y \, \partial_t (\Delta_y g_t^Y) \, dy \\
            &= 2 \int_{\R^{d}} \rho_t^Y(y) \, \int_{\R^d} \rho_t^X(x) \left(-\langle \nabla g_t^X(x), \nabla_x \Delta_y V(x, y) \rangle \,+\, \Delta_x \Delta_y V(x, y) \right) dx \, dy \\
            &= 2 \int_{\R^{2d}} \rho_t^Z \left(-\langle \nabla g_t^X, \nabla_x \Delta_y V \rangle \,+\, \Delta_x \Delta_y V \right) dz.
        \end{align*}
    \end{enumerate}
    When we combine the calculations above, we see that the terms involving $\Delta_x \Delta_y V = \Delta_y \Delta_x V$ cancel, so we get:
    \begin{align}
        \purple{\frac{d}{dt} \left(-2 \E_{\rho_t^Z}\left[\Delta g_t \right] \right)}
        &\purple{\,= \E_{\rho_t^Z}\left[ 2\reg \langle \nabla g_t, \, \nabla \Delta g_t \rangle 
        + 2\reg\langle \nabla^2 f_t, \nabla^2 g_t \rangle_{\HS} 
        - 2\reg \left\langle \nabla f_t, \, (\nabla^2 g_t) \, \nabla f_t \right\rangle
        \right]} \notag \\
        &\qquad \purple{\,+\, \E_{\rho_t^Z}\left[2\langle \nabla g_t^Y, \nabla_y \Delta_x V \rangle \,-\, 2\langle \nabla g_t^X, \nabla_x \Delta_y V \rangle\right].} \label{Eq:ThirdTerm-1}      
    \end{align}

    %%%%%%%%%
    \paragraph{Combining the terms.}
    Combining the calculations in~\eqref{Eq:FirstTerm-1},~\eqref{Eq:SecondTerm-1}, and~\eqref{Eq:ThirdTerm-1} below,
    we see that the terms involving $\langle \nabla g_t^X, \Delta_y \nabla_x V \rangle$, $\langle \nabla g_t^Y, \Delta_x \nabla_y V \rangle$, 
    and $\left\langle \nabla f_t, \, (\nabla^2 g_t) \, \nabla f_t \right\rangle$ cancel.
    Then we can rearrange the remaining terms to get:
    \begin{align*}
        \frac{d}{dt} \FI(\rho_t^Z \,\|\, \nu_t^Z)
        &= \blue{\frac{d}{dt} \E_{\rho_t^Z}\left[\|\nabla f_t\|^2 \right]} 
        + \red{\frac{d}{dt} \E_{\rho_t^Z}\left[\|\nabla g_t\|^2 \right]}
        +
        \purple{\frac{d}{dt} \left(-2 \E_{\rho_t^Z}\left[\Delta g_t \right] \right)} \\
        &= \blue{-2\reg \, \E_{\rho_t^Z}\left[\left\|\nabla^2 f_t \right\|^2_{\HS} \right] 
        +2\reg \, \E_{\rho_t^Z}\left[\langle \nabla f_t, (\nabla^2 g_t) \, \nabla f_t \rangle  \right] 
        - 2\reg \, \E_{\rho_t^Z}\left[\langle \nabla f_t, \, \nabla \Delta g_t \rangle  \right]} \\
        &\qquad \red{\,+\, \E_{\rho_t^Z}\left[2\reg \, \langle \nabla f_t-\nabla g_t, (\nabla^2 g_t) \nabla g_t \rangle + 2\langle \nabla g_t^X, \Delta_y \nabla_x V \rangle - 2\langle \nabla g_t^Y, \Delta_x \nabla_y V \rangle \right]} \\
        &\qquad \purple{\,+\, \E_{\rho_t^Z}\left[ 2\reg \langle \nabla g_t, \, \nabla \Delta g_t \rangle 
        + 2\reg\langle \nabla^2 f_t, \nabla^2 g_t \rangle_{\HS} 
        - 2\reg \left\langle \nabla f_t, \, (\nabla^2 g_t) \, \nabla f_t \right\rangle
        \right]} \\
        &\qquad \purple{\,+\, \E_{\rho_t^Z}\left[2\langle \nabla g_t^Y, \nabla_y \Delta_x V \rangle \,-\, 2\langle \nabla g_t^X, \nabla_x \Delta_y V \rangle\right]} \\
        &= \blue{-2\reg \, \E_{\rho_t^Z}\left[\left\|\nabla^2 f_t \right\|^2_{\HS} \right] 
        - 2\reg \, \E_{\rho_t^Z}\left[\langle \nabla f_t, \, \nabla \Delta g_t \rangle  \right]} \\
        &\qquad \red{\,+\, \E_{\rho_t^Z}\left[2\reg \, \langle \nabla f_t-\nabla g_t, (\nabla^2 g_t) \nabla g_t \rangle \right]} \\
        &\qquad \purple{\,+\, \E_{\rho_t^Z}\left[ 2\reg \langle \nabla g_t, \, \nabla \Delta g_t \rangle 
        + 2\reg\langle \nabla^2 f_t, \nabla^2 g_t \rangle_{\HS}
        \right]}.
    \end{align*}
    We then complete the squares to form the first two terms below, which are the good terms that we want, and collect the rest in a remainder term:
    \begin{align*}
        \frac{d}{dt} \FI(\rho_t^Z \,\|\, \nu_t^Z)
        &= -2\reg \, \E_{\rho_t^Z}\left[\left\|\nabla^2 f_t-\nabla^2 g_t \right\|^2_{\HS} \right] - 2\reg \, \E_{\rho_t^Z}\left[\left\|\nabla f_t-\nabla g_t\right\|^2_{\nabla^2 g_t} \right] + 2\reg R(t)
    \end{align*}
    where we define the remainder term $R(t)$ as the difference of the two sides above (scaled by $2\reg$):
    \begin{align*}
        R(t) &:= -\E_{\rho_t^Z}\left[\langle \nabla^2 f_t, \nabla^2 g_t \rangle_{\HS}
        \right] + \E_{\rho_t^Z}\left[\left\|\nabla^2 g_t \right\|^2_{\HS} \right] \\
        &\qquad - \E_{\rho_t^Z}\left[\langle \nabla f_t, (\nabla^2 g_t) \nabla g_t \rangle \right] + \E_{\rho_t^Z}\left[\langle \nabla f_t, (\nabla^2 g_t) \nabla f_t \rangle \right]   \\
        &\qquad - \E_{\rho_t^Z}\left[\langle \nabla f_t, \, \nabla \Delta g_t \rangle  \right] + \E_{\rho_t^Z}\left[\langle \nabla g_t, \, \nabla \Delta g_t \rangle \right]
    \end{align*}
    We claim this remainder term is identically zero.
    Indeed, by relation~\eqref{Eq:Rel1} and integration by parts, we have:
    \begin{align*}
        \E_{\rho_t^Z} \left[\langle \nabla f_t, (\nabla^2 g_t) \nabla f_t \rangle \right]
        &= -\int_{\R^{2d}} \langle \nabla \rho_t^Z, \, (\nabla^2 g_t) \, \nabla f_t \rangle \, dz \\
        &= \int_{\R^{2d}} \rho_t^Z \left( \langle \nabla \Delta g_t, \nabla f_t \rangle + \langle \nabla^2 g_t, \nabla^2 f_t \rangle_{\HS} \right) \, dz \\
        &= \E_{\rho_t^Z}\left[\langle \nabla \Delta g_t, \nabla f_t \rangle + \langle \nabla^2 g_t, \nabla^2 f_t \rangle_{\HS} \right].
    \end{align*}
    where in the above we have used the identity $\nabla \cdot \nabla^2 g_t = \nabla \Delta g_t$.
    Similarly:
    \begin{align*}
        \E_{\rho_t^Z} \left[\langle \nabla f_t, (\nabla^2 g_t) \nabla g_t \rangle \right]
        &= -\int_{\R^{2d}} \langle \nabla \rho_t^Z, \, (\nabla^2 g_t) \, \nabla g_t \rangle \, dz \\
        &= \int_{\R^{2d}} \rho_t^Z \left( \langle \nabla \Delta g_t, \nabla g_t \rangle + \langle \nabla^2 g_t, \nabla^2 g_t \rangle_{\HS} \right) \, dz \\
        &= \E_{\rho_t^Z}\left[\langle \nabla \Delta g_t, \nabla g_t \rangle + \left\| \nabla^2 g_t \right\|^2_{\HS} \right].
    \end{align*}
    Therefore, we see that all the terms in the remainder term cancel, so indeed $R(t) = 0$.
    Thus, we have shown the identity:
    \begin{align*}
        \frac{d}{dt} \FI(\rho_t^Z \,\|\, \nu_t^Z)
        &= -2\reg \, \E_{\rho_t^Z}\left[\left\|\nabla^2 f_t-\nabla^2 g_t \right\|^2_{\HS} \right] - 2\reg \, \E_{\rho_t^Z} \left[\left\|\nabla f_t - \nabla g_t \right\|^2_{\nabla^2 g_t}\right].
    \end{align*}
\end{proof}

%%%%%%%%%%%%%%%%%
\subsection{Bound on the Second Moment Along the Mean-Field Dynamics}
\label{Sec:BoundSecondMomentMF}

We show that along the mean-field dynamics~\eqref{Eq:MFSystem}, the distributions remain in $\P(\R^{2dN})$.

%%%%%%%
\begin{lemma}\label{Lem:BoundMomentParticleMF}
    Assume Assumption~\ref{As:SCSmooth}.
    Suppose $\bar Z_t \sim \bar \rho_t^{Z}$ evolves following the mean-field min-max Langevin dynamics~\eqref{Eq:MFSystem} in $\R^{2d}$ from $\bar Z_0 \sim \bar \rho_0^{Z} \in \P(\R^{2d})$.
    Then $\bar \rho_t^{Z} \in \P(\R^{2d})$ for all $t \ge 0$.
\end{lemma}
\begin{proof}
    We note that $\bar \rho_t^{Z}$ is absolutely continuous with respect to the Lebesgue measure on $\R^{2d}$ by virtue of the Brownian motion component in the dynamics~\eqref{Eq:MFSystem}.
    We now show that the second moment of $\bar \rho_t^{Z}$ remains finite for all $t \ge 0$.

    Under Assumption~\ref{As:SCSmooth}, recall from Theorem~\ref{Thm:DetMinMaxGF} (in Section~\ref{Sec:ReviewDeterministic}) that there exists a unique equilibrium point $z^* = (x^*,y^*) \in \R^{2d}$ which satisfies $\nabla V(z^*) = 0$.
    Define the vector field $b^Z \colon \R^{2d} \to \R^{2d}$ by, for all $z = (x,y) \in \R^{2d}$:
    \begin{align*}
        b^Z(x,y) = \begin{pmatrix}
            -\nabla_x V(x,y) \\
            \nabla_y V(x,y)
        \end{pmatrix}.
    \end{align*}
    Observe that $b^Z$ is the case $N=1$ of the vector field $b^{\bZ}$ that we defined in~\eqref{Eq:Defbz} (in $\R^{2dN} = \R^{2d}$).
    Note $b^Z(z^*) = 0$, since $\nabla_x V(z^*) = \nabla_y V(z^*) = 0$. Lemma~\ref{Lem:LipschitzMonotone_bz} in the next section guarantees that $-b^Z$ is $\alpha$-strongly monotone.
    
    Since $\bar Z_t = (\bar X_t, \bar Y_t)$ evolves following the mean-field dynamics~\eqref{Eq:MFSystem}, its density $\bar \rho_t^{Z} = \bar \rho_t^Z \otimes \bar \rho_t^Y$ evolves following the Fokker-Planck equations:
    \begin{align*}
        \part{\bar \rho_t^{X}}{t} &= \nabla \cdot \left(\bar \rho_t^{X} \, \E_{\bar \rho_t^Y}[\nabla_x V(\cdot, \bar Y_t)] \right) + \reg \, \Delta \bar \rho_t^{X} \\
        \part{\bar \rho_t^{Y}}{t} &= -\nabla \cdot \left(\bar \rho_t^{Y} \, \E_{\bar \rho_t^X}[\nabla_y V(\bar X_t, \cdot)] \right) + \reg \, \Delta \bar \rho_t^{Y}.
    \end{align*}
    Then we can compute, using integration by parts:
    \begin{align*}
        &\frac{d}{dt} \E_{\bar \rho_t^{X}}\left[\left\|\bar X_t-x^*\right\|^2\right] \\
        &= \int_{\R^{d}} \part{\bar \rho_t^{X}(x)}{t} \, \|x-x^*\|^2 \, dx \\
        &= \int_{\R^{d}} \left(\nabla \cdot \left(\bar \rho_t^{X} \, \E_{\bar \rho_t^Y}[\nabla_x V(\cdot, \bar Y_t)]\right)(x) + \reg \, \Delta \bar \rho_t^{X}(x)\right) \|x-x^*\|^2 \, dx \\
        &= -\int_{\R^{d}} \bar \rho_t^{X}(x) \left\langle \E_{\bar \rho_t^Y}[\nabla_x V(x, \bar Y_t)], \nabla \left(\|x-x^*\|^2\right) \right \rangle dx + \reg \int_{\R^{d}} \bar \rho_t^{X}(x) \Delta \left(\|x-x^*\|^2\right) dx \\
        &= -2 \int_{\R^{d}} \bar \rho_t^{X}(x) \left\langle \E_{\bar \rho_t^Y}[\nabla_x V(x, \bar Y_t)], x-x^* \right \rangle \, dx + 2\reg d \\
        &= -2 \E_{\bar \rho_t^X}\left[ \left\langle \E_{\bar \rho_t^Y}[\nabla_x V(\bar X_t, \bar Y_t)], X_t-x^* \right \rangle \right] + 2\reg d \\
        &= -2 \E_{\bar \rho_t^Z}\left[ \left\langle \nabla_x V(\bar X_t, \bar Y_t), \bar X_t-x^* \right \rangle \right] + 2\reg d.
    \end{align*}
    Similarly, we can also compute:
    \begin{align*}
        \frac{d}{dt} \E_{\bar \rho_t^{Y}}\left[\left\|\bar Y_t-y^*\right\|^2\right] 
        &= 2 \E_{\bar \rho_t^Z}\left[ \left\langle \nabla_y V(\bar X_t, \bar Y_t), \bar Y_t-y^* \right \rangle \right] + 2\reg d.
    \end{align*}
    Adding the two identities above gives:
    \begin{align*}
        \frac{d}{dt} \E_{\bar \rho_t^{Z}}\left[\left\|\bar Z_t-z^*\right\|^2\right] 
        &= -2 \E_{\bar \rho_t^Z}\left[ \left\langle \nabla_x V(\bar X_t, \bar Y_t), \bar X_t-x^* \right \rangle \right]
        +2 \E_{\bar \rho_t^Z}\left[ \left\langle \nabla_y V(\bar X_t, \bar Y_t), \bar Y_t-y^* \right \rangle \right] + 4\reg d \\
        &= 2\E_{\bar \rho_t^Z}\left[ \left\langle b^Z(\bar Z_t), \bar Z_t-z^* \right \rangle \right] + 4\reg d \\
        &= 2\E_{\bar \rho_t^Z}\left[ \left\langle b^Z(\bar Z_t) - b^Z(z^*), \bar Z_t-z^* \right \rangle \right] + 4\reg d \\
        &\le -2\alpha \, \E_{\bar \rho_t^Z}\left[ \left\| \bar Z_t-z^* \right \|^2 \right] + 4\reg d
    \end{align*}
    where the inequality follows from the property that $-b^Z$ is $\alpha$-strongly monotone.
    We can write the differential inequality above equivalently as:
    \begin{align*}
        \frac{d}{dt} \left(e^{2\alpha t} \, \E_{\bar \rho_t^Z}\left[ \left\| \bar Z_t-z^* \right \|^2 \right] \right)
        \le e^{2\alpha t} \, 4 \reg d.
    \end{align*}
    Integrating from $0$ to $t$ and rearranging the result gives:
    \begin{align*}
        \E_{\bar \rho_t^Z}\left[ \left\| \bar Z_t-z^* \right \|^2 \right]
        &\le e^{-2\alpha t} \, \E_{\bar \rho_0^Z}\left[ \left\| \bar Z_0-z^* \right \|^2 \right] + \frac{(1-e^{-2\alpha t})}{\alpha} \, 2\reg d \\
        &\le \E_{\bar \rho_0^Z}\left[ \left\| \bar Z_0-z^* \right \|^2 \right] + \frac{2\reg d}{\alpha}.
    \end{align*}
    Since $\bar \rho_0^{Z} \in \P(\R^{2d})$, $\E_{\bar \rho_0^{Z}}\left[\left\|\bar Z_0\right\|^2\right] < \infty$, so $\E_{\bar \rho_0^{Z}}\left[\left\|\bar Z_0-z^*\right\|^2\right] \le 2\E_{\bar \rho_0^{Z}}\left[\left\|\bar Z_0\right\|^2\right] + 2\|z^*\|^2 < \infty$.
    Therefore, we also have for all $t \ge 0$:
    \begin{align*}
        \E_{\bar \rho_t^{Z}}\left[\left\|\bar Z_t\right\|^2\right]
        &\le 2 \, \E_{\bar \rho_t^Z}\left[ \left\| \bar Z_t-z^* \right \|^2 \right] + 2\|z^*\|^2 \\
        &\le 2 \, \E_{\bar \rho_0^Z}\left[ \left\| \bar Z_0-z^* \right \|^2 \right] + \frac{4\reg d}{\alpha} + 2\|z^*\|^2 
        \,<\, \infty
    \end{align*}
    which shows that $\bar \rho_t^{Z} \in \P(\R^{2d})$.
\end{proof}

%%%%%%%%%%%
\subsection{Proof of Theorem~\ref{Thm:ConvMF} (Convergence of the Mean-Field Min-Max Langevin Dynamics)}
\label{Sec:ConvMFProof}

\begin{proof}[Proof of Theorem~\ref{Thm:ConvMF}]
By Assumption~\ref{As:SCSmooth}, both $\bar \nu_t^X$ and $\bar \nu_t^Y$ are $(\alpha/\reg)$-strongly log-concave, which can be directly verified from the form of their log-density functions.
Then $\bar \nu_t^Z = \bar \nu_t^X \otimes \bar \nu_t^Y$ is also $(\alpha/\reg)$-SLC, i.e., 
$-\nabla^2 \log \bar \nu_t^Z \succeq (\alpha/\reg) I.$
Then by the identity from Lemma~\ref{Lem:ddtFI_MF}, we have:
\begin{align*}
    \frac{d}{dt} \FI(\bar \rho_t^Z \,\|\, \bar \nu_t^Z)
    &= -2\reg \, \E_{\bar \rho_t^Z}\left[\left\|\nabla^2 \log \frac{\bar \rho_t^Z}{\bar \nu_t^Z} \right\|^2_{\HS} \right] - 2\reg \, \E_{\bar \rho_t^Z} \left[\left\|\nabla \log \frac{\bar \rho_t^Z}{\bar \nu_t^Z} \right\|^2_{(-\nabla^2 \log \bar \nu_t^Z)}\right] \\
    &\le - 2\reg \, \E_{\bar \rho_t^Z} \left[\left\|\nabla \log \frac{\bar \rho_t^Z}{\bar \nu_t^Z} \right\|^2_{(-\nabla^2 \log \bar \nu_t^Z)}\right] \\
    &\le -2\alpha \, \E_{\bar \rho_t^Z} \left[\left\|\nabla \log \frac{\bar \rho_t^Z}{\bar \nu_t^Z} \right\|^2 \right] \\
    &= -2\alpha \, \FI(\bar \rho_t^Z \,\|\, \bar \nu_t^Z).
\end{align*}
In the first inequality above, we drop the first term in the identity, which is the second-order relative Fisher information.
In the second inequality, we use the property that $\bar \nu_t^Z$ is $(\alpha/\reg)$-SLC.
In the next step, we recognize the term in the previous line as also equal to the relative Fisher information.
Integrating the differential inequality above implies the desired convergence rate:
$$\FI(\bar \rho_t^Z \,\|\, \bar \nu_t^Z) \le e^{-2\alpha t} \, \FI(\bar \rho_0^Z \,\|\, \bar \nu_0^Z).$$

Since $\bar \nu_t^Z$ is $(\alpha/\reg)$-SLC, it also satisfies $(\alpha/\reg)$-LSI. 
Then by combining with the convergence rate for relative Fisher information above, we get:
\begin{align*}
    \KL(\bar \rho_t^Z \,\|\, \bar \nu_t^Z) \le \frac{\reg}{2\alpha} \, \FI(\bar \rho_t^Z \,\|\, \bar \nu_t^Z)
    \le \frac{\reg}{2\alpha} \, e^{-2\alpha t} \, \FI(\bar \rho_0^Z \,\|\, \bar \nu_0^Z).
\end{align*}
Multiplying both sides by $\reg$ yields the desired convergence rate in duality gap.
\end{proof}

%%%%%%%%%%%%%%
\subsubsection{Bound on Relative Fisher Information to the Best-Response Distribution}
\label{Sec:BoundFIBestResponse}

We provide the following bound on the initial Fisher information $\FI(\bar \rho_0^Z \,\|\, \bar \nu_0^Z)$ when $\bar \rho_0^Z$ is Gaussian and $\bar \nu_0^Z$ is its best-response distribution.
Under Assumption~\ref{As:SCSmooth}, recall from Theorem~\ref{Thm:DetMinMaxGF} (in Section~\ref{Sec:ReviewDeterministic}) that there exists a unique equilibrium point $z^* = (x^*,y^*) \in \R^{2d}$ which satisfies $\nabla V(z^*) = 0$,
and that $\nabla V$ is $L$-Lipschitz.

\begin{lemma}\label{Lem:InitialGaussianFI} Assume Assumption \ref{As:SCSmooth}. Let $\bar \rho_0^X = \bar \rho_0^Y = \mathcal{N}(0,\frac{\reg^2}{L^2}I_d)$, and $\bar \rho_0^Z = \bar \rho_0^X \otimes \bar \rho_0^Y$. Then:
\begin{align*}    
    \FI(\bar \rho_0^Z \,\|\, \bar \nu_0^Z) \leq 2d \left(1 + \frac{L^2}{\reg^2} \right) + \frac{L^2 \|z^*\|^2}{\reg^2}.
\end{align*}
\end{lemma}
\begin{proof}
    Recall that the best-response distributions are defined by: 
    \begin{align*}        
        \bar \nu_0^X &\propto \exp\left(-\reg^{-1}\E_{\rho_0^Y }[V(x, Y)] \right), \\
        \bar \nu_0^Y &\propto \exp\left(\reg^{-1}\E_{\bar \rho_0^X }[V(X, y)] \right),
    \end{align*}
    and $\bar \nu_0^Z = \bar \nu_0^X \otimes \bar \nu_0^Y$.
    Since both $\bar \rho_0^Z$ and $\bar \nu_0^Z$ are product distributions,
    $$\FI(\bar \rho_0^Z \,\|\, \bar \nu_0^Z) = \FI(\rho^X \,\|\, \bar \nu^X) + \FI(\rho^Y \,\|\, \bar \nu^Y).$$

    Define $g^X = -\log \bar \nu_0^X$, so $\nabla g^X(x) = \reg^{-1} \E_{\bar \rho_0^Y}[\nabla_x V(x, Y)]$
    and $\Delta g^X(x) = \reg^{-1} \E_{\bar \rho_0^Y}[\Delta_x V(x,\bar Y)]$.
    Since we assume $V(x,y)$ is $\alpha$-strongly convex in $x$, we have $\Delta g^X(x) \ge \alpha d/\reg \ge 0$ for all $x \in \R^d$.
    Note also that for $\bar \rho_0^X = \N(0, \frac{\reg^2}{L^2}I)$ on $\R^d$, we have
    $$\E_{\bar \rho_0^X}\left[\left\| \nabla \log \bar \rho_0^X \right\|^2\right] = \frac{L^4}{\reg^4} \E_{\bar \rho_0^X}\left[\left\|X\right\|^2\right] = \frac{L^2d}{\reg^2}.$$
    Then by expanding the square and using integration by parts, we can write:
    \begin{align*}
        \FI(\bar \rho_0^X \,\|\, \bar \nu_0^X)
        &= \E_{\bar \rho_0^X}\left[\left\| \nabla \log \bar \rho_0^X + \nabla g^X \right\|^2\right] \\
        &= \E_{\bar \rho_0^X}\left[\left\| \nabla \log \bar \rho_0^X \right\|^2\right] -2 \, \E_{\bar \rho_0^X}\left[\Delta g^X\right] + \E_{\bar \rho_0^X}\left[\left\|\nabla g^X \right\|^2\right] \\
        &\le \frac{L^2d}{\reg^2} + \E_{X \sim \bar \rho_0^X}\left[\left\|\reg^{-1} \E_{Y \sim \bar \rho_0^Y}[\nabla_x V(X, Y)] \right\|^2\right] \\
        &\le \frac{L^2d}{\reg^2} + \frac{1}{\reg^2} \E_{\bar \rho_0^X \otimes \bar \rho_0^Y}\left[\left\|\nabla_x V(X, Y) \right\|^2\right].
    \end{align*}
    By an identical argument, we can similarly show:
    \begin{align*}
        \FI(\bar \rho_0^Y \,\|\, \bar \nu_0^Y)
        &\le \frac{L^2d}{\reg^2} + \frac{1}{\reg^2} \E_{\bar \rho_0^X \otimes \bar \rho_0^Y}\left[\left\|\nabla_y V(X, Y) \right\|^2\right].
    \end{align*}
    Therefore,
    \begin{align*}
        \FI(\bar \rho_0^Z \,\|\, \bar \nu_0^Z)
        &\le \frac{2L^2d}{\reg^2} + \frac{1}{\reg^2} \E_{\bar \rho_0^X \otimes \bar \rho_0^Y}\left[\left\|\nabla_x V(X, Y) \right\|^2 + \left\|\nabla_y V(X, Y) \right\|^2\right] \\
        &= \frac{2Ld}{\reg} + \frac{1}{\reg^2} \E_{\bar \rho_0^Z}\left[\left\|\nabla V(Z) \right\|^2\right].
    \end{align*}
    Using $\nabla V(z^*) = 0$ and $\nabla V$ is $L$-Lipschitz, and since $\bar \rho_0^Z = \N(0, \frac{\reg^2}{L^2} I_{2d})$, we can bound:
    \begin{align*}
        \E_{\bar \rho_0^Z}\left[\left\|\nabla V(Z) \right\|^2\right]
        &= \E_{\bar \rho_0^Z}\left[\left\|\nabla V(Z) - \nabla V(z^*) \right\|^2\right] 
        \le L^2\E_{\bar \rho_0^Z}\left[\left\|Z - z^* \right\|^2\right] 
        = 2\reg^2 d + L^2 \|z^*\|^2.
    \end{align*}
    Plugging this in to our earlier calculation above, we obtain:
    \begin{align*}
        \FI(\bar \rho_0^Z \,\|\, \bar \nu_0^Z)
        &\le \frac{2L^2d}{\reg^2} + \frac{1}{\reg^2} \left(2\reg^2 d + L^2 \|z^*\|^2 \right) 
        = 2d \left(1 + \frac{L^2}{\reg^2} \right) + \frac{L^2 \|z^*\|^2}{\reg^2}.
    \end{align*}
\end{proof}

%%%%%%%%%%%
\section{Technical Lemmas for the Finite-Particle Analysis}
\label{Sec:ParticleLemmas}

We collect some lemmas for our analysis of the finite-particle dynamics and algorithm.
Here we consider a particle discretization of the mean-field dynamics with $N$ particles.
This means we work with a joint vector $\bz = (\bx,\by) = (x^1,\dots,x^N,y^1,\dots,y^N) \in \R^{2dN}$, where each $x^i, y^j \in \R^d$.

We recall some definitions.
Recall $\bar \nu^Z = \bar \nu^X \otimes \bar \nu^Y \in \P(\R^{2d})$ is the stationary distribution of the mean-field dynamics~\eqref{Eq:MFSystem} in $\R^{2d}$.
Let $\Var_{\bar \nu^Z}(\bar Z) = \E_{\bar \nu^Z}[\|\bar Z - \E_{\bar \nu^Z}[\bar Z]\|^2]$ be the variance of $\bar \nu^Z$.
Recall under Assumption~\ref{As:SCSmooth}, $\bar \nu^Z$ is $(\alpha/\reg)$-SLC, so we can bound the variance by $\Var_{\bar \nu^Z}(\bar Z) \le \frac{2\reg d}{\alpha}$ (see Lemma~\ref{Lem:BoundVar} in Section~\ref{Sec:BoundInitFI}), but in our computations below, we will keep it as the variance.

Recall we defined in~\eqref{Eq:TensorMF} the tensorized power of the stationary mean-field distribution: 
$$\bar \nu^{\bZ} = (\bar \nu^X)^{\otimes N} \otimes (\bar \nu^Y)^{\otimes N} \in \P(\R^{2dN}).$$
Recall we defined in~\eqref{Eq:Defbz} the vector field $b^{\bZ} \colon \R^{2dN} \to \R^{2dN}$, which is the drift term in the finite-particle dynamics~\eqref{Eq:ParticleSystem3} and the finite-particle algorithm~\eqref{Eq:ParticleAlgorithm3}.

Let us also define a vector field $\bar b^{\bZ} \colon \R^{2dN} \to \R^{2dN}$ by, for all $\bz = (\bx,\by) \in \R^{2dN}$:
\begin{align}\label{Eq:Defbzbar}
    \bar b^{\bZ}(\bz) = \begin{pmatrix}
        \bar b^{\bX}(\bx) \\
        \bar b^{\bY}(\by)
    \end{pmatrix}
\end{align}
where we define the vector fields $\bar b^{\bX} \colon \R^{dN} \to \R^{dN}$ and $\bar b^{\bY} \colon \R^{dN} \to \R^{dN}$ by, for all $\bx = (x^1,\dots,x^N) \in \R^{dN}$ and $\by = (y^1,\dots,y^N) \in \R^{dN}$:
\begin{align*}
    \bar b^{\bX}(\bx) &= 
    \begin{pmatrix}
        \bar b^X(x^1) \\
        \cdots \\
        \bar b^X(x^N)
    \end{pmatrix}
    := \begin{pmatrix}
        \E_{\bar \nu^Y}[-\nabla_x V(x^1, \bar Y)] \\
        \cdots \\
        \E_{\bar \nu^Y}[-\nabla_x V(x^N, \bar Y)]
    \end{pmatrix}, \\
    \bar b^{\bY}(\by) &= 
    \begin{pmatrix}
        \bar b^Y(y^1) \\
        \cdots \\
        \bar b^Y(y^N)
    \end{pmatrix}
    := \begin{pmatrix}
        \E_{\bar \nu^X}[\nabla_y V(\bar X, y^1)] \\
        \cdots \\
        \E_{\bar \nu^X}[\nabla_y V(\bar X, y^N)]
    \end{pmatrix}.
\end{align*}
By the definition of $\bar \nu^Z$ as the stationary distribution of the mean-field dynamics~\eqref{Eq:MFSystem}, so it is a fixed point for the best-response distribution, we observe that the vector field $\bar b^{\bZ}$ can be written as a scaled version of the score function of the tensorized stationary mean-field distribution $\bar \nu^{\bZ}$:
\begin{align}\label{Eq:TensorMFScore}
    \bar b^{\bZ}(\bz) = \reg \nabla \log \bar \nu^{\bZ}(\bz).    
\end{align}

Finally, we also define the following stochastic process for $\bar \bZ_t \in \R^{2dN}$ with drift term $\bar b^{\bZ}$:
\begin{align}\label{Eq:MFSystemTensor}
    d\bar \bZ_t = \bar b^{\bZ}(\bar \bZ_t) \, dt + \sqrt{2\reg} \, dW_t^{\bZ}
\end{align}
where $W_t^{\bZ}$ is the standard Brownian motion in $\R^{2dN}$.
Observe that since $\bar b^{\bZ} = \reg \nabla \log \bar \nu^{\bZ}$,
the stationary distribution of the process~\eqref{Eq:MFSystemTensor} is equal to $\bar \nu^{\bZ}$.
Therefore, we call the process~\eqref{Eq:MFSystemTensor} the \textit{tensorized mean-field dynamics}, since we can obtain it by replicating the base mean-field dynamics~\eqref{Eq:MFSystem} from $\R^{2d}$ for $N$ times.
We will use the process~\eqref{Eq:MFSystemTensor} to compare the finite-particle dynamics~\eqref{Eq:ParticleSystem3} and algorithm~\eqref{Eq:ParticleAlgorithm3}.

%%%%%%%%%%%%%%%%%%%%%%%%
\subsection{Properties of the Vector Fields}
\label{Sec:PropertiesVectorFields}

%%%%%%%%%%%%%%%%%%%%%%%%
\subsubsection{Properties of the finite-particle vector field}

\begin{lemma}\label{Lem:LipschitzMonotone_bz}
    Assume Assumption~\ref{As:SCSmooth}.
    Then:
    \begin{enumerate}
        \item The vector field $b^{\bZ}$ defined in~\eqref{Eq:Defbz} is $(2L)$-Lipschitz, which means for all $\bz, \bar \bz \in \R^{2dN}$:
        $$\left\| b^{\bZ}(\bz) - b^{\bZ}(\bar \bz) \right\| \le 2L \|\bz-\bar \bz\|.$$    
        \item Furthermore, $-b^{\bZ}$ is $\alpha$-strongly monotone, which means for all $\bz, \bar \bz \in \R^{2dN}$:
        $$\left\langle b^{\bZ}(\bz) - b^{\bZ}(\bar \bz), \bz-\bar \bz \right\rangle \le -\alpha \|\bz-\bar \bz\|^2.$$
    \end{enumerate}
\end{lemma}
\begin{proof}        
    Let $\bz = (\bx,\by) = (x^1,\dots,x^N,y^1,\dots,y^N)$
    and $\bar \bz = (\bar \bx,\bar \by) = (\bar x^1,\dots,\bar x^N,\bar y^1,\dots,\bar y^N) \in \R^{2dN}$ be given.
    
    \textbf{(1)~} We show $b^{\bZ}$ is $(2L)$-Lipschitz.
    Recall by assumption, $(x,y) \mapsto \nabla V(x,y)$ is $L$-Lipschitz; in particular, each component $\nabla_x V(x,y)$ and $\nabla_y V(x,y)$ is also $L$-Lipschitz. 
    By definition,
    \begin{align*}
        \left\| b^{\bZ}(\bz) - b^{\bZ}(\bar \bz) \right\|^2
        &= \left\| b^{\bX}(\bx,\by) - b^{\bX}(\bar \bx, \bar \by) \right\|^2 + \left\| b^{\bY}(\bx,\by) - b^{\bY}(\bar \bx, \bar \by) \right\|^2.
    \end{align*}
    We bound each term above separately.
    For the first term:
    {\allowdisplaybreaks
    \begin{align*}
        &\left\| b^{\bX}(\bx,\by) - b^{\bX}(\bar \bx, \bar \by) \right\|^2 \\
        &\stackrel{(1)}{=} \sum_{i \in [N]} \left\| b^{X}(x^i,\by) - b^{X}(\bar x^i, \bar \by) \right\|^2 \\
        &\stackrel{(2)}{=} \sum_{i \in [N]} \left\| -\frac{1}{N} \sum_{j \in [N]} \nabla_x V(x^i, y^j) + \frac{1}{N} \sum_{j \in [N]} \nabla_x V(\bar x^i, \bar y^j) \right\|^2 \\
        &\stackrel{(3)}{=} \frac{1}{N^2} \sum_{i \in [N]} \left\| \sum_{j \in [N]} (\nabla_x V(x^i, y^j) - \nabla_x V(\bar x^i, \bar y^j)) \right\|^2 \\
        &\stackrel{(4)}{\le} \frac{1}{N} \sum_{i \in [N]} \sum_{j \in [N]} \left\| \nabla_x V(x^i, y^j) - \nabla_x V(\bar x^i, \bar y^j) \right\|^2 \\
        &\stackrel{(5)}{\le} \frac{2}{N} \sum_{i \in [N]} \sum_{j \in [N]} \left(\left\| \nabla_x V(x^i, y^j) - \nabla_x V(\bar x^i, y^j) \right\|^2 + \left\| \nabla_x V(\bar x^i, y^j) - \nabla_x V(\bar x^i, \bar y^j) \right\|^2\right) \\
        &\stackrel{(6)}{\le} \frac{2L^2}{N} \sum_{i \in [N]} \sum_{j \in [N]} \left(\left\| x^i-\bar x^i \right\|^2 + \left\| y^j - \bar y^j \right\|^2\right) \\
        &\stackrel{(7)}{=} \frac{2L^2}{N} \sum_{j \in [N]} \left\| \bx-\bar \bx \right\|^2 + \frac{2L^2}{N} \sum_{i \in [N]} \left\| \by - \bar \by \right\|^2 \\
        &\stackrel{(8)}{=} 2L^2 \left( \left\| \bx-\bar \bx \right\|^2 + \left\| \by - \bar \by \right\|^2 \right) \\
        &\stackrel{(9)}{=} 2L^2 \left\| \bz-\bar \bz \right\|^2.
    \end{align*}
    }
    In the above, steps (1) and (2) are by definitions; in step (3) we pull the $1/N$ outside the square.
    Step (4) follows from Cauchy-Schwarz inequality ($\|\sum_{j \in [N]} a_j\|^2 \le N \sum_{j \in [N]} \|a_j\|^2$).
    In step (5) we introduce an intermediate term $\nabla_x V(\bar x^i, y^j)$ and use the inequality $\|a+b\|^2 \le 2\|a\|^2 + 2\|b\|^2$.
    In step (6) we use the property that $\nabla_x V$ is $L$-Lipschitz, and note that in the first term the $y^j$ part is common, while in the second term the $\bar x^i$ part is common.
    In step (7) we write the previous terms in vector notation.
    In step (8) we collect the terms inside the previous summation which are all equal.
    In the last step (9) we use the definitions $\bz = (\bx,\by)$ and $\bar \bz = (\bar \bx, \bar \by)$.

    By a similar argument and using the $L$-Lipschitz property of $\nabla_y V$, we can show the second term is also bounded by the same quantity:
    \begin{align*}
        \left\| b^{\bY}(\bx,\by) - b^{\bY}(\bar \bx, \bar \by) \right\|^2
        \le 2L^2 \left\| \bz-\bar \bz \right\|^2.
    \end{align*}
    Combining the two bounds above, we obtain:
    \begin{align*}
        \left\| b^{\bZ}(\bz) - b^{\bZ}(\bar \bz) \right\|^2
        \,\le\, 2L^2 \left\| \bz-\bar \bz \right\|^2 + 2L^2 \left\| \bz-\bar \bz \right\|^2 
        \,=\, 4L^2 \left\| \bz-\bar \bz \right\|^2
    \end{align*}
    which shows that $b^{\bZ}$ is $(2L)$-Lipschitz.

    \textbf{(2)~}
    We now show that $-b^{\bZ}$ is $\alpha$-strongly monotone.
    This is equivalent to showing that the symmetrized Jacobian of $b^{\bZ}$ satisfies, for all $\bz \in \R^{2dN}$:
    \begin{align}\label{Eq:SymJac}
        (\nabla b^{\bZ}(\bz))_{\sym} := \frac{1}{2} \left(\nabla b^{\bZ}(\bz) + \nabla b^{\bZ}(\bz)^\top \right) \preceq -\alpha I.
    \end{align}
    Indeed, if we have~\eqref{Eq:SymJac}, then for all $\bz, \bar \bz \in \R^{2dN}$, by the mean-value theorem we can write:
    \begin{align*}
        \left\langle b^{\bZ}(\bz) - b^{\bZ}(\bar \bz), \bz-\bar \bz \right\rangle
        &= \left\langle \int_0^1 \nabla b^{\bZ}(\bz_t) \, (\bz - \bar \bz) \, dt, \bz-\bar \bz \right\rangle \\
        &= \int_0^1 (\bz-\bar \bz)^\top \nabla b^{\bZ}(\bz_t) \, (\bz - \bar \bz) \, dt \\
        &= \int_0^1 (\bz-\bar \bz)^\top (\nabla b^{\bZ}(\bz_t))_\sym \, (\bz - \bar \bz) \, dt \\
        &\le -\alpha \int_0^1 \left\|\bz-\bar \bz \right\|^2 \, dt \\
        &= -\alpha \left\|\bz-\bar \bz \right\|^2
    \end{align*}
    where in the above we have defined $\bz_t := (1-t)\bz + t \bar \bz$ for $0 \le t \le 1$, and used the property that $u^\top A u = \frac{1}{2} u^\top (A+A^\top) u = u^\top (A_\sym) u$ for all $u \in \R^D$ and $A \in \R^{D \times D}$.

    To show~\eqref{Eq:SymJac}, we compute the Jacobian matrix $\nabla b^{\bZ}(\bz) \in \R^{2dN \times 2dN}$, which is a block matrix with $2N$ blocks on each dimension, indexed by $x^1,\dots,x^N,y^1,\dots,y^N$, with the following entries:
    {\allowdisplaybreaks
    \begin{align*}
        (\nabla b^{\bZ}(\bz))[x^i,x^i] &= \part{}{x^i} b^{\bZ}(\bz)[x^i] = -\frac{1}{N} \sum_{k \in [N]} \nabla^2_{xx} V(x^i,y^k) \qquad &\forall ~ i \in [N], \\
        (\nabla b^{\bZ}(\bz))[x^i,x^j] &= \part{}{x^j} b^{\bZ}(\bz)[x^i] = 0 &\forall ~ i \neq j \in [N], \\
        (\nabla b^{\bZ}(\bz))[x^i,y^j] &= \part{}{y^j} b^{\bZ}(\bz)[x^i] = -\frac{1}{N} \nabla^2_{yx} V(x^i,y^j) &\forall ~ i, j \in [N], \\
        (\nabla b^{\bZ}(\bz))[y^i,y^i] &= \part{}{y^i} b^{\bZ}(\bz)[y^i] = \frac{1}{N} \sum_{k \in [N]} \nabla^2_{yy} V(x^k,y^i) \qquad &\forall ~ i \in [N], \\
        (\nabla b^{\bZ}(\bz))[y^i,y^j] &= \part{}{y^j} b^{\bZ}(\bz)[y^i] = 0 &\forall ~ i \neq j \in [N], \\
        (\nabla b^{\bZ}(\bz))[y^i,x^j] &= \part{}{x^j} b^{\bZ}(\bz)[y^i] = \frac{1}{N} \nabla^2_{xy} V(x^j,y^i) &\forall ~ i, j \in [N].
    \end{align*}
    }
    Now for the symmetrized Jacobian $(\nabla b^{\bZ}(\bz))_\sym$, we can compute its block entries:
    \begin{align*}
        (\nabla b^{\bZ}(\bz))_\sym[x^i,x^i] &= -\frac{1}{N} \sum_{k \in [N]} \nabla^2_{xx} V(x^i,y^k) \qquad &\forall ~ i \in [N], \\
        (\nabla b^{\bZ}(\bz))_\sym[x^i,x^j] &= 0 &\forall ~ i \neq j \in [N], \\
        (\nabla b^{\bZ}(\bz))_\sym[x^i,y^j] &= \frac{1}{2} \left(\nabla b^{\bZ}(\bz)[x^i,y^j] + \nabla b^{\bZ}(\bz)[y^j,x^i]^\top \right) \\
        &= \frac{1}{2} \left(-\frac{1}{N} \nabla^2_{yx} V(x^i,y^j) + \frac{1}{N} \nabla^2_{xy} V(x^i,y^j)^\top \right) = 0 &\forall ~ i, j \in [N], \\
        (\nabla b^{\bZ}(\bz))_\sym[y^i,y^i] &= \frac{1}{N} \sum_{k \in [N]} \nabla^2_{yy} V(x^k,y^i) \qquad &\forall ~ i \in [N], \\
        (\nabla b^{\bZ}(\bz))_\sym[y^i,y^j] &= 0 &\forall ~ i \neq j \in [N], \\
        (\nabla b^{\bZ}(\bz))_\sym[y^i,x^j] 
        &= \frac{1}{2} \left(\nabla b^{\bZ}(\bz)[y^i,x^j] + \nabla b^{\bZ}(\bz)[x^j,y^i]^\top \right) \\
        &= \frac{1}{2} \left(\frac{1}{N} \nabla^2_{xy} V(x^j,y^i) - \frac{1}{N} \nabla^2_{yx} V(x^j,y^i)^\top \right) = 0 &\forall ~ i, j \in [N].
    \end{align*}
    Therefore, we see that the symmetrized Jacobian $(\nabla b^{\bZ}(\bz))_\sym$ is block-diagonal, since all the off-diagonal block entries are $0$.
    Moreover, from the assumption that $V(x,y)$ is $\alpha$-strongly convex in $x$ and $\alpha$-strongly concave in $y$, the block-diagonal entries satisfy, for all $i \in [N]$:
    \begin{align*}
        (\nabla b^{\bZ}(\bz))_\sym[x^i,x^i] &= -\frac{1}{N} \sum_{k \in [N]} \nabla^2_{xx} V(x^i,y^k)
        \preceq -\alpha I \\
        (\nabla b^{\bZ}(\bz))_\sym[y^i,y^i] &= \frac{1}{N} \sum_{k \in [N]} \nabla^2_{yy} V(x^k,y^i) 
        \preceq -\alpha I.
    \end{align*}    
    This shows that $(\nabla b^{\bZ}(\bz))_\sym \preceq -\alpha I$, and thus $-b^{\bZ}$ is $\alpha$-strongly monotone, as desired.
\end{proof}

%%%%%%%%%%%%%%%%%%%%%%%%
\subsubsection{Properties of the mean-field vector field}

\begin{lemma}\label{Lem:LipschitzMonotone_bzbar}
    Assume Assumption~\ref{As:SCSmooth}.
    Then:
    \begin{enumerate}
        \item The vector field $\bar b^{\bZ}$ defined in~\eqref{Eq:Defbzbar} is $L$-Lipschitz, which means for all $\bz, \bar \bz \in \R^{2dN}$:
        $$\left\|\bar b^{\bZ}(\bz) - \bar b^{\bZ}(\bar \bz) \right\| \le L \|\bz-\bar \bz\|.$$
        \item Furthermore, $-\bar b^{\bZ}$ is $\alpha$-strongly monotone, which means for all $\bz, \bar \bz \in \R^{2dN}$:
        $$\left\langle \bar b^{\bZ}(\bz) - \bar b^{\bZ}(\bar \bz), \bz-\bar \bz \right\rangle \le -\alpha \|\bz-\bar \bz\|^2.$$
    \end{enumerate}    
\end{lemma}
\begin{proof}        
    Let $\bz = (\bx,\by) = (x^1,\dots,x^N,y^1,\dots,y^N)$
    and $\bar \bz = (\bar \bx,\bar \by) = (\bar x^1,\dots,\bar x^N,\bar y^1,\dots,\bar y^N) \in \R^{2dN}$ be given.

    \begin{enumerate}
        \item We first show $\bar b^{\bZ}$ is $L$-Lipschitz.
        By assumption, $(x,y) \mapsto \nabla V(x,y)$ is $L$-Lipschitz; in particular, each component $\nabla_x V(x,y)$ and $\nabla_y V(x,y)$ is also $L$-Lipschitz. 
        By definition,
        \begin{align*}
            \left\|\bar b^{\bZ}(\bz) - \bar b^{\bZ}(\bar \bz) \right\|^2
            &= \left\| \bar b^{\bX}(\bx) - \bar b^{\bX}(\bar \bx) \right\|^2 + \left\| \bar b^{\bY}(\by) - \bar b^{\bY}(\bar \by) \right\|^2.
        \end{align*}
        We can bound the first term as:
        \begin{align*}
            \left\| \bar b^{\bX}(\bx) - \bar b^{\bX}(\bar \bx) \right\|^2 
            &= \sum_{i \in [N]} \left\| \E_{\bar \nu^Y}[\nabla_x V(x^i, \bar Y) - \nabla_x V(\bar x^i, \bar Y)] \right\|^2  \\
            &\le \sum_{i \in [N]} \E_{\bar \nu^Y}\left[\left\| \nabla_x V(x^i, \bar Y) - \nabla_x V(\bar x^i, \bar Y) \right] \right\|^2  \\
            &\le L^2 \sum_{i \in [N]} \E_{\bar \nu^Y}\left[\left\| x^i - \bar x^i \right] \right\|^2  \\
            &= L^2 \left\| \bx - \bar \bx \right\|^2
        \end{align*}
        where the first inequality is by Cauchy-Schwarz, and the second inequality is by the $L$-Lipschitz property of $\nabla_x V$.
        Similarly, using the $L$-Lipschitz property of $\nabla_y V$, we can bound the second term as:
        \begin{align*}
            \left\| \bar b^{\bY}(\by) - \bar b^{\bY}(\bar \by) \right\|^2
            \le L^2 \left\| \by - \bar \by \right\|^2.
        \end{align*}
        Combining the two bounds above gives:
        \begin{align*}
            \left\|\bar b^{\bZ}(\bz) - \bar b^{\bZ}(\bar \bz) \right\|^2
            &\le L^2 \left\| \bx - \bar \bx \right\|^2 + L^2 \left\| \by - \bar \by \right\|^2 
            \,=\, L^2 \left\| \bz - \bar \bz \right\|^2
        \end{align*}
        which shows that $\bar b^{\bZ}$ is $L$-Lipschitz.
    
        \item We now show $-\bar b^{\bZ}$ is $\alpha$-strongly monotone.
        By assumption, for each $y \in \R^d$, $x \mapsto V(x,y)$ is $\alpha$-strongly convex, so $x \mapsto \nabla_x V(x,y)$ is $\alpha$-strongly monotone.
        Similarly, for each $x \in \R^d$, $y \mapsto V(x,y)$ is $\alpha$-strongly concave, so $y \mapsto -\nabla_y V(x,y)$ is $\alpha$-strongly monotone.
        By definition,
        \begin{align*}
            \left\langle \bar b^{\bZ}(\bz) - \bar b^{\bZ}(\bar \bz), \bz-\bar \bz \right\rangle
            &= \left\langle \bar b^{\bX}(\bx) - \bar b^{\bX}(\bar \bx), \bx-\bar \bx \right\rangle
            + \left\langle \bar b^{\bY}(\by) - \bar b^{\bY}(\bar \by), \by-\bar \by \right\rangle.
        \end{align*}
        We can bound the first term as:
        \begin{align*}
            \left\langle \bar b^{\bX}(\bx) - \bar b^{\bX}(\bar \bx), \bx-\bar \bx \right\rangle
            &= \sum_{i \in [N]} \left\langle \E_{\bar \nu^Y}[-\nabla_x V(x^i, \bar Y)] + \E_{\bar \nu^Y}[\nabla_x V(\bar x^i, \bar Y)], x^i-\bar x^i \right\rangle \\
            &= -\sum_{i \in [N]} \E_{\bar \nu^Y}\left[\left\langle \nabla_x V(x^i, \bar Y) - \nabla_x V(\bar x^i, \bar Y), x^i-\bar x^i \right\rangle \right] \\
            &\le -\alpha \sum_{i \in [N]} \E_{\bar \nu^Y}\left[\left\| x^i - \bar x^i \right\|^2 \right] \\
            &= -\alpha \left\| \bx - \bar \bx \right\|^2
        \end{align*}
        where the inequality above follows from the $\alpha$-strong monotonicity of $x \mapsto \nabla_x V(x, \bar Y)$.
        
        Similarly, we can bound the second term as:
        \begin{align*}
            \left\langle \bar b^{\bY}(\by) - \bar b^{\bY}(\bar \by), \by-\bar \by \right\rangle
            &= \sum_{i \in [N]} \left\langle \E_{\bar \nu^X}[\nabla_y V(\bar X, y^i)] - \E_{\bar \nu^X}[\nabla_y V(\bar X, \bar y^i)], y^i-\bar y^i \right\rangle \\
            &= \sum_{i \in [N]} \E_{\bar \nu^X}\left[\left\langle \nabla_y V(\bar X, y^i) - \nabla_y V(\bar X, \bar y^i), y^i-\bar y^i \right\rangle \right] \\
            &\le -\alpha \sum_{i \in [N]} \E_{\bar \nu^X}\left[\left\| y^i - \bar y^i \right\|^2 \right] \\
            &= -\alpha \left\| \by - \bar \by \right\|^2
        \end{align*}
        where the inequality above follows from the $\alpha$-strong monotonicity of $y \mapsto -\nabla_y V(X, \bar y)$.

        Combining the two bounds above gives:
        \begin{align*}
            \left\langle \bar b^{\bZ}(\bz) - \bar b^{\bZ}(\bar \bz), \bz-\bar \bz \right\rangle
            &\le -\alpha \left\| \bx - \bar \bx \right\|^2
            -\alpha \left\| \by - \bar \by \right\|^2
            \,=\, -\alpha \left\| \bz - \bar \bz \right\|^2
        \end{align*}
        which shows that $-\bar b^{\bZ}$ is $\alpha$-strongly monotone.
    \end{enumerate}
\end{proof}

%%%%%%%%%%%%%%%%%%%%%%%%
\subsection{Comparison Between the Vector Fields}
\label{Sec:ComparisonVectorFields}

%%%%%%%%%%%%%%%%%%%%%
\subsubsection{Comparison at stationary distribution}

\begin{lemma}\label{Lem:ComparisonStationary}
    Assume Assumption~\ref{As:SCSmooth}.
    Then the vector fields $b^{\bZ}$ defined in~\eqref{Eq:Defbz} and $\bar b^{\bZ}$ defined in~\eqref{Eq:Defbzbar} satisfy:
    \begin{align*}
        \E_{\bar \nu^{\bZ}}\left[\left\|\bar b^{\bZ}(\bar \bZ) - b^{\bZ}(\bar \bZ) \right\|^2\right]
        \le L^2 \, \Var_{\bar \nu^Z}(\bar Z).
    \end{align*}
\end{lemma}
\begin{proof}
    Let $\bar \bZ = (\bar \bX, \bar \bY) = (\bar X^1, \dots, \bar X^N, \bar Y^1, \dots, \bar Y^N) \sim \bar \nu^{\bZ} = (\bar \nu^X)^{\otimes N} \otimes (\bar \nu^Y)^{\otimes N}$, so all the random variables $\bar X^i \sim \bar \nu^X$ and $\bar Y^j \sim \bar \nu^Y$ are independent, for all $i,j \in [N]$.

    By definition, the quantity we wish to bound is:
    \begin{align*}
        \E_{\bar \nu^{\bZ}}\left[\left\|\bar b^{\bZ}(\bar \bZ) - b^{\bZ}(\bar \bZ) \right\|^2\right]
        = \E_{\bar \nu^{\bZ}}\left[\left\|\bar b^{\bX}(\bar \bX) - b^{\bX}(\bar \bX, \bar \bY) \right\|^2\right]
        + \E_{\bar \nu^{\bZ}}\left[\left\|\bar b^{\bY}(\bar \bY) - b^{\bY}(\bar \bX, \bar \bY) \right\|^2\right].
    \end{align*}
    We bound each term in the right-hand side separately.

    We will use the following formula for variance:
    If $U$ and $U'$ are independent random variables with the same distribution $\rho$, then
    $\Var_\rho(U) = \E_{\rho}[\|U-\E_\rho[U]\|^2] 
    = \frac{1}{2} \E_{\rho \otimes \rho}[\|U-U'\|^2].$
    
    We bound the first term above.
    We introduce an independent random variable $\bar Y \sim \bar \nu^Y$.
    Then:
    {\allowdisplaybreaks
    \begin{align*}
        \E_{\bar \nu^{\bZ}}\left[\left\|\bar b^{\bX}(\bar \bX) - b^{\bX}(\bar \bX, \bar \bY) \right\|^2\right]
        &\stackrel{(1)}{=} 
        \sum_{i \in [N]} 
        \E_{\bar \nu^{\bZ}}\left[\left\|\bar b^{X}(\bar X^i) - b^{X}(\bar X^i, \bar \bY) \right\|^2\right] \\
        &\stackrel{(2)}{=} \sum_{i \in [N]} 
        \E_{\bar \nu^{\bZ}}\left[\left\|
        \E_{\bar \nu^{Y}}\left[-\nabla_x V(\bar X^i, \bar Y) \right] + 
        \frac{1}{N} \sum_{j \in [N]} \nabla_x V(\bar X^i, \bar Y^j) \right\|^2\right] \\
        &\stackrel{(3)}{=} \frac{1}{N^2} \sum_{i \in [N]} 
        \E_{\bar \nu^{\bZ}}\left[\left\|
        \sum_{j \in [N]} \left(\nabla_x V(\bar X^i, \bar Y^j) - \E_{\bar \nu^{Y}}\left[\nabla_x V(\bar X^i, \bar Y) \right]  \right) \right\|^2\right] \\
        &\stackrel{(4)}{=} \frac{1}{N^2} \sum_{i \in [N]} \sum_{j \in [N]} 
        \E_{\bar \nu^{\bZ}}\left[\left\|
        \nabla_x V(\bar X^i, \bar Y^j) - \E_{\bar \nu^{Y}}\left[\nabla_x V(\bar X^i, \bar Y) \right] \right\|^2\right] \\
        &\stackrel{(5)}{=} \frac{1}{2N^2} \sum_{i \in [N]} \sum_{j \in [N]} 
        \E_{\bar \nu^{\bZ}} \E_{\bar \nu^{Y}}\left[\left\|
        \nabla_x V(\bar X^i, \bar Y^j) - \nabla_x V(\bar X^i, \bar Y) \right\|^2\right] \\
        &\stackrel{(6)}{\le} \frac{L^2}{2N^2} \sum_{i \in [N]} \sum_{j \in [N]} 
        \E_{\bar \nu^{\bZ}} \E_{\bar \nu^{Y}}\left[\left\|\bar Y^j -  \bar Y^j \right\|^2\right] \\
        &\stackrel{(7)}{=} \frac{L^2}{N^2} \sum_{i \in [N]} \sum_{j \in [N]} 
        \Var_{\bar \nu^{Y}}(\bar Y) \\
        &\stackrel{(8)}{=} L^2 \, \Var_{\bar \nu^Y}(\bar Y).
    \end{align*}
    }
    Above, in step~(1) we use the definitions of $b^{\bX}$ and $\bar b^{\bX}$ and split the squared norm across coordinates.
    In step (2), we use the definitions of $b^X$ and $\bar b^X$.
    In step (3), we pull the $1/N$ outside the square. 
    In step (4), we expand the square and note the cross terms are zero since all the random variables are independent;
    at this point, we recognize each term in the summation, conditioned on $\bar X^i$, is the variance of the random variable $\nabla_x V(\bar X^i, \bar Y^j)$ where $\bar Y^j \sim \bar \nu^Y$.
    In step (5), we use the variance formula by introducing an independent random variable $\bar Y \sim \bar \nu^Y$.
    In step (6), we use the property that $y \mapsto \nabla_x V(x,y)$ is $L$-Lipschitz, which follows from the smoothness of $V$ from Assumption~\ref{As:SCSmooth}.
    In step (7), we use the variance formula for the random variable $\bar Y \sim \bar \nu^Y$.
    In the last step (8), we collect the terms in the double summation which are all the same.

    By an identical argument, we can also bound the $Y$-component in the quantity above as:
    \begin{align*}
        \E_{\bar \nu^{\bZ}}\left[\left\|\bar b^{\bY}(\bar \bY) - b^{\bY}(\bar \bX, \bar \bY) \right\|^2\right] 
        \le L^2 \, \Var_{\bar \nu^X}(\bar X).
    \end{align*}

    Combining the two calculations above, we obtain the desired bound:
    \begin{align*}
        \E_{\bar \nu^{\bZ}}\left[\left\|\bar b^{\bZ}(\bar \bZ) - b^{\bZ}(\bar \bZ) \right\|^2\right]
        &\le L^2 \, \Var_{\bar \nu^Y}(\bar Y) + L^2 \, \Var_{\bar \nu^X}(\bar X) 
        = L^2 \, \Var_{\bar \nu^Z}(\bar Z)
    \end{align*}
    where the last step follows from the fact that $\bar \nu^Z = \bar \nu^X \otimes \bar \nu^Y$ is a product distribution.
\end{proof}

%%%%%%%%%%%%%%%%%%%%%
\subsubsection{Comparison at arbitrary distribution}

\begin{lemma}\label{Lem:ComparisonArbitrary}
    Assume Assumption~\ref{As:SCSmooth}.
    For any $\rho^{\bZ} \in \P(\R^{2dN})$, the vector fields $b^{\bZ}$ defined in~\eqref{Eq:Defbz} and $\bar b^{\bZ}$ defined in~\eqref{Eq:Defbzbar} satisfy:
    \begin{align*}
        \E_{\rho^{\bZ}}\left[\left\|\bar b^{\bZ}(\bZ) - b^{\bZ}(\bZ) \right\|^2\right]
        \le 2L^2 \, W_2(\rho^{\bZ}, \bar \nu^{\bZ})^2 + 4L^2 \, \Var_{\bar \nu^{Z}}(\bar Z).
    \end{align*}
\end{lemma}
\begin{proof}
    Let $\bZ = (\bX, \bY) = (X^1, \dots, X^N, Y^1, \dots, Y^N) \sim \rho^{\bZ}$.
    We first introduce some set up.
    
    Let $\sigma_1, \dots, \sigma_N \colon [N] \to [N]$ be a collection of permutations such that $\sigma_j(i) \neq \sigma_k(i)$ for all $j \neq k$ and $i \in [N]$.
    For example, we can take $\sigma_j(i) = i+j \pmod{N}$.
    
    We introduce new random variables $\bar \bZ^1, \dots, \bar \bZ^N \in \R^{2dN}$, where each $\bar \bZ^i = (\bar \bX^i, \bar \bY^i) = (\bar X^{i,1}, \dots, \bar X^{i,N}, \bar Y^{i,1}, \dots, \bar Y^{i,N})$ with $\bar X^{i,j}, \bar Y^{i,j} \in \R^d$ for each $i,j \in [N]$, with the following structure of joint distribution:
    \begin{itemize}
        \item For each $i \in [N]$, $\bar \bZ^i \sim \bar \nu^{\bZ}$ marginally, and $(\bZ, \bar \bZ^i)$ is jointly distributed as the optimal $W_2$ coupling between $\rho^{\bZ}$ and $\bar \nu^{\bZ}$, so $\E[\|\bZ-\bar \bZ^i\|^2] = W_2(\rho^{\bZ}, \bar \nu^{\bZ})^2$.
        \item The random variables $\bar \bZ^1, \dots, \bar \bZ^N$ are pairwise independent conditioned on $\bZ$.
        In particular, this implies that for all $i,j,k \in [N]$ with $j \neq k$:
        (1) $\bar Y^{\sigma_j(i),j}$ and $\bar Y^{\sigma_k(i),k}$ are independent when conditioned on $\bZ$, since $\sigma_j(i) \neq \sigma_k(i)$; and similarly,
        (2) $\bar X^{\sigma_j(i),j}$ and $\bar X^{\sigma_k(i),k}$ are independent when conditioned on $\bZ$.
    \end{itemize}

    By definition, the quantity we wish to bound is:
    \begin{align*}
        \E_{\rho^{\bZ}}\left[\left\|\bar b^{\bZ}(\bZ) - b^{\bZ}(\bZ) \right\|^2\right]
        = \E_{\rho^{\bZ}}\left[\left\|\bar b^{\bX}(\bX) - b^{\bX}(\bX, \bY) \right\|^2\right]
        + \E_{\rho^{\bZ}}\left[\left\|\bar b^{\bY}(\bY) - b^{\bY}(\bX, \bY) \right\|^2\right].
    \end{align*}
    We will bound each term in the right-hand side separately.    

    We bound the first term above.
    Below, we write $\E$ to denote the expectation over the collective joint distribution of the random variables $(\bZ, \bZ^1,\dots, \bZ^N)$ introduced above.
    We can bound:
    \begin{align}
        \E_{\rho^{\bZ}}&\left[\left\|\bar b^{\bX}(\bX) - b^{\bX}(\bX, \bY) \right\|^2\right] 
        = \sum_{i \in [N]} \E_{\rho^{\bZ}}\left[ \left\|\bar b^{X}(X^i) - b^{X}(X^i, \bY)\right\|^2 \right] \notag \\
        &\qquad = \sum_{i \in [N]} \E_{\rho^{\bZ}}\left[ \left\|\E_{\bar Y \sim \bar \nu^Y}[-\nabla_x V(X^i, \bar Y)] + \frac{1}{N} \sum_{j \in [N]} \nabla_x V(X^i, Y^j)\right\|^2 \right] \notag \\
        &\qquad = \frac{1}{N^2} \sum_{i \in [N]} \E_{\rho^{\bZ}}\left[ \left\|\sum_{j \in [N]} \left(\nabla_x V(X^i, Y^j) + \E_{\bar Y \sim \bar \nu^Y}[-\nabla_x V(X^i, \bar Y)] \right)\right\|^2 \right] \notag \\
        &\qquad \le \frac{2}{N^2} \sum_{i \in [N]} \E\left[\left\|\sum_{j \in [N]} \left(\nabla_x V(X^i, Y^j) - \nabla_x V(X^i, \bar Y^{\sigma_j(i),j})\right)\right\|^2 \right]  \notag \\
        &\qquad \quad + \frac{2}{N^2} \sum_{i \in [N]} \E\left[\left\|\sum_{j \in [N]} \left(\nabla_x V(X^i, \bar Y^{\sigma_j(i),j}) - \E_{\bar Y \sim \bar \nu^{Y}}[\nabla_x V(X^i, \bar Y)] \right)\right\|^2 \right]  \notag \\
        &\qquad =: \Err_{1}^X + \Err_{2}^X. \label{Eq:CalcX}
    \end{align}
    In the above, we have used the inequality $\|a+b\|^2 \le 2\|a\|^2 + 2\|b\|^2$.
    
    We can bound the first error term in~\eqref{Eq:CalcX} by:
    {\allowdisplaybreaks
    \begin{align*}
        \Err_{1}^X
        &:= \frac{2}{N^2} \sum_{i \in [N]} \E\left[\left\|\sum_{j \in [N]} \left(\nabla_x V(X^i, Y^j) - \nabla_x V(X^i, \bar Y^{\sigma_j(i),j})\right)\right\|^2 \right] \\
        &\le \frac{2}{N} \sum_{i \in [N]} \sum_{j \in [N]} \E\left[\left\|\nabla_x V(X^i, Y^j) - \nabla_x V(X^i, \bar Y^{\sigma_j(i),j})\right\|^2 \right] \\
        &\le \frac{2L^2}{N} \sum_{i \in [N]} \sum_{j \in [N]} \E\left[\left\|Y^j-\bar Y^{\sigma_j(i),j}\right\|^2 \right] \\
        &= \frac{2L^2}{N} \sum_{i \in [N]} \sum_{j \in [N]} \E\left[\left\|Y^j-\bar Y^{i,j}\right\|^2 \right] \\
        &= \frac{2L^2}{N} \sum_{i \in [N]} \E\left[\left\|\bY-\bar \bY^{i}\right\|^2 \right].
    \end{align*}
    }
    In the above, the first inequality is by Cauchy-Schwarz ($\|\sum_{j \in [N]} a_j\|^2 \le N \sum_{j \in [N]} \|a_j\|^2$);
    the next inequality is by the smoothness assumption on $V$;
    the next step is by rearranging the permutation;
    and in the last step we rewrite the expression in terms of the vector form.
    
    We can bound the second error term in~\eqref{Eq:CalcX} by:
    \begin{align*}
        \Err_{2}^X
        &:= \frac{2}{N^2} \sum_{i \in [N]} \E\left[\left\|\sum_{j \in [N]} \left(\nabla_x V(X^i, \bar Y^{\sigma_j(i),j}) - \E_{\bar Y \sim \bar \nu^{Y}}[\nabla_x V(X^i, \bar Y)] \right)\right\|^2 \right] \\
        &\stackrel{(1)}{=} \frac{2}{N^2} \sum_{i \in [N]} \sum_{j \in [N]} \E\left[\left\|\nabla_x V(X^i, \bar Y^{\sigma_j(i),j}) - \E_{\bar Y \sim \bar \nu^{Y}}[\nabla_x V(X^i, \bar Y)] \right\|^2 \right] \\
        &\stackrel{(2)}{\le} \frac{2}{N^2} \sum_{i \in [N]} \sum_{j \in [N]} \E \, \E_{\bar Y \sim \bar \nu^Y} \left[\left\|\nabla_x V(X^i, \bar Y^{\sigma_j(i),j}) - \nabla_x V(X^i, \bar Y)\right\|^2 \right]  \\
        &\stackrel{(3)}{\le} \frac{2L^2}{N^2} \sum_{i \in [N]} \sum_{j \in [N]} \E \, \E_{\bar Y \sim \bar \nu^Y} \left[\left\|\bar Y^{\sigma_j(i),j} - \bar Y\right\|^2 \right]  \\
        &\stackrel{(4)}{=} \frac{4L^2}{N^2} \sum_{i \in [N]} \sum_{j \in [N]} \Var_{\bar \nu^{Y}}(\bar Y) \\
        &\stackrel{(5)}{=} 4L^2 \, \Var_{\bar \nu^{Y}}(\bar Y).
    \end{align*}
    In the above, step (1) follows by expanding the square and noting the cross terms are zero; this follows from our construction, since for each $i$ and for each $j \neq k$, when conditioned on $\bZ$ (which includes $X^i$), the random variables $\bar Y^{\sigma_j(i),j}$ and $\bar Y^{\sigma_k(i),k}$ are independent (recall $\sigma_j(i) \neq \sigma_k(i)$), and each term in the summation has mean $0$.
    The next step (2) follows by introducing an independent random variable $\bar Y \sim \bar \nu^Y$ and using Cauchy-Schwarz inequality to pull the expectation outside the square.
    The next step (3) follows from the smoothness assumption on $V$ from Assumption~\ref{As:SCSmooth}.
    The next step (4) follows by recognizing each term in the double summation is a variance of $\bar Y \sim \bar \nu^Y$ (recall $\bar Y^{\sigma_j(i),j} \sim \bar \nu^Y$ also and is independent of $\bar Y$).
    The last step (5) is by collecting the terms in the double summation which are all equal.

    Plugging in the two calculations above to~\eqref{Eq:CalcX}, we can control the first term in the quantity we wish to bound as:
    \begin{align*}
        \E_{\rho^{\bZ}}\left[\left\|\bar b^{\bX}(\bX) - b^{\bX}(\bX, \bY) \right\|^2\right]
        &\le \frac{2L^2}{N} \sum_{i \in [N]} \E\left[\left\|\bY-\bar \bY^{i}\right\|^2 \right] + 4L^2 \, \Var_{\bar \nu^{Y}}(\bar Y).
    \end{align*}

    By an identical argument, we can control the second term in the quantity we wish to bound as:
    \begin{align*}
        \E_{\rho^{\bZ}}\left[\left\|\bar b^{\bY}(\bY) - b^{\bY}(\bX, \bY) \right\|^2\right]
        &\le \frac{2L^2}{N} \sum_{i \in [N]} \E\left[\left\|\bX-\bar \bX^{i}\right\|^2 \right] + 4L^2 \, \Var_{\bar \nu^{X}}(\bar X).
    \end{align*}

    Combining the two bounds above, and recalling that each $(\bZ, \bar \bZ^i)$ has the optimal $W_2$ coupling, we obtain:
    \begin{align*}
        \E_{\rho^{\bZ}}\left[\left\|\bar b^{\bZ}(\bZ) - b^{\bZ}(\bZ) \right\|^2\right]
        &\le \frac{2L^2}{N} \sum_{i \in [N]} \E\left[\left\|\bY-\bar \bY^{i}\right\|^2 \right] + 4L^2 \, \Var_{\bar \nu^{Y}}(\bar Y) \\
        &\quad  + \frac{2L^2}{N} \sum_{i \in [N]} \E\left[\left\|\bX-\bar \bX^{i}\right\|^2 \right] + 4L^2 \, \Var_{\bar \nu^{X}}(\bar X) \\
        &= \frac{2L^2}{N} \sum_{i \in [N]} \E\left[\left\|\bZ-\bar \bZ^{i}\right\|^2 \right] + 4L^2 \, \Var_{\bar \nu^{Z}}(\bar Z) \\
        &= \frac{2L^2}{N} \sum_{i \in [N]} W_2(\rho^{\bZ}, \bar \nu^{\bZ})^2 + 4L^2 \, \Var_{\bar \nu^{Z}}(\bar Z) \\
        &= 2L^2 W_2(\rho^{\bZ}, \bar \nu^{\bZ})^2 + 4L^2 \, \Var_{\bar \nu^{Z}}(\bar Z)
    \end{align*}
    which is the desired bound.
\end{proof}

%%%%%%%%%%%%%%%%%%%%%%%%%%%%
\subsection{Bounds on the Vector Fields}

%%%%%%%%%%%%%%%%%%%%%%%%%%%%
\subsubsection{Bound on the mean-field vector field under mean-field distribution}

We bound the magnitude of the mean-field vector field under the stationary mean-field distribution, which is proportional to the Fisher information.

\begin{lemma}\label{Lem:BoundFisher}
    Assume Assumption~\ref{As:SCSmooth}.
    Then for the mean-field vector field $\bar b^{\bZ}$ defined in~\eqref{Eq:Defbzbar}, under the tensorized stationary mean-field distribution $\bar \nu^{\bZ}$ defined in~\eqref{Eq:TensorMF}:
    \begin{align*}
        \E_{\bar \nu^{\bZ}}\left[\left\|\bar b^{\bZ}(\bar \bZ)\right\|^2\right] \le 2\reg dLN.
    \end{align*}
\end{lemma}
\begin{proof}
    Recall by construction, as stated in~\eqref{Eq:TensorMFScore}, that the mean-field vector field $\bar b^{\bZ}$ is a scaled score function of the tensorized stationary mean-field distribution $\bar \nu^{\bZ}$, i.e.,
    $\bar b^{\bZ}(\bz) = \reg \nabla \log \bar \nu^{\bZ}(\bz)$.    
    Note from the definition~\eqref{Eq:BestResponse}, since $V$ is $L$-smooth, $\log \bar \nu^X$ and $\log \bar \nu^Y$ are $(L/\reg)$-smooth, and thus $\log \bar \nu^{\bZ}$ is also $(L/\reg)$-smooth, i.e., 
    $\left\|\nabla^2 \log \bar \nu^{\bZ}(\bz)\right\|_\op \le L/\reg$ for all $\bz \in \R^{2dN}$.
    In particular, 
    \begin{align*}
        \Delta \log \bar \nu^{\bZ}(\bz) = \Tr\left(\nabla^2 \log \bar \nu^{\bZ}(\bz)\right) \le \frac{L}{\reg} \cdot 2dN.
    \end{align*}
    Using integration by parts, where the boundary term is $0$ since $\bar \nu^{\bZ}(\bz) \nabla \log \bar \nu^{\bZ}(\bz) \to 0$ as $\|\bz\| \to \infty$:
    \begin{align*}
        \E_{\bar \nu^{\bZ}}\left[\left\|\bar b^{\bZ}(\bar \bZ)\right\|^2\right]
        &= \reg^2 \, \E_{\bar \nu^{\bZ}}\left[\left\|\nabla \log \bar \nu^{\bZ}(\bar \bZ)\right\|^2\right] \\
        &= \reg^2 \int_{\R^{2dN}} \bar \nu^{\bZ}(\bz) \left\langle \nabla \log \bar \nu^{\bZ}(\bz), \, \nabla \log \bar \nu^{\bZ}(\bz) \right\rangle \, d\bz \\
        &= -\reg^2 \int_{\R^{2dN}} \left\langle \nabla \bar \nu^{\bZ}(\bz), \, \nabla \log \bar \nu^{\bZ}(\bz) \right\rangle \, d\bz \\
        &= \reg^2 \int_{\R^{2dN}} \bar \nu^{\bZ}(\bz) \, \Delta \log \bar \nu^{\bZ}(\bz) \, d\bz \\
        &\le 2\reg dLN.
    \end{align*}
\end{proof}

%%%%%%%%%%%%%%%%%%%%%%%%%%%%
\subsubsection{Bound on the finite-particle vector field under mean-field distribution}

\begin{lemma}\label{Lem:BoundFPMF}
    Assume Assumption~\ref{As:SCSmooth}.
    Then for the finite-particle vector field $b^{\bZ}$ defined in~\eqref{Eq:Defbz}, and for the tensorized stationary mean-field distribution~\eqref{Eq:TensorMF}:
    \begin{align*}
        \E_{\bar \nu^{\bZ}}\left[\left\|b^{\bZ}(\bar \bZ)\right\|^2\right] \le 2L^2 \, \Var_{\bar \nu^Z}(\bar Z) + 4\reg dLN.
    \end{align*}
\end{lemma}
\begin{proof}
    We can bound:
    \begin{align*}
        \E_{\bar \nu^{\bZ}}\left[\left\|b^{\bZ}(\bar \bZ) \right\|^2 \right]
        &= \E_{\bar \nu^{\bZ}}\left[\left\|b^{\bZ}(\bar \bZ) - \bar b^{\bZ}(\bar \bZ) + \bar b^{\bZ}(\bar \bZ) \right\|^2\right] \\
        &\le 2\E_{\bar \nu^{\bZ}}\left[\left\|b^{\bZ}(\bar \bZ) - \bar b^{\bZ}(\bar \bZ) \right\|^2\right]
        + 2\E_{\bar \nu^{\bZ}} \left[\left\|\bar b^{\bZ}(\bar \bZ) \right\|^2 \right] \\
        &\le 2L^2 \, \Var_{\bar \nu^Z}(\bar Z) + 4\reg dLN.
    \end{align*}
    In the above, we have introduced an intermediate term $\bar b^{\bZ}(\bar \bZ)$, used the inequality $\|a+b\|^2 \le 2\|a\|^2 + 2\|b\|^2$, and used the results from Lemma~\ref{Lem:ComparisonStationary} and Lemma~\ref{Lem:BoundFisher}.
\end{proof}

%%%%%%%%%%%%%%%%%%%%%%%%%%%%
\subsubsection{Bound on the finite-particle vector field under arbitrary distribution}

\begin{lemma}\label{Lem:BoundFPArb}
    Assume Assumption~\ref{As:SCSmooth}.
    Then for the finite-particle vector field $b^{\bZ}$ defined in~\eqref{Eq:Defbz}, and for any distribution $\rho^{\bZ} \in \P(\R^{2dN})$:
    \begin{align*}
        \E_{\rho^{\bZ}}\left[\left\|b^{\bZ}(\bZ)\right\|^2\right] \le 8L^2 \, W_2(\rho^{\bZ}, \bar \nu^{\bZ})^2 + 4L^2 \, \Var_{\bar \nu^Z}(\bar Z) + 8\reg dLN.
    \end{align*}
\end{lemma}
\begin{proof}
    Let $\bZ \sim \rho^{\bZ}$ and $\bar \bZ \sim \bar \nu^{\bZ}$ such that $(\bZ,\bar \bZ)$ has the optimal $W_2$ coupling between $\rho^{\bZ}$ and $\bar \nu^{\bZ}$.
    We use $\E$ to denote the expectation over this joint coupling.
    Then we can bound:
    \begin{align*}
        \E_{\rho^{\bZ}}\left[\left\|b^{\bZ}(\bZ) \right\|^2 \right]
        &= \E\left[\left\|b^{\bZ}(\bZ) - b^{\bZ}(\bar \bZ) + b^{\bZ}(\bar \bZ) \right\|^2\right] \\
        &\le 2 \E\left[\left\|b^{\bZ}(\bZ) - b^{\bZ}(\bar \bZ) \right\|^2 \right] + 2\E_{\bar \nu^{\bZ}}\left[\left\|b^{\bZ}(\bar \bZ) \right\|^2 \right] \\
        &\le 8L^2 \, \E\left[\left\|\bZ-\bar \bZ \right\|^2 \right] + 2 \left(2L^2 \, \Var_{\bar \nu^Z}(\bar Z) + 4\reg dLN\right) \\
        &= 8L^2 \, W_2(\rho^{\bZ}, \bar \nu^{\bZ})^2 + 4L^2 \, \Var_{\bar \nu^Z}(\bar Z) + 8\reg dLN.
    \end{align*}
    In the above, we introduce an additional term $b^{\bZ}(\bar \bZ)$ with the coupling defined above, and use the inequality $\|a+b\|^2 \le 2\|a\|^2 + 2\|b\|^2$.
    In the next step, we use the property that $b^{\bZ}$ is $(2L)$-Lipschitz from Lemma~\ref{Lem:LipschitzMonotone_bz}, and the bound from Lemma~\ref{Lem:BoundFPMF}.
    In the last step, we use the fact that $(\bZ,\bar \bZ)$ has the optimal $W_2$ coupling between $\rho^{\bZ}$ and $\bar \nu^{\bZ}$.
\end{proof}

%%%%%%%%%%%%%%%%%%%
\subsection{Time Derivative of KL Divergence Along Fokker-Planck Equations}

We have the following formula on the time derivative of KL divergence of two distributions which evolve following their Fokker-Planck equations, which can have different drift terms, but with the same diffusion term.
This is a classical formula that has been used in many previous works, including for analyzing stochastic interpolants~\cite[Lemma~2.2]{albergo2023stochastic} and for showing the propagation of chaos in interacting particle systems~\cite[Lemma~3.1]{lacker2023sharp}.

\begin{lemma}\label{Lem:ddtKLDiv}
    Suppose $(\rho_t)_{t \ge 0}$ and $(\bar \rho_t)_{t \ge 0}$ are probability distributions in $\P(\R^D)$ which evolve following the Fokker-Planck equations:
    \begin{align*}
        \part{\rho_t}{t} &= -\nabla \cdot (\rho_t b_t) + c \Delta \rho_t \\
        \part{\bar \rho_t}{t} &= -\nabla \cdot (\bar \rho_t \bar b_t) + c \Delta \bar \rho_t
    \end{align*}
    for some time-dependent vector fields $b_t, \bar b_t \colon \R^D \to \R^D$, and for some constant $c \ge 0$.
    Then:
    \begin{align*}
        \frac{d}{dt} \KL(\rho_t \,\|\, \bar \rho_t) = -c \, \FI(\rho_t \,\|\, \bar \rho_t) + \E_{\rho_t}\left[\left\langle \nabla \log \frac{\rho_t}{\bar \rho_t}, \, b_t-\bar b_t \right\rangle \right].
    \end{align*}
\end{lemma}
\begin{proof}
    We can compute by differentiating under the integral sign and using chain rule:
    \begin{align}
        \frac{d}{dt} \KL(\rho_t \,\|\, \bar \rho_t)
        &= \frac{d}{dt} \int_{\R^D} \rho_t \log \frac{\rho_t}{\bar \rho_t} \, dx \notag \\
        &= \int_{\R^D} \part{\rho_t}{t} \log \frac{\rho_t}{\bar \rho_t} \, dx + \int_{\R^D} \rho_t \frac{1}{\rho_t} \part{\rho_t}{t} \, dx - \int_{\R^D} \rho_t \frac{1}{\bar \rho_t} \part{\bar \rho_t}{t}. \label{Eq:ddtCalc1}
    \end{align}
    We compute each term above.
    For the first term in~\eqref{Eq:ddtCalc1}, 
    by the Fokker-Planck equation and  using integration by parts, where all the boundary terms vanish:
    \begin{align*}
        \int_{\R^D} \part{\rho_t}{t} \log \frac{\rho_t}{\bar \rho_t} \, dx
        &= \int_{\R^D} \left(-\nabla \cdot (\rho_t b_t) + c \Delta \rho_t\right) \log \frac{\rho_t}{\bar \rho_t} \, dx \\
        &= \int_{\R^D} \rho_t \left\langle b_t, \nabla \log \frac{\rho_t}{\bar \rho_t} \right\rangle \, dx
        - c\int_{\R^D} \rho_t \left\langle \nabla \log \rho_t, \nabla \log \frac{\rho_t}{\bar \rho_t} \right\rangle \, dx
    \end{align*}
    where for the second term above we have used the identity that $\Delta \rho_t = \nabla \cdot \nabla \rho_t = \nabla \cdot (\rho_t \nabla \log \rho_t)$.

    For the second term in~\eqref{Eq:ddtCalc1}, we can show it is equal to $0$:
    \begin{align*}
        \int_{\R^D} \rho_t \frac{1}{\rho_t} \part{\rho_t}{t} \, dx
        &= \int_{\R^D} \part{\rho_t}{t} \, dx
        \,=\, \part{}{t} \int_{\R^D} \rho_t \, dx
        \,=\, \part{}{t} (1)
        \,=\, 0.
    \end{align*}

    For the third term in~\eqref{Eq:ddtCalc1}, 
    by the Fokker-Planck equation and  using integration by parts:
    \begin{align*}
        - \int_{\R^D} \rho_t \frac{1}{\bar \rho_t} \part{\bar \rho_t}{t}
        &= -\int_{\R^D} \left(-\nabla \cdot (\bar \rho_t \bar b_t) + c \Delta \bar \rho_t\right) \frac{\rho_t}{\bar \rho_t} \, dx \\
        &= - \int_{\R^D} \bar \rho_t \left\langle \bar b_t, \nabla \frac{\rho_t}{\bar \rho_t} \right\rangle \, dx
        + c \int_{\R^D} \bar \rho_t \left\langle \nabla \log \bar \rho_t, \nabla \frac{\rho_t}{\bar \rho_t} \right\rangle \, dx \\
        &= - \int_{\R^D} \rho_t \left\langle \bar b_t, \nabla \log \frac{\rho_t}{\bar \rho_t} \right\rangle \, dx
        + c \int_{\R^D} \rho_t \left\langle \nabla \log \bar \rho_t, \nabla \log \frac{\rho_t}{\bar \rho_t} \right\rangle \, dx
    \end{align*}
    where in the above we have used the identity that $\Delta \bar \rho_t = \nabla \cdot (\bar \rho_t \nabla \log \bar \rho_t)$,
    and $\bar \rho_t \nabla \frac{\rho_t}{\bar \rho_t} = \rho_t \nabla \log \frac{\rho_t}{\bar \rho_t}$.

    Combining the three terms above in~\eqref{Eq:ddtCalc1}, we obtain:
    \begin{align*}
        \frac{d}{dt} \KL(\rho_t \,\|\, \bar \rho_t)
        &= \int_{\R^D} \rho_t \left\langle b_t, \nabla \log \frac{\rho_t}{\bar \rho_t} \right\rangle \, dx
        - c\int_{\R^D} \rho_t \left\langle \nabla \log \rho_t, \nabla \log \frac{\rho_t}{\bar \rho_t} \right\rangle \, dx \\
        &\qquad - \int_{\R^D} \rho_t \left\langle \bar b_t, \nabla \log \frac{\rho_t}{\bar \rho_t} \right\rangle \, dx
        + c \int_{\R^D} \rho_t \left\langle \nabla \log \bar \rho_t, \nabla \log \frac{\rho_t}{\bar \rho_t} \right\rangle \, dx \\
        &= \int_{\R^D} \rho_t \left\langle b_t-\bar b_t, \nabla \log \frac{\rho_t}{\bar \rho_t} \right\rangle \, dx
        - c\int_{\R^D} \rho_t \left\| \nabla \log \frac{\rho_t}{\bar \rho_t} \right\|^2 \, dx  \\
        &= \E_{\rho_t}\left[\left\langle b_t-\bar b_t, \nabla \log \frac{\rho_t}{\bar \rho_t} \right\rangle \right]
        - c \, \FI(\rho_t \,\|\, \bar \rho_t)
    \end{align*}
    as desired.
\end{proof}

%%%%%%%%%%%%%%%%%%%%%%
\section{Proofs for the Finite-Particle Dynamics}

%%%%%%%%%%%%%%%%%%
\subsection{Proof of Theorem~\ref{Thm:ConvergenceDynToMF} (Biased Convergence of the Finite-Particle Dynamics)}
\label{Sec:ConvergenceDynToMFProof}

\begin{proof}[Proof of Theorem~\ref{Thm:ConvergenceDynToMF}]
    \textbf{(1)~Biased $W_2$ convergence bound:} We use the synchronous coupling technique.
    Concretely, we consider $(\bZ_t)_{t \ge 0}$ and $(\bar \bZ_t)_{t \ge 0}$ in $\R^{2dN}$, where $\bZ_t \sim \rho_t^{\bZ}$ evolves following the finite-particle dynamics~\eqref{Eq:ParticleSystem3}:
    \begin{align*}
        d\bZ_t = b^{\bZ}(\bZ_t) \, dt + \sqrt{2\reg} \, dW_t^{\bZ}
    \end{align*}
    and where $\bar \bZ_t \sim \bar \rho_t^{\bZ} = \bar \nu^{\bZ}$ evolves following the stationary tensorized mean-field dynamics~\eqref{Eq:ParticleSystem3}:
    \begin{align*}
        d\bar \bZ_t = \bar b^{\bZ}(\bar \bZ_t) \, dt + \sqrt{2\reg} \, dW_t^{\bZ}
    \end{align*}
    where we start from the stationary distribution $\bar \rho_0^{\bZ} = \bar \nu^{\bZ}$, so $\bar \rho_t^{\bZ} = \bar \nu^{\bZ}$ for all $t \ge 0$.
    Suppose we run the two stochastic processes above using the same standard Brownian motion $(W_t^{\bZ})_{t \ge 0}$ in $\R^{2dN}$.
    Furthermore, suppose we start the two processes above from $(\bZ_0, \bar \bZ_0)$ which has a joint distribution which is the optimal $W_2$ coupling between $\rho_0^{\bZ}$ and $\bar \rho_0^{\bZ} = \bar \nu^{\bZ}$.

    With the set up above, the difference $\bZ_t-\bar \bZ_t$ evolves via:
    \begin{align*}
        d(\bZ_t-\bar \bZ_t) = \left(b^{\bZ}(\bZ_t)-\bar b^{\bZ}(\bar \bZ_t)\right) dt
    \end{align*}
    where the Brownian motion terms cancel because they are equal in the two processes.
    Then we get:
    \begin{align*}
        d\left\|\bZ_t-\bar \bZ_t\right\|^2 = 2\left\langle \bZ_t-\bar \bZ_t, \, b^{\bZ}(\bZ_t)-\bar b^{\bZ}(\bar \bZ_t) \right \rangle dt.
    \end{align*}
    Taking expectation over the joint distribution, we obtain:
    \begin{align*}
        \frac{d}{dt} \E\left[ \left\|\bZ_t-\bar \bZ_t\right\|^2 \right]
        &= 2 \E\left[\left\langle \bZ_t-\bar \bZ_t, \, b^{\bZ}(\bZ_t)-\bar b^{\bZ}(\bar \bZ_t) \right \rangle \right] \\
        &= 2 \E\left[\left\langle \bZ_t-\bar \bZ_t, b^{\bZ}(\bZ_t)- b^{\bZ}(\bar \bZ_t) \right \rangle \right] 
        + 2 \E\left[\left\langle \bZ_t-\bar \bZ_t, b^{\bZ}(\bar \bZ_t)-\bar b^{\bZ}(\bar \bZ_t) \right \rangle \right],
    \end{align*}
    where in the last step we have introduced an intermediate term $\E\left[\left\langle \bZ_t-\bar \bZ_t, b^{\bZ}(\bar \bZ_t) \right \rangle \right]$.
    
    We can bound the two terms above separately as follows.
    For the first term, recall from Lemma~\ref{Lem:LipschitzMonotone_bz} that $-b^{\bZ}$ is $\alpha$-strongly monotone, so we can bound:
    \begin{align*}
        \E\left[\left\langle \bZ_t-\bar \bZ_t, \, b^{\bZ}(\bZ_t)- b^{\bZ}(\bar \bZ_t) \right \rangle \right] 
        \le -\alpha \, \E\left[\left\|\bZ_t-\bar \bZ_t \right \|^2 \right].
    \end{align*}
    For the second term, we use the inequality $\langle a,b \rangle \le \frac{\alpha}{4} \|a\|^2 + \frac{1}{\alpha} \|b\|^2$, and apply the bound from Lemma~\ref{Lem:ComparisonStationary} to get:
    \begin{align*}
        \E\left[\left\langle \bZ_t-\bar \bZ_t, b^{\bZ}(\bar \bZ_t)-\bar b^{\bZ}(\bar \bZ_t) \right \rangle \right]
        &\le \frac{\alpha}{4} \, \E\left[\left\|\bZ_t-\bar \bZ_t \right \|^2 \right] + \frac{1}{\alpha} \E_{\bar \nu^{\bZ}}\left[\left\|b^{\bZ}(\bar \bZ_t)-\bar b^{\bZ}(\bar \bZ_t) \right \|^2 \right] \\
        &\le \frac{\alpha}{4} \, \E\left[\left\|\bZ_t-\bar \bZ_t \right \|^2 \right] + \frac{L^2}{\alpha} \Var_{\bar \nu^{Z}}(\bar Z).
    \end{align*}
    Plugging in these two bounds to the computation above, we obtain:
    \begin{align*}
        \frac{d}{dt} \E\left[ \left\|\bZ_t-\bar \bZ_t\right\|^2 \right]
        &\le -\frac{3\alpha}{2} \, \E\left[\left\|\bZ_t-\bar \bZ_t \right \|^2 \right] + \frac{2L^2}{\alpha} \Var_{\bar \nu^{Z}}(\bar Z).
    \end{align*}    
    This is equivalent to:
    \begin{align*}
        \frac{d}{dt}\left( e^{\frac{3}{2}\alpha t} \, \E\left[ \left\|\bZ_t-\bar \bZ_t\right\|^2 \right] \right)
        &\le e^{\frac{3}{2}\alpha t} \,\frac{2L^2}{\alpha} \Var_{\bar \nu^{Z}}(\bar Z).
    \end{align*}
    Integrating from $0$ to $t$ and rearranging yields:
    \begin{align*}
        \E\left[ \left\|\bZ_t-\bar \bZ_t\right\|^2 \right]
        &\le e^{-\frac{3}{2}\alpha t} \, \E\left[ \left\|\bZ_0-\bar \bZ_0\right\|^2 \right] + \left(\frac{1-e^{-\frac{3}{2}\alpha t}}{\frac{3}{2}\alpha} \right) \,\frac{2L^2}{\alpha} \Var_{\bar \nu^{Z}}(\bar Z) \\
        &\le e^{-\frac{3}{2}\alpha t} \, \E\left[ \left\|\bZ_0-\bar \bZ_0\right\|^2 \right] + \,\frac{2L^2}{\alpha^2} \Var_{\bar \nu^{Z}}(\bar Z) \\
        &= e^{-\frac{3}{2}\alpha t} \, W_2(\rho_0^{\bZ}, \bar \nu^{\bZ})^2 + \,\frac{2L^2}{\alpha^2} \Var_{\bar \nu^{Z}}(\bar Z)
    \end{align*}
    where in the second inequality above we drop the $-e^{-\frac{3}{2}\alpha t}$ term and further bound $\frac{4}{3} \le 2$,
    and in the last step we use the fact that $(\bZ_0, \bar \bZ_0)$ has the optimal $W_2$ coupling between $\rho_0^{\bZ}$ and $\bar \nu^{\bZ}$.
    Finally, since $W_2$ distance is the infimum over all coupling, from the above inequality we conclude:
    \begin{align}\label{Eq:CalcW2Bound}
        W_2(\rho_t^{\bZ}, \bar \nu^{\bZ})^2
        &\le e^{-\frac{3}{2}\alpha t} \, W_2(\rho_0^{\bZ}, \bar \nu^{\bZ})^2 + \,\frac{2L^2}{\alpha^2} \Var_{\bar \nu^{Z}}(\bar Z) 
    \end{align}
    as desired.

    \medskip
    \noindent
    \textbf{(2)~Biased convergence in KL divergence:}
    We consider $\bZ_t \sim \rho_t^{\bZ}$ which evolves following the finite-particle dynamics~\eqref{Eq:ParticleSystem3} in $\R^{2dN}$, so $\rho_t^{\bZ}$ evolves following the Fokker-Planck equation:
    \begin{align*}
        \part{\rho_t^{\bZ}}{t} = -\nabla \cdot \left(\rho_t^{\bZ} \, b^{\bZ} \right) + \reg \Delta \rho_t^{\bZ}.
    \end{align*}
    We also consider $\bar \bZ_t \sim \bar \rho_t^{\bZ} = \bar \nu^{\bZ}$ which evolves following the stationary tensorized mean-field dynamics~\eqref{Eq:MFSystemTensor} from $\bar \bZ_0 \sim \bar \rho_0^{\bZ} = \bar \nu^{\bZ}$, so $\bar \rho_t^{\bZ}$ satisfies the Fokker-Planck equation:
    \begin{align*}
        \part{\bar \rho_t^{\bZ}}{t} = -\nabla \cdot \left(\bar \rho_t^{\bZ} \, \bar b^{\bZ} \right) + \reg \Delta \bar \rho_t^{\bZ}.
    \end{align*}
    (Note since $\bar \rho_t^{\bZ} = \bar \nu^{\bZ}$, both sides of the Fokker-Planck equation above are in fact equal to $0$.)

    Then using the formula from Lemma~\ref{Lem:ddtKLDiv}, we can compute:
    \begin{align*}
        \frac{d}{dt} \KL\left(\rho_t^{\bZ} \,\|\, \bar \rho_t^{\bZ}\right)
        &\stackrel{(1)}{=} -\reg \, \FI\left(\rho_t^{\bZ} \,\|\, \bar \rho_t^{\bZ}\right) + \E_{\rho_t^{\bZ}}\left[\left\langle \nabla \log \frac{\rho_t^{\bZ}}{\bar \rho_t^{\bZ}}, \, b^{\bZ} - \bar b^{\bZ} \right\rangle \right] \\
        &\stackrel{(2)}{\le} -\frac{\reg}{2} \, \FI\left(\rho_t^{\bZ} \,\|\, \bar \rho_t^{\bZ}\right) + \frac{1}{2\reg} \, \E_{\rho_t^{\bZ}}\left[\left\|b^{\bZ} - \bar b^{\bZ} \right\|^2 \right] \\
        &\stackrel{(3)}{\le} -\alpha \, \KL\left(\rho_t^{\bZ} \,\|\, \bar \rho_t^{\bZ}\right) + \frac{1}{2\reg} \, \E_{\rho_t^{\bZ}}\left[\left\|b^{\bZ} - \bar b^{\bZ} \right\|^2 \right] \\
        &\stackrel{(4)}{\le} -\alpha \, \KL\left(\rho_t^{\bZ} \,\|\, \bar \rho_t^{\bZ}\right) + 
        \frac{L^2}{\reg} \, W_2(\rho_t^{\bZ}, \bar \nu^{\bZ})^2 + \frac{2L^2}{\reg} \, \Var_{\bar \nu^{Z}}(\bar Z) \\
        &\stackrel{(5)}{\le} -\alpha \, \KL\left(\rho_t^{\bZ} \,\|\, \bar \rho_t^{\bZ}\right) + 
        e^{-\frac{3}{2}\alpha t} \, \frac{L^2}{\reg} \, W_2(\rho_0^{\bZ}, \bar \nu^{\bZ})^2 + \frac{2L^2}{\reg} \left(\frac{L^2}{\alpha^2} + 1\right) \, \Var_{\bar \nu^{Z}}(\bar Z) \\
        &\stackrel{(6)}{\le} -\alpha \, \KL\left(\rho_t^{\bZ} \,\|\, \bar \rho_t^{\bZ}\right) + 
        e^{-\frac{3}{2}\alpha t} \, \frac{L^2}{\reg} \, W_2(\rho_0^{\bZ}, \bar \nu^{\bZ})^2 + \frac{4L^4}{\alpha^2 \reg} \, \Var_{\bar \nu^{Z}}(\bar Z).
    \end{align*}
    In the above, in step (1) we use the time derivative formula for KL divergence from Lemma~\ref{Lem:ddtKLDiv}.
    In step (2), we apply the Cauchy-Schwarz inequality ($\langle a,b \rangle \le \frac{\reg}{2}\|a\|^2 + \frac{1}{2\reg}\|b\|^2$) to the second term.
    In step (3), we use the fact that $\bar \rho_t^{\bZ} = \bar \nu^{\bZ}$ is $(\alpha/\reg)$-strongly log-concave by Lemma~\ref{Lem:Equilibrium}, so it also satisfies $(\alpha/\reg)$-LSI, which allows us to bound the relative Fisher information by KL divergence.
    In step (4), we use the bound from Lemma~\ref{Lem:ComparisonArbitrary}.
    In step (5), we use the $W_2$ biased convergence bound~\eqref{Eq:CalcW2Bound} that we derived earlier, in the first part of this Theorem. 
    In step (6), we bound $1 \le L^2/\alpha^2$ since $\alpha \le L$.

    Let us now substitute $\bar \rho_t^{\bZ} = \bar \nu^{\bZ}$.
    Then we can write the differential inequality above as:
    \begin{align*}
        \frac{d}{dt} \left(e^{\alpha t} \KL\left(\rho_t^{\bZ} \,\|\, \bar \nu^{\bZ} \right) \right)
        &\le 
         e^{-\frac{1}{2}\alpha t} \, \frac{L^2}{\reg} \, W_2(\rho_0^{\bZ}, \bar \nu^{\bZ})^2 + e^{\alpha t} \, \frac{4L^4}{\alpha^2 \reg} \, \Var_{\bar \nu^{Z}}(\bar Z).
    \end{align*}
    Integrating from $0$ to $t$ and rearranging yields:
    \begin{align*}
        &\KL\left(\rho_t^{\bZ} \,\|\, \bar \nu^{\bZ} \right) \\
        &\le e^{-\alpha t} \, \KL\left(\rho_0^{\bZ} \,\|\, \bar \nu^{\bZ} \right) + e^{-\alpha t} \left( \frac{1-e^{-\frac{1}{2}\alpha t}}{\frac{1}{2}\alpha} \right)
         \frac{L^2}{\reg} \, W_2(\rho_0^{\bZ}, \bar \nu^{\bZ})^2 + \left(\frac{1-e^{-\alpha t}}{\alpha}\right) \frac{4L^4}{\alpha^2 \reg} \, \Var_{\bar \nu^{Z}}(\bar Z) \\
         &\le e^{-\alpha t} \left( \KL\left(\rho_0^{\bZ} \,\|\, \bar \nu^{\bZ} \right) + 
         \frac{2L^2}{\alpha \reg} \, W_2(\rho_0^{\bZ}, \bar \nu^{\bZ})^2\right) + \frac{4L^4}{\alpha^3 \reg} \, \Var_{\bar \nu^{Z}}(\bar Z) \\
         &\le e^{-\alpha t} \left( \KL\left(\rho_0^{\bZ} \,\|\, \bar \nu^{\bZ} \right) + 
         \frac{2L^2}{\alpha \reg} \, W_2(\rho_0^{\bZ}, \bar \nu^{\bZ})^2\right) + \frac{4L^4}{\alpha^3 \reg} \, \Var_{\bar \nu^{Z}}(\bar Z)
    \end{align*}
    as desired.
\end{proof}

%%%%%%%%%%%%%%%%%%%%%%%
\subsection{Proof of Corollary~\ref{Cor:AvgParticleSystem} (Average Particle Along the Finite-Particle Dynamics)}
\label{Sec:AvgParticleSystemProof}

\begin{proof}[Proof of Corollary~\ref{Cor:AvgParticleSystem}]
    For $\bZ_t = (\bX_t,\bY_t) = (X_t^1,\dots,X_t^N,Y_t^1,\dots,Y_t^N) \sim \rho_t^{\bZ}$ in $\R^{2dN}$, let $\rho_t^{Z,i} \in \P(\R^{2d})$ be the marginal distribution of the component $Z_t^i = (X_t^i,Y_t^i) \in \R^{2d}$, for $i \in [N]$.
    We rearrange the coordinates to write $\bZ_t = (Z_t^1,\dots,Z_t^N)$ for convenience, and still denote its distribution by $\rho_t^{\bZ}$. 
    We introduce an independent product of the marginal distributions:
    \begin{align*}
        \hat \rho_t^{\bZ} := \bigotimes_{i \in [N]} \rho_t^{Z,i}.
    \end{align*}    
    Since $\bar \nu^{\bZ} = (\bar \nu^X)^{\otimes N} \otimes (\bar \nu^Y)^{\otimes N}$ is an independent product, after rearranging the coordinates as above, we can write it as $\bar \nu^{\bZ} = (\bar \nu^Z)^{\otimes N}$ where $\bar \nu^Z = \bar \nu^X \otimes \bar \nu^Y$.
    
    Then we can bound the KL divergence by:
    \begin{align*}
        \KL(\rho_t^{\bZ} \,\|\, \bar \nu^{\bZ})
        = \E_{\rho_t^{\bZ}}\left[ \log \frac{\rho_t^{\bZ}}{\bar \nu^{\bZ}} \right]
        &= \E_{\rho_t^{\bZ}}\left[ \log \frac{\rho_t^{\bZ}}{\hat \rho_t^{\bZ}} \right] + \E_{\rho_t^{\bZ}}\left[\log \frac{\hat \rho_t^{\bZ}}{(\bar \nu^{Z})^{\otimes N}} \right] \\
        &= \E_{\rho_t^{\bZ}}\left[ \log \frac{\rho_t^{\bZ}}{\hat \rho_t^{\bZ}} \right] + \sum_{i \in [N]} \E_{\rho_t^{Z,i}} \left[\log \frac{\rho_t^{Z,i}}{\bar \nu^{Z}} \right] \\
        &= \KL\left(\rho_t^{\bZ} \,\|\, \hat \rho_t^{\bZ}\right) + \sum_{i \in [N]} \KL\left(\rho_t^{Z,i} \,\|\, \bar \nu^{Z}\right) \\
        &\ge \sum_{i \in [N]} \KL\left(\rho_t^{Z,i} \,\|\, \bar \nu^{Z}\right)
    \end{align*}
    where the inequality follows by dropping the first term $\KL\left(\rho_t^{\bZ} \,\|\, \hat \rho_t^{\bZ}\right) \ge 0$, which is the multivariate mutual information of $\bZ_t$. 
    Furthermore, since KL divergence is jointly convex in both arguments, we can further bound the above by:
    \begin{align*}
        \KL(\rho_t^{\bZ} \,\|\, \bar \nu^{\bZ})
        &\ge N \cdot \frac{1}{N} \sum_{i \in [N]} \KL\left(\rho_t^{Z,i} \,\|\, \bar \nu^{Z}\right) 
        \,\ge\, N \cdot \KL\left(\rho_t^{Z,\avg} \,\|\, \bar \nu^{Z}\right)
    \end{align*}
    using the definition $\rho_t^{Z,\avg} = \frac{1}{N} \sum_{i \in [N]} \rho_t^{Z,i}$.
    
    Combining this with the bound for $\KL(\rho_t^{\bZ} \,\|\, \bar \nu^{\bZ})$ from Theorem~\ref{Thm:ConvergenceDynToMF} gives the desired result.
\end{proof}

%%%%%%%%%%%%%%%%%%%%%%%%%
\section{Proofs for the Finite-Particle Algorithm}

%%%%%%%%%%%%%%%%%%%%%%%%%
\subsection{Preliminary Results}

%%%%%%%%%%%%%%%%%%%%%%%%%
\subsubsection{Bound in one step of the finite-particle algorithm}

We present the following bound which will be useful in our subsequent analysis.

\begin{lemma}\label{Lem:BoundDiffAlg}
    Assume Assumption~\ref{As:SCSmooth}.
    Let $\bZ_0 \sim \rho_0^{\bZ}$ for any $\rho_0^{\bZ} \in \P(\R^{2dN})$.
    For $t \ge 0$, define $\bZ_t \sim \rho_t^{\bZ}$ by:
    \begin{align*}
        \bZ_t = \bZ_0 + t b^{\bZ}(\bZ_0) + \sqrt{2\reg t} \, \bm{\zeta}
    \end{align*}
    where $\bm{\zeta} \sim \N(0,I)$ is an independent Gaussian random variable in $\R^{2dN}$, and $b^{\bZ}$ is the vector field defined in~\eqref{Eq:Defbz}.
    If $0 \le t \le 1/(4L)$, then:
    \begin{align*}
        \E\left[\left\|b^{\bZ}(\bZ_t)-b^{\bZ}(\bZ_0)\right\|^2\right] \le 128t^2 L^4 \, W_2(\rho_t^{\bZ}, \bar \nu^{\bZ})^2 + 64 t^2L^4 \, \Var_{\bar \nu^Z}(\bar Z) + 64\reg t dL^2 N.
    \end{align*}
\end{lemma}
\begin{proof}
    We can bound:
    \begin{align*}
        \E\left[\left\|b^{\bZ}(\bZ_0) - b^{\bZ}(\bZ_t)\right\|^2\right] 
        &\stackrel{(1)}{\le} 4L^2 \, \E\left[\left\|\bZ_0 - \bZ_t\right\|^2 \right] \\
        &\stackrel{(2)}{=} 4L^2 \, \E\left[\left\|t b^{\bZ}(\bZ_0) + \sqrt{2\reg t} \, \bm{\zeta} \right\|^2 \right] \\
        &\stackrel{(3)}{=} 4t^2 L^2 \, \E\left[\left\|b^{\bZ}(\bZ_0) \right\|^2 \right] + 8\reg t L^2 \, \E\left[\left\|\bm{\zeta}\right\|^2\right] \\
        &\stackrel{(4)}{=} 4t^2 L^2 \, \E\left[\left\|b^{\bZ}(\bZ_0) \right\|^2 \right] + 16\reg t dL^2 N \\
        &\stackrel{(5)}{\le} 8t^2 L^2 \, \E\left[\left\|b^{\bZ}(\bZ_0) - b^{\bZ}(\bZ_t)\right\|^2\right] + 8t^2 L^2 \, \E\left[\left\|b^{\bZ}(\bZ_t) \right\|^2 \right] + 16\reg t dL^2 N \\
        &\stackrel{(6)}{\le} \frac{1}{2} \, \E\left[\left\|b^{\bZ}(\bZ_0) - b^{\bZ}(\bZ_t)\right\|^2\right] + 8t^2 L^2 \, \E\left[\left\|b^{\bZ}(\bZ_t) \right\|^2 \right] + 16\reg t dL^2 N.
    \end{align*}
    In the above, in step (1) we have used the property that $b^{\bZ}$ is $(2L)$-Lipschitz from Lemma~\ref{Lem:LipschitzMonotone_bz}.
    In step (2), we plug in the definition of $\bZ_t$.
    In step (3), we expand the square, and note the cross term vanishes since $\bm{\zeta}$ is independent of $\bZ_0$ and $\E[\bm{\zeta}] = 0$.
    In step (4), we use the property that $\E[\|\bm{\zeta}\|^2] = 2dN$ since $\bm{\zeta} \sim \N(0,I)$ is a standard Gaussian in $\R^{2dN}$.
    In step (5), we introduce the term $b^{\bZ}(\bZ_t)$ again, and use the inequality $\|a+b\|^2 \le 2\|a\|^2 + 2\|b\|^2$.
    In step (6), we use the assumption that $t \le 1/(4L)$, so $8t^2 L^2 \le 1/2$.
    
    Rearranging the inequality above, we get:
    \begin{align*}
        \E\left[\left\|b^{\bZ}(\bZ_0) - b^{\bZ}(\bZ_t)\right\|^2\right]
        &\stackrel{(7)}{\le} 16t^2 L^2 \, \E_{\rho_t^{\bZ}}\left[\left\|b^{\bZ}(\bZ_t) \right\|^2 \right] + 32\reg t dL^2 N \\
        &\stackrel{(8)}{\le} 16t^2 L^2 \left(8L^2 \, W_2(\rho_t^{\bZ}, \bar \nu^{\bZ})^2 + 4L^2 \, \Var_{\bar \nu^Z}(\bar Z) + 8\reg dLN\right) + 32\reg t dL^2 N \\
        &\stackrel{(9)}{=} 128t^2 L^4 \, W_2(\rho_t^{\bZ}, \bar \nu^{\bZ})^2 + 64 t^2L^4 \, \Var_{\bar \nu^Z}(\bar Z) + 32\reg t dL^2 N (4tL+1) \\
        &\stackrel{(10)}{\le} 128t^2 L^4 \, W_2(\rho_t^{\bZ}, \bar \nu^{\bZ})^2 + 64 t^2L^4 \, \Var_{\bar \nu^Z}(\bar Z) + 64\reg t dL^2 N.
    \end{align*}
    In the above, in step (7) we rearrange the inequality from step (6).
    In step (8), we use the bound from Lemma~\ref{Lem:BoundFPArb} for the distribution $\rho_t^{\bZ}$.
    In step (9), we collect the terms from the previous line.
    In step (10), we again use the assumption that $t \le 1/(4L)$.
    This gives us the desired bound.    
\end{proof}

%%%%%%%%%%%%%%%%%%%%%%%%%
\subsubsection{Fokker-Planck equation for one-step interpolation of finite-particle algorithm}

We show the following continuous-time interpolation of one step of the finite-particle algorithm~\eqref{Eq:ParticleAlgorithm3}.
This is similar to the interpolation technique that has been used, for example, to analyze the mixing time of the Unadjusted Langevin Algorithm for sampling~\citep{VW19}.

\begin{lemma}\label{Lem:InterpolationFP}
    Let $\bZ_0 \sim \rho_0^{\bZ}$ for any $\rho_0^{\bZ} \in \P(\R^{2dN})$.
    For $t \ge 0$, define $\bZ_t \sim \rho_t^{\bZ}$ by:
    \begin{align}\label{Eq:Interpolation}
        \bZ_t = \bZ_0 + t b^{\bZ}(\bZ_0) + \sqrt{2t \reg} \, \bm{\zeta}
    \end{align}
    where $\bm{\zeta} \sim \N(0,I)$ is an independent standard Gaussian random variable in $\R^{2dN}$, and $b^{\bZ}$ is the vector field defined in~\eqref{Eq:Defbz}.
    Then the density $\rho_t^{\bZ}$ evolves following the Fokker-Planck equation:
    \begin{align*}
        \part{\rho_t^{\bZ}}{t} = -\nabla \cdot \left(\rho_t^{\bZ} \, b^{\bZ}_t \right) + \reg \Delta \rho_t^{\bZ}
    \end{align*}
    where we define the vector field $b^{\bZ}_t \colon \R^{2dN} \to \R^{2dN}$ by, for all $\bz \in \R^{2dN}$:
    \begin{align*}
        b^{\bZ}_t(\bz) = \E_{\rho_{0 \mid t}^{\bZ}}\left[b^{\bZ}(\bZ_0) \mid \bZ_t = \bz\right]
    \end{align*}
    where $\rho_{0 \mid t}^{\bZ}(\cdot \mid \bz)$ is the conditional distribution of $\bZ_0$ given $\bZ_t = \bz$ from the model~\eqref{Eq:Interpolation}.    
\end{lemma}
\begin{proof}
    Let $\rho_{0t}^{\bZ}$ be the joint distribution of $(\bZ_0, \bZ_t)$ following the model~\eqref{Eq:Interpolation}, which we can write in terms of the marginal and conditional distributions as, for all $\bz_0, \bz_t \in \R^{2dN}$:
    \begin{align*}
        \rho_{0t}^{\bZ}(\bz_0, \bz_t)
        = \rho_0^{\bZ}(\bz_0) \, \rho_{t \mid 0}^{\bZ}(\bz_t \mid \bz_0)
        = \rho_t^{\bZ}(\bz_t) \, \rho_{0 \mid t}^{\bZ}(\bz_0 \mid \bz_t).
    \end{align*}    
    Notice that $\bZ_t$ as defined in~\eqref{Eq:Interpolation} is the solution to the following stochastic process:
    \begin{align}\label{Eq:SDEInterp0}
        d\bZ_t = b^{\bZ}(\bZ_0) \, dt + \sqrt{2\reg} \, dW_t^{\bZ}
    \end{align}
    where $(W_t^{\bZ})_{t \ge 0}$ is the standard Brownian motion in $\R^{2dN}$ which is independent of $\bZ_0$.    
    We derive the Fokker-Planck equation for $\rho_t^{\bZ}$ as follows.
    First, when we condition on a fixed $\bZ_0 = \bz_0$, the drift in the process~\eqref{Eq:SDEInterp0} is a constant, so the conditional density $\rho_{t \mid 0}^{\bZ}(\cdot \mid \bz_0)$ of $\bZ_t$ conditioned on $\bZ_0 = \bz_0$ is given by the Fokker-Planck equation:
    \begin{align*}
        \part{\rho_{t \mid 0}^{\bZ}(\cdot \mid \bz_0)}{t}
        = -\nabla \cdot \left(\rho_{t \mid 0}^{\bZ}(\cdot \mid \bz_0) \, b^{\bZ}(\bz_0) \right) + \reg \Delta \rho_{t \mid 0}^{\bZ}(\cdot \mid \bz_0)
    \end{align*}
    where note that the divergence $\nabla \cdot$ and Laplacian $\Delta$ operate on the $\bz$ variable, not $\bz_0$.
    
    Then we can compute:
    \begin{align*}
        \part{\rho_t^{\bZ}}{t}
        &= \part{}{t} \int_{\R^{2dN}} \rho_0^{\bZ}(\bz_0) \, \rho_{t \mid 0}^{\bZ}(\cdot \mid \bz_0) \, d\bz_0 \\
        &= \int_{\R^{2dN}} \rho_0^{\bZ}(\bz_0) \, \part{\rho_{t \mid 0}^{\bZ}(\cdot \mid \bz_0)}{t} \, d\bz_0 \\
        &= \int_{\R^{2dN}} \rho_0^{\bZ}(\bz_0) \, \left(-\nabla \cdot \left(\rho_{t \mid 0}^{\bZ}(\cdot \mid \bz_0) \, b^{\bZ}(\bz_0) \right) + \reg \Delta \rho_{t \mid 0}^{\bZ}(\cdot \mid \bz_0)\right) \, d\bz_0 \\
        &= -\nabla \cdot \left( \int_{\R^{2dN}} \rho_0^{\bZ}(\bz_0) \, \rho_{t \mid 0}^{\bZ}(\cdot \mid \bz_0) \, b^{\bZ}(\bz_0) \, d\bz_0 \right) + \reg \Delta \left(\int_{\R^{2dN}} \rho_0^{\bZ}(\bz_0) \, \rho_{t \mid 0}^{\bZ}(\cdot \mid \bz_0) \, d\bz_0\right)  \\
        &= -\nabla \cdot \left( \rho_t^{\bZ} \, \int_{\R^{2dN}} \rho_{0 \mid t}^{\bZ}(\bz_0 \mid \cdot) \, b^{\bZ}(\bz_0) \, d\bz_0 \right) + \reg \Delta \rho_t^{\bZ}  \\
        &= -\nabla \cdot \left( \rho_t^{\bZ} \, \E_{\rho_{0 \mid t}^{\bZ}}\left[b^{\bZ}(\bZ_0) \mid \bZ_t = \cdot \right] \right) + \reg \Delta \rho_t^{\bZ} \\
        &= -\nabla \cdot \left( \rho_t^{\bZ} \, b_t^{\bZ} \right) + \reg \Delta \rho_t^{\bZ},
    \end{align*}
    as desired.
\end{proof}

%%%%%%%%%%%%%%%%%%%%%%%%%%%
\subsection{One-Step Recurrence for the Biased Convergence of the Finite-Particle Algorithm}
\label{Sec:OneStepRecurrence}

%%%%%%%%%%%%%%%%%%%%%%%%%%%
\subsubsection{One-step recurrence in $W_2$ distance}
\label{Sec:OneStepRecurrenceW2}

\begin{lemma}\label{Lem:OneStepW2}
    Assume Assumption~\ref{As:SCSmooth}.
    Let $\bz_k \sim \rho_k^{\bz,\eta}$ for any $\rho_k^{\bz,\eta} \in \P(\R^{2dN})$, and let $\bz_{k+1} \sim \rho_{k+1}^{\bz,\eta}$ be one step of the finite-particle algorithm~\eqref{Eq:ParticleAlgorithm3} with step size $0 < \eta \le \frac{\alpha}{64 L^2}$.
    Then:
    \begin{align*}
        W_2(\rho_{k+1}^{\bz,\eta}, \bar \nu^{\bZ})^2 \le e^{-\frac{3}{2} \alpha \eta} \, W_2(\rho_k^{\bz,\eta}, \bar \nu^{\bZ})^2 + \frac{8\eta L^2}{\alpha} \left(\Var_{\bar \nu^Z}(\bar Z) + 64 \, \reg \eta d N \right).
    \end{align*}
\end{lemma}
\begin{proof}
    We consider a continuous-time interpolation of one step of the algorithm~\eqref{Eq:ParticleAlgorithm3} as follows.
    Let $\bZ_0 \sim \rho_0^{\bZ}$ where we define $\rho_0^{\bZ} = \rho_k^{\bz,\eta}$.
    We define $(\bZ_t)_{0 \le t \le \eta}$ where $\bZ_t \sim \rho_t^{\bZ}$ evolves following:
    \begin{align}\label{Eq:SDEInterp}
        d\bZ_t = b^{\bZ}(\bZ_0) \, dt + \sqrt{2\reg} \, dW_t^{\bZ}
    \end{align}
    where $(W_t^{\bZ})_{t \ge 0}$ is the standard Brownian motion in $\R^{2dN}$ which is independent of $\bZ_0$.
    Then we notice that the distribution of $\bZ_\eta \sim \rho_\eta^{\bZ}$ of the process~\eqref{Eq:SDEInterp} is equal to the distribution of the next iterate $\bx_{k+1} \sim \rho_{k+1}^{\bz,\eta}$ along the algorithm~\eqref{Eq:ParticleAlgorithm3}, i.e., $\rho_{k+1}^{\bz,\eta} = \rho_\eta^{\bZ}$.
    This is because the solution to the process~\eqref{Eq:SDEInterp} at time $t = \eta$ is:
    \begin{align*}
        \bZ_t &= \bZ_0 + t b^{\bZ}(\bZ_0) + \sqrt{2\reg} \, W_t^{\bZ}
        \,\stackrel{d}{=}\, \bZ_0 + t b^{\bZ}(\bZ_0) + \sqrt{2\reg t} \, \bm{\zeta}
    \end{align*}
    where $\bm{\zeta} \sim \N(0,I)$ is an independent Gaussian random variable in $\R^{2dN}$, and $\stackrel{d}{=}$ means equality in distribution.
    This is the same update of the algorithm~\eqref{Eq:ParticleAlgorithm3},
    so indeed $\bZ_\eta \stackrel{d}{=} \bz_{k+1}$, so $\rho_\eta^{\bZ} = \rho_{k+1}^{\bz,\eta}$.

    We also consider $(\bar \bZ_t)_{0 \le t \le \eta}$ where $\bar \bZ_t \sim \bar \rho_t^{\bZ} = \bar \nu^{\bZ}$ evolves following the stationary tensorized mean-field dynamics~\eqref{Eq:ParticleSystem3}:
    \begin{align}\label{Eq:CalcSDE1}
        d\bar \bZ_t = \bar b^{\bZ}(\bar \bZ_t) \, dt + \sqrt{2\reg} \, dW_t^{\bZ}
    \end{align}
    starting from the stationary distribution $\bar \rho_0^{\bZ} = \bar \nu^{\bZ}$, so $\bar \rho_t^{\bZ} = \bar \nu^{\bZ}$ for all $t \ge 0$.
    Suppose we run the stochastic processes~\eqref{Eq:SDEInterp} and~\eqref{Eq:CalcSDE1} using the same standard Brownian motion $(W_t^{\bZ})_{t \ge 0}$, which is independent of $\bZ_0$ and $\bar \bZ_0$.
    Furthermore, suppose we start the processes~\eqref{Eq:SDEInterp} and~\eqref{Eq:CalcSDE1} from $(\bZ_0, \bar \bZ_0)$ which has a joint distribution which is the optimal $W_2$ coupling between $\rho_0^{\bZ} = \rho_k^{\bz,\eta}$ and $\bar \rho_0^{\bZ} = \bar \nu^{\bZ}$.

    From the set up above, we have that the difference $\bZ_t-\bar \bZ_t$ satisfies:
    \begin{align*}
        d\left(\bZ_t-\bar \bZ_t\right)
        = \left(b^{\bZ}(\bZ_0) - \bar b^{\bZ}(\bar \bZ_t) \right) \, dt
    \end{align*}
    where the Brownian motion terms cancel because they are equal in the two processes.
    Therefore:
    \begin{align*}
        d\left\|\bZ_t-\bar \bZ_t\right\|^2 = 2\left\langle \bZ_t-\bar \bZ_t, \, b^{\bZ}(\bZ_0)-\bar b^{\bZ}(\bar \bZ_t) \right \rangle dt.
    \end{align*}
    Taking expectation over the joint distribution of all the variables, we obtain:
    \begin{align}
        \frac{d}{dt} \E\left[ \left\|\bZ_t-\bar \bZ_t\right\|^2 \right]
        &= 2 \E\left[\left\langle \bZ_t-\bar \bZ_t, \, b^{\bZ}(\bZ_0)-\bar b^{\bZ}(\bar \bZ_t) \right \rangle \right] \notag \\
        &= 2 \E\left[\left\langle \bZ_t-\bar \bZ_t, b^{\bZ}(\bZ_t)- b^{\bZ}(\bar \bZ_t) \right \rangle \right] \notag \\
        &\qquad + 2 \E\left[\left\langle \bZ_t-\bar \bZ_t, b^{\bZ}(\bZ_0)-b^{\bZ}(\bZ_t)+ b^{\bZ}(\bar \bZ_t)-\bar b^{\bZ}(\bar \bZ_t)\right \rangle \right]. \label{Eq:CalcAlg2}
    \end{align}
    By Lemma~\ref{Lem:LipschitzMonotone_bzbar}, $-b^{\bZ}$ is $\alpha$-strongly monotone, so we can bound the first term above by:
    \begin{align*}
        \E\left[\left\langle \bZ_t-\bar \bZ_t, b^{\bZ}(\bZ_t)- b^{\bZ}(\bar \bZ_t) \right \rangle \right]
        \le -\alpha \,\E\left[\left\| \bZ_t-\bar \bZ_t \right \|^2 \right].
    \end{align*}
    We can bound the second term above by:
    \begin{align*}
        &\E\left[\left\langle \bZ_t-\bar \bZ_t, b^{\bZ}(\bZ_0)-b^{\bZ}(\bZ_t)+ b^{\bZ}(\bar \bZ_t)-\bar b^{\bZ}(\bar \bZ_t)\right \rangle \right] \\
        &\stackrel{(1)}{\le} \frac{\alpha}{8} \E\left[\left\| \bZ_t-\bar \bZ_t \right \|^2 \right]
        + \frac{2}{\alpha} \E\left[\left\|b^{\bZ}(\bZ_0)-b^{\bZ}(\bZ_t)+ b^{\bZ}(\bar \bZ_t)-\bar b^{\bZ}(\bar \bZ_t)\right \|^2 \right] \\
        &\stackrel{(2)}{\le} \frac{\alpha}{8} \E\left[\left\| \bZ_t-\bar \bZ_t \right \|^2 \right]
        + \frac{4}{\alpha} \E\left[\left\|b^{\bZ}(\bZ_0)-b^{\bZ}(\bZ_t) \right \|^2 \right] + \frac{4}{\alpha} \E_{\bar \nu^{\bZ}}\left[\left\|b^{\bZ}(\bar \bZ_t)-\bar b^{\bZ}(\bar \bZ_t)\right \|^2 \right] \\
        &\stackrel{(3)}{\le} \frac{\alpha}{8} \E\left[\left\| \bZ_t-\bar \bZ_t \right \|^2 \right]
        + \frac{4}{\alpha} \left(128t^2 L^4 \, W_2(\rho_t^{\bZ}, \bar \nu^{\bZ})^2 + 64 t^2L^4 \, \Var_{\bar \nu^Z}(\bar Z) + 64\reg t dL^2 N \right) \\
        &\qquad + \frac{4}{\alpha} \, L^2 \, \Var_{\bar \nu^Z}(\bar Z) \\
        &\stackrel{(4)}{\le} \left( \frac{\alpha}{8} + \frac{512t^2 L^4}{\alpha} \right) \E\left[\left\| \bZ_t-\bar \bZ_t \right \|^2 \right] 
        + \frac{4L^2}{\alpha}\left(64 t^2 L^2 + 1 \right) \, \Var_{\bar \nu^Z}(\bar Z) + \frac{256 \reg t dL^2 N}{\alpha} \\
        &\stackrel{(5)}{\le} \frac{\alpha}{4} \E\left[\left\| \bZ_t-\bar \bZ_t \right \|^2 \right] 
        + \frac{8L^2}{\alpha} \, \Var_{\bar \nu^Z}(\bar Z) + \frac{256 \reg \eta dL^2 N}{\alpha}.
    \end{align*}
    In the above, in step (1) we use the inequality $\langle a,b \rangle \le \frac{\alpha}{8} \|a\|^2 + \frac{2}{\alpha} \|b\|^2$.
    In step (2), we use the inequality $\|a+b\|^2 \le 2\|a\|^2 + 2\|b\|^2$ to the second term.
    In step (3), we use the bound from Lemma~\ref{Lem:BoundDiffAlg} to the second term (note $t \le \eta \le \frac{\alpha}{64L^2} \le \frac{1}{64L} \le \frac{1}{4L}$ so the assumption in Lemma~\ref{Lem:BoundDiffAlg} is satisfied), and we use the bound from Lemma~\ref{Lem:ComparisonStationary} to the third term.
    In step (4), we use the bound $W_2(\rho_t^{\bZ}, \bar \nu^{\bZ})^2 \le \E\left[\left\| \bZ_t-\bar \bZ_t \right \|^2 \right]$ which follows from the definition of the $W_2$ distance.
    In step (5), we use the assumption $t \le \eta \le \frac{\alpha}{64 L^2}$, so in the first term, $\frac{512t^2 L^4}{\alpha} \le \frac{\alpha}{8}$; we also use the assumption $t \le \eta \le \frac{\alpha}{64 L^2} \le \frac{1}{64 L}$, so in the second term, $64 t^2 L^2 \le \frac{1}{64} \le 1$; and we use $t \le \eta$ in the third term.

    Plugging in the two bounds above to our earlier calculation~\eqref{Eq:CalcAlg2}, we obtain:
    \begin{align*}
        \frac{d}{dt} \E\left[ \left\|\bZ_t-\bar \bZ_t\right\|^2 \right]
        &\le -\frac{3\alpha}{2} \E\left[\left\| \bZ_t-\bar \bZ_t \right \|^2 \right] 
        + \frac{8L^2}{\alpha} \, \Var_{\bar \nu^Z}(\bar Z) + \frac{512 \reg \eta dL^2 N}{\alpha}.
    \end{align*}
    We can write this differential inequality equivalently as:
    \begin{align*}
        \frac{d}{dt} \left( e^{\frac{3}{2} \alpha t} \E\left[ \left\|\bZ_t-\bar \bZ_t\right\|^2 \right] \right)
        &\le e^{\frac{3}{2} \alpha t} \left( \frac{8L^2}{\alpha} \, \Var_{\bar \nu^Z}(\bar Z) + \frac{512 \reg \eta dL^2 N}{\alpha} \right).
    \end{align*}
    Integrating from $t=0$ to $t=\eta$ and rearranging yields:
    \begin{align*}
        \E\left[ \left\|\bZ_\eta-\bar \bZ_\eta\right\|^2 \right]
        &\le e^{-\frac{3}{2} \alpha \eta} \, \E\left[ \left\|\bZ_0-\bar \bZ_0\right\|^2 \right] + \left(\frac{1-e^{-\frac{3}{2} \alpha \eta}}{\frac{3}{2}\alpha}\right) \left(  \frac{8L^2}{\alpha} \, \Var_{\bar \nu^Z}(\bar Z) + \frac{512 \reg \eta dL^2 N}{\alpha} \right) \\
        &= e^{-\frac{3}{2} \alpha \eta} \, W_2(\rho_0^{\bZ}, \bar \nu^{\bZ})^2 + \left(\frac{1-e^{-\frac{3}{2} \alpha \eta}}{\frac{3}{2}\alpha}\right) \frac{8L^2}{\alpha} \left(\Var_{\bar \nu^Z}(\bar Z) + 64 \, \reg \eta d N \right) \\
        &\le e^{-\frac{3}{2} \alpha \eta} \, W_2(\rho_0^{\bZ}, \bar \nu^{\bZ})^2 + \frac{8\eta L^2}{\alpha} \left(\Var_{\bar \nu^Z}(\bar Z) + 64 \, \reg \eta d N \right).  
    \end{align*}
    In the second step above, we have used the assumption that $(\bZ_0, \bar \bZ_0)$ has the optimal $W_2$ coupling between $\rho_0^{\bZ}$ and $\bar \nu^{\bZ}$.
    In the last step, we use the inequality $1-e^{-c} \le c$, which holds for all $c = \frac{3}{2}\alpha \eta \ge 0$.
    Using the definition of the $W_2$ distance as the infimum over all coupling, from the above, we conclude:
    \begin{align*}
        W_2(\rho_\eta^{\bZ}, \bar \nu^{\bZ})^2 \le e^{-\frac{3}{2} \alpha \eta} \, W_2(\rho_0^{\bZ}, \bar \nu^{\bZ})^2 + \frac{8\eta L^2}{\alpha} \left(\Var_{\bar \nu^Z}(\bar Z) + 64 \, \reg \eta d N \right).   
    \end{align*}
    Substituting back $\rho_0^{\bZ} = \rho_k^{\bz,\eta}$ and $\rho_\eta^{\bZ} = \rho_{k+1}^{\bz,\eta}$ gives us the desired bound.
\end{proof}

%%%%%%%%%%%%%%%%%%%%%%%%%%%
\subsubsection{One-step recurrence in KL divergence}
\label{Sec:OneStepRecurrenceKL}

\begin{lemma}\label{Lem:OneStepKL}
    Assume Assumption~\ref{As:SCSmooth}.
    Let $\bz_k \sim \rho_k^{\bz,\eta}$ for any $\rho_k^{\bz,\eta} \in \P(\R^{2dN})$, and let $\bz_{k+1} \sim \rho_{k+1}^{\bz,\eta}$ be one step of the finite-particle algorithm~\eqref{Eq:ParticleAlgorithm3} with step size $0 < \eta \le \frac{\alpha}{64 L^2}$.
    Then:
    \begin{align*}
        \KL\left(\rho_{k+1}^{\bz,\eta} \,\|\, \bar \nu^{\bZ} \right) 
        &\le e^{-\alpha \eta} \left( \KL\left(\rho_{k}^{\bz,\eta} \,\|\, \bar \nu^{\bZ} \right) + 
        \frac{3 \eta L^2}{\reg}  W_2(\rho_k^{\bz,\eta}, \bar \nu^{\bZ})^2 \right) + 88 \eta^2 dL^2 N + \frac{6 \eta L^2}{\reg} \Var_{\bar \nu^{Z}}(\bar Z).
    \end{align*}
\end{lemma}
\begin{proof}
    We consider a continuous-time interpolation of one step of the algorithm~\eqref{Eq:ParticleAlgorithm3}, as in the proof of Lemma~\ref{Lem:OneStepW2}, but now we work with the Fokker-Planck equation, rather than the stochastic process.
    Let $\bZ_0 \sim \rho_0^{\bZ}$ where we define $\rho_0^{\bZ} = \rho_k^{\bz,\eta}$.
    For $0 \le t \le \eta$, we define $\bZ_t \sim \rho_t^{\bZ}$ by:
    \begin{align}\label{Eq:Interpolation2}
        \bZ_t = \bZ_0 + t b^{\bZ}(\bZ_0) + \sqrt{2t \reg} \, \bm{\zeta}
    \end{align}
    where $\bm{\zeta} \sim \N(0,I)$ is an independent standard Gaussian random variable in $\R^{2dN}$.
    By the same argument as in the proof of Lemma~\ref{Lem:OneStepW2}, the distribution of $\bZ_\eta \sim \rho_\eta^{\bZ}$ is equal to the distribution of the next iterate $\bz_{k+1} \sim \rho_{k+1}^{\bz,\eta}$ along the algorithm~\eqref{Eq:ParticleAlgorithm3}.    
    Let $\rho_{0t}^{\bZ}$ be the joint distribution of $(\bZ_0, \bZ_t)$ following the model~\eqref{Eq:Interpolation2}, and let $\rho_{0 \mid t}^{\bZ}(\cdot \mid \bz)$ be the conditional distribution of $\bZ_0$ given $\bZ_t = \bz$.
    Define the vector field $b^{\bZ}_t \colon \R^{2dN} \to \R^{2dN}$ by, for all $\bz \in \R^{2dN}$:
    \begin{align*}
        b^{\bZ}_t(\bz) = \E_{\rho_{0 \mid t}^{\bZ}}\left[b^{\bZ}(\bZ_0) \mid \bZ_t = \bz\right] = \int_{\R^{2dN}} \rho_{0 \mid t}^{\bZ}(\bz_0 \mid \bz) \, b^{\bZ}(\bz_0) \, d\bz_0.
    \end{align*}
    Then by Lemma~\ref{Lem:InterpolationFP}, we know that the density $\rho_t^{\bZ}$ evolves following the Fokker-Planck equation:
    \begin{align*}
        \part{\rho_t^{\bZ}}{t} = -\nabla \cdot \left(\rho_t^{\bZ} \, b^{\bZ}_t \right) + \reg \Delta \rho_t^{\bZ}.
    \end{align*}

    We also define $\bar \bZ_t \sim \bar \rho_t^{\bZ} = \bar \nu^{\bZ}$ which evolves following the stationary tensorized mean-field dynamics~\eqref{Eq:MFSystemTensor} from $\bar \bZ_0 \sim \bar \rho_0^{\bZ} = \bar \nu^{\bZ}$, so $\bar \rho_t^{\bZ}$ satisfies the Fokker-Planck equation:
    \begin{align*}
        \part{\bar \rho_t^{\bZ}}{t} = -\nabla \cdot \left(\bar \rho_t^{\bZ} \, \bar b^{\bZ} \right) + \reg \Delta \bar \rho_t^{\bZ}.
    \end{align*}
    (Note since $\bar \rho_t^{\bZ} = \bar \nu^{\bZ}$, both sides of the Fokker-Planck equation above are in fact equal to $0$.)

    Then using the formula from Lemma~\ref{Lem:ddtKLDiv}, we can compute:
    \begin{align}
        \frac{d}{dt} \KL\left(\rho_t^{\bZ} \,\|\, \bar \rho_t^{\bZ}\right)
        &\stackrel{(1)}{=} -\reg \, \FI\left(\rho_t^{\bZ} \,\|\, \bar \rho_t^{\bZ}\right) + \E_{\rho_t^{\bZ}}\left[\left\langle \nabla \log \frac{\rho_t^{\bZ}}{\bar \rho_t^{\bZ}}, \, b_t^{\bZ}(\bZ_t) - \bar b^{\bZ}(\bZ_t) \right\rangle \right] \notag \\
        &\stackrel{(2)}{=} -\reg \, \FI\left(\rho_t^{\bZ} \,\|\, \bar \rho_t^{\bZ}\right) + \E_{\rho_{0t}^{\bZ}}\left[\left\langle \nabla \log \frac{\rho_t^{\bZ}}{\bar \rho_t^{\bZ}}(\bZ_t), \, b^{\bZ}(\bZ_0) - \bar b^{\bZ}(\bZ_t) \right\rangle \right] \notag \\
        &\stackrel{(3)}{\le} -\frac{\reg}{2} \, \FI\left(\rho_t^{\bZ} \,\|\, \bar \rho_t^{\bZ}\right) + \frac{1}{2\reg} \, \E_{\rho_{0t}^{\bZ}}\left[\left\|b^{\bZ}(\bZ_0) - \bar b^{\bZ}(\bZ_t) \right\|^2 \right] \notag \\
        &\stackrel{(4)}{\le} -\alpha \, \KL\left(\rho_t^{\bZ} \,\|\, \bar \rho_t^{\bZ}\right) + \frac{1}{2\reg} \, \E_{\rho_{0t}^{\bZ}}\left[\left\|b^{\bZ}(\bZ_0) - \bar b^{\bZ}(\bZ_t) \right\|^2 \right]. \label{Eq:CalcOneStep1}
    \end{align}
    In the above, in step (1) we use the time derivative formula for KL divergence from Lemma~\ref{Lem:ddtKLDiv}.
    In step (2), we plug in the definition $b^{\bZ}_t(\bz) = \E_{\rho_{0 \mid t}^{\bZ}}\left[b^{\bZ}(\bZ_0) \mid \bZ_t = \bz\right]$, and use the tower property to write the iterated expectation $\E_{\rho_t^{\bZ}}\E_{\rho_{0 \mid t}^{\bZ}}$ as a joint expectation $\E_{\rho_{0t}^{\bZ}}$ over $(\bZ_0,\bZ_t)$.
    In step (3), we use the inequality $\langle a,b \rangle \le \frac{\reg}{2}\|a\|^2 + \frac{1}{2\reg}\|b\|^2$ to the second term.
    In step (4), we use the fact that $\bar \rho_t^{\bZ} = \bar \nu^{\bZ}$ is $(\alpha/\reg)$-strongly log-concave by Lemma~\ref{Lem:Equilibrium}, so it also satisfies $(\alpha/\reg)$-LSI, which allows us to bound the relative Fisher information by KL divergence.

    We can bound the second term above by:
    {\allowdisplaybreaks
    \begin{align*}
        \frac{1}{2\reg} \, &\E_{\rho_{0t}^{\bZ}}\left[\left\|b^{\bZ}(\bZ_0) - \bar b^{\bZ}(\bZ_t) \right\|^2 \right] \\
        &\stackrel{(5)}{\le} 
        \frac{1}{\reg} \, \E_{\rho_{0t}^{\bZ}}\left[\left\|b^{\bZ}(\bZ_0) - b^{\bZ}(\bZ_t) \right\|^2 \right]
        + \frac{1}{\reg} \, \E_{\rho_{t}^{\bZ}}\left[\left\|b^{\bZ}(\bZ_t) - \bar b^{\bZ}(\bZ_t) \right\|^2 \right] \\ 
        &\stackrel{(6)}{\le} 
        \frac{1}{\reg} \left(128t^2 L^4 \, W_2(\rho_t^{\bZ}, \bar \nu^{\bZ})^2 + 64 t^2L^4 \, \Var_{\bar \nu^Z}(\bar Z) + 64\reg t dL^2 N\right) \\
        &\qquad 
        + \frac{1}{\reg} \left(2L^2 \, W_2(\rho_t^{\bZ}, \bar \nu^{\bZ})^2 + 4L^2 \, \Var_{\bar \nu^{Z}}(\bar Z) \right) \\
        &\stackrel{(7)}{\le} 
        \frac{3L^2}{\reg} \, W_2(\rho_t^{\bZ}, \bar \nu^{\bZ})^2 + \frac{5L^2}{\reg} \, \Var_{\bar \nu^{Z}}(\bar Z) + 64 t dL^2 N \\
        &\stackrel{(8)}{\le} 
        \frac{3L^2}{\reg} \left(e^{-\frac{3}{2} \alpha t} \, W_2(\rho_0^{\bZ}, \bar \nu^{\bZ})^2 + \frac{8t L^2}{\alpha} \left(\Var_{\bar \nu^Z}(\bar Z) + 64 \, \reg t d N \right) \right) + \frac{5L^2}{\reg} \, \Var_{\bar \nu^{Z}}(\bar Z) + 64 t dL^2 N \\
        &\stackrel{(9)}{\le} 
        e^{-\frac{3}{2} \alpha t} \, \frac{3L^2}{\reg} \, W_2(\rho_0^{\bZ}, \bar \nu^{\bZ})^2 + \frac{24 \eta L^4}{\alpha \reg} \left(\Var_{\bar \nu^Z}(\bar Z) + 64 \, \reg \eta d N \right) + \frac{5L^2}{\reg} \, \Var_{\bar \nu^{Z}}(\bar Z) + 64 \eta dL^2 N \\
        &\stackrel{(10)}{\le} 
        e^{-\frac{3}{2} \alpha t} \, \frac{3L^2}{\reg} \, W_2(\rho_0^{\bZ}, \bar \nu^{\bZ})^2 + 88 \, \eta dL^2 N + \frac{6L^2}{\reg} \, \Var_{\bar \nu^{Z}}(\bar Z).
    \end{align*}
    }
    In step (5), we introduce an intermediate term $b^{\bZ}(\bZ_t)$ and use the inequality $\|a+b\|^2 \le 2\|a\|^2 + 2\|b\|^2$ to the second term.
    In step (6), we use the bound from Lemma~\ref{Lem:BoundDiffAlg} in the second term (note $t \le \eta \le \frac{\alpha}{64L^2} \le \frac{1}{64L} \le \frac{1}{4L}$ so the assumption in Lemma~\ref{Lem:BoundDiffAlg} is satisfied), and the bound from Lemma~\ref{Lem:ComparisonArbitrary} in the third term.
    In step (7), we use the assumption $t \le \eta \le \frac{\alpha}{64L^2} \le \frac{1}{64L}$, so $128t^2 L^4 \le L^2$ and $64 t^2L^4 \le L^2$.
    In step (8), we apply the $W_2$ bound from Lemma~\ref{Lem:OneStepW2}, applied to $\rho_t^{\bZ}$ when considered as one step of the discrete-time algorithm~\eqref{Eq:ParticleAlgorithm3} from $\rho_0^{\bZ}$ with step size $t \le \eta \le \frac{\alpha}{64L^2}$.
    In step (9), we use the bound $t \le \eta$ in the second and fourth terms.
    In step (10), we use the bound $\eta \le \frac{\alpha}{64L^2}$, so $\frac{24 \eta L^4}{\alpha \reg} \cdot 64 \reg \eta dN \le 24 \eta dL^2 N$, and we also bound $\frac{24 \eta L^4}{\alpha \reg} \le \frac{L^2}{\reg}$.

    Plugging in the bound above to the calculation in~\eqref{Eq:CalcOneStep1}, we obtain:
    \begin{align*}
        \frac{d}{dt} \KL\left(\rho_t^{\bZ} \,\|\, \bar \rho_t^{\bZ}\right)
        &\le -\alpha \, \KL\left(\rho_t^{\bZ} \,\|\, \bar \rho_t^{\bZ}\right) + e^{-\frac{3}{2} \alpha t} \, \frac{3L^2}{\reg} \, W_2(\rho_0^{\bZ}, \bar \nu^{\bZ})^2 + 88 \, \eta dL^2 N + \frac{6L^2}{\reg} \, \Var_{\bar \nu^{Z}}(\bar Z).
    \end{align*}
    We can write this differential inequality equivalently as:
    \begin{align*}
        \frac{d}{dt} \left(e^{\alpha t} \, \KL\left(\rho_t^{\bZ} \,\|\, \bar \rho_t^{\bZ}\right) \right)
        &\le e^{-\frac{1}{2} \alpha t} \, \frac{3L^2}{\reg} \, W_2(\rho_0^{\bZ}, \bar \nu^{\bZ})^2 
        + e^{\alpha t} \left( 88 \, \eta dL^2 N + \frac{6L^2}{\reg} \, \Var_{\bar \nu^{Z}}(\bar Z)\right).
    \end{align*}
    Integrating from $t=0$ to $t=\eta$ and rearranging yields:
    \begin{align*}
        \KL\left(\rho_\eta^{\bZ} \,\|\, \bar \rho_\eta^{\bZ}\right)
        &\le e^{-\alpha \eta} \, \KL\left(\rho_0^{\bZ} \,\|\, \bar \rho_0^{\bZ}\right) + 
        e^{-\alpha \eta} \left(\frac{1-e^{-\frac{1}{2}\alpha \eta}}{\frac{1}{2} \alpha} \right) \frac{3L^2}{\reg} \, W_2(\rho_0^{\bZ}, \bar \nu^{\bZ})^2 \\
        &\qquad + \left(\frac{1-e^{-\alpha \eta}}{\alpha}\right) \left( 88 \, \eta dL^2 N + \frac{6L^2}{\reg} \, \Var_{\bar \nu^{Z}}(\bar Z) \right) \\
        &\le e^{-\alpha \eta} \left( \KL\left(\rho_0^{\bZ} \,\|\, \bar \rho_0^{\bZ}\right) + 
        \frac{3 \eta L^2}{\reg} \, W_2(\rho_0^{\bZ}, \bar \nu^{\bZ})^2 \right) + 88 \, \eta^2 dL^2 N + \frac{6 \eta L^2}{\reg} \, \Var_{\bar \nu^{Z}}(\bar Z)
    \end{align*}
    where in the second step above we use the inequality $1-e^{-c} \le c$ for $c = \frac{1}{2} \alpha \eta$ in the second term, and for $c = \alpha \eta$ in the third term.
    Substituting $\rho_0^{\bZ} = \rho_k^{\bz,\eta}$, $\rho_\eta^{\bZ} = \rho_{k+1}^{\bz,\eta}$, and $\bar \rho_0^{\bZ} = \bar \rho_\eta^{\bZ} = \bar \nu^{\bZ}$ gives the desired bound.
\end{proof}

%%%%%%%%%%%%%%%%%%%
\subsection{Proof of Theorem~\ref{Thm:ConvergenceAlgToMF} (Biased Convergence of Finite-Particle Algorithm to Stationary Mean-Field Distribution)}
\label{Sec:ConvergenceAlgToMFProof}

\begin{proof}[Proof of Theorem~\ref{Thm:ConvergenceAlgToMF}]
    \textbf{(1)~Biased $W_2$ convergence:}
    For simplicity, let $D_k := W_2(\rho_k^{\bz,\eta}, \bar \nu^{\bZ})^2$.
    Recall from Lemma~\ref{Lem:OneStepW2} we have the recurrence:
    \begin{align*}
        D_{k+1} \le e^{-\frac{3}{2} \alpha \eta} \, D_k + C
    \end{align*}
    where $C := \frac{8\eta L^2}{\alpha} \left(\Var_{\bar \nu^Z}(\bar Z) + 64 \, \reg \eta d N \right).$
    Iterating this recurrence gives us:
    \begin{align*}
        D_k &\le e^{-\frac{3}{2} \alpha \eta k} \, D_0 + C \sum_{i=0}^{k-1} e^{-\frac{3}{2} \alpha \eta i} \\
        &\le e^{-\frac{3}{2} \alpha \eta k} \, D_0 + \frac{C}{1-e^{-\frac{3}{2} \alpha \eta}} \\
        &\le e^{-\frac{3}{2} \alpha \eta k} \, D_0 + \frac{C}{\alpha \eta}.
    \end{align*}
    In the second step above, we use the bound $\sum_{i=0}^{k-1} e^{-\frac{3}{2} \alpha \eta i} \le \sum_{i=0}^{\infty} e^{-\frac{3}{2} \alpha \eta i} = 1/(1-e^{-\frac{3}{2} \alpha \eta})$.
    In the last step, we use the inequality $1-e^{-c} \ge \frac{2}{3} c$ for $0 \le c = \frac{3}{2} \alpha \eta \le \frac{3}{4}$, which holds since $\eta \le \frac{\alpha}{64 L^2} \le \frac{1}{64 \alpha} \le \frac{1}{8 \alpha}$.
    Substituting the definition of $D_k$ and $C$ gives the $W_2$ convergence bound.

    \medskip
    \noindent
    \textbf{(2)Biased convergence in KL divergence:}
    Let $H_k := \KL\left(\rho_{k+1}^{\bz,\eta} \,\|\, \bar \nu^{\bZ} \right)$, and $D_k = W_2(\rho_k^{\bz,\eta}, \bar \nu^{\bZ})^2$ as before.
    Recall from Lemma~\ref{Lem:OneStepKL} we have the recurrence:
    \begin{align*}
        H_{k+1}
        &\le e^{-\alpha \eta} \left( H_k + 
        \frac{3 \eta L^2}{\reg} D_k \right) + C'
    \end{align*}
    where $C' :=  88 \eta^2 dL^2 N + \frac{6 \eta L^2}{\reg} \Var_{\bar \nu^{Z}}(\bar Z)$.
    Iterating this recurrence gives us:
    \begin{align}\label{Eq:CalcKLProof1}
        H_k &\le e^{-\alpha \eta k} H_0 + \frac{3 \eta L^2}{\reg} \sum_{i=0}^{k-1} e^{-\alpha \eta (k-i)} D_i + C' \sum_{i=0}^{k-1} e^{-\alpha \eta i}.
    \end{align}
    By using the $W_2$ convergence bound we proved above, we can write:
    \begin{align*}
        \sum_{i=0}^{k-1} e^{-\alpha \eta (k-i)} D_i
        &\le \sum_{i=0}^{k-1} e^{-\alpha \eta (k-i)} \left(e^{-\frac{3}{2} \alpha \eta i} \, D_0 + \frac{C}{\alpha \eta} \right) \\
        &= e^{-\alpha \eta k} D_0 \sum_{i=0}^{k-1} e^{-\frac{1}{2} \alpha \eta i} + \frac{C}{\alpha \eta} \sum_{i=0}^{k-1} e^{-\alpha \eta (k-i)} \\
        &\le \frac{e^{-\alpha \eta k} D_0}{(1-e^{-\frac{1}{2} \alpha \eta})} + \frac{C}{\alpha \eta} \frac{e^{-\alpha \eta}}{(1-e^{-\alpha \eta})} \\
        &\le \frac{3e^{-\alpha \eta k} D_0}{\alpha \eta} + \frac{3C}{2\alpha^2 \eta^2}
    \end{align*}
    where in the last step we again use the bound $1-e^{-c} \ge \frac{2}{3} c$ for $c = \frac{1}{2} \alpha \eta$ and $c = \alpha \eta$, and we also bound $e^{-\alpha \eta} \le 1$.
    Plugging in the value $C = \frac{8\eta L^2}{\alpha} \left(\Var_{\bar \nu^Z}(\bar Z) + 64 \, \reg \eta d N \right)$, the middle term in~\eqref{Eq:CalcKLProof1} can be bounded by:
    \begin{align*}
        \frac{3 \eta L^2}{\reg} \sum_{i=0}^{k-1} e^{-\alpha \eta (k-i)} D_i
        &\le \frac{3 \eta L^2}{\reg} \left( \frac{3e^{-\alpha \eta k} D_0}{\alpha \eta} + \frac{3}{2\alpha^2 \eta^2} \left(\frac{8\eta L^2}{\alpha} \left(\Var_{\bar \nu^Z}(\bar Z) + 64 \, \reg \eta d N \right)\right) \right) \\
        &= e^{-\alpha \eta k} \, \frac{9 L^2 D_0}{\alpha \reg} + \frac{36 L^4}{\alpha^3 \reg} \left(\Var_{\bar \nu^Z}(\bar Z) + 64 \, \reg \eta d N \right)
    \end{align*}
    For the last term in~\eqref{Eq:CalcKLProof1}, we can bound:
    \begin{align*}
        C' \sum_{i=0}^{k-1} e^{-\alpha \eta i}
        &\le \frac{C'}{1-e^{-\alpha \eta}}
        \,\le\, \frac{3C'}{2\alpha \eta}
        \,=\, \frac{132 \eta dL^2 N}{\alpha} + \frac{9 L^2}{\alpha \reg} \Var_{\bar \nu^{Z}}(\bar Z)
    \end{align*}
    where again we use the bound $1-e^{-c} \ge \frac{2}{3} c$.
    Plugging in these two bounds to~\eqref{Eq:CalcKLProof1}, we obtain:
    \begin{align*}
        H_k 
        &\le e^{-\alpha \eta k} H_0 + e^{-\alpha \eta k} \frac{9 L^2 D_0}{\alpha \reg} + \frac{36 L^4}{\alpha^3 \reg} \left(\Var_{\bar \nu^Z}(\bar Z) + 64 \, \reg \eta d N \right)
        + \frac{132 \eta dL^2 N}{\alpha} + \frac{9 L^2}{\alpha \reg} \Var_{\bar \nu^{Z}}(\bar Z) \\
        &= e^{-\alpha \eta k} \left(H_0 + \frac{9 L^2}{\alpha \reg} D_0 \right) + \frac{9 L^2}{\alpha \reg} \left(1+\frac{4L^2}{\alpha^2}\right) \Var_{\bar \nu^{Z}}(\bar Z) 
        + \frac{6 \eta dL^2 N}{\alpha} \left( 22 + 384 \frac{L^2}{\alpha^2} \right) \\
        &\le e^{-\alpha \eta k} \left(H_0 + \frac{9 L^2}{\alpha \reg} D_0 \right) + \frac{45 L^4}{\alpha^3 \reg} \left(\Var_{\bar \nu^{Z}}(\bar Z) 
        + 55 \, \eta \reg d N \right)
    \end{align*} 
    where in the last step we use the bound $1 \le \frac{L^2}{\alpha^2}$ to simplify the second and third terms, and bound $6 \times (22+384) = 2436 \le 2475 = 45 \times 55$.
\end{proof}

%%%%%%%%%%%%%%%%%%%%%%%%%%%%%%
\subsection{Proof of Corollary~\ref{Cor:AvgParticleAlgorithm} (Average Particle Along the Finite-Particle Algorithm)}
\label{Sec:AvgParticleAlgorithmProof}

\begin{proof}[Proof of Corollary~\ref{Cor:AvgParticleAlgorithm}]
    The proof below follows identically as in the proof of Corollary~\ref{Cor:AvgParticleSystem}.

    For $\bz_k = (\bx_k,\by_k) = (x_k^1,\dots,x_k^N,y_k^1,\dots,y_k^N) \sim \rho_k^{\bz,\eta}$ in $\R^{2dN}$, let $\rho_k^{z,\eta,i} \in \P(\R^{2d})$ be the marginal distribution of the component $z_k^i = (x_k^i,y_k^i) \in \R^{2d}$, for $i \in [N]$.
    We rearrange the coordinates to write $\bz_k = (z_k^1,\dots,z_k^N)$ for convenience, and still denote its distribution by $\rho_k^{\bz,\eta}$. 
    We introduce an independent product of the marginal distributions:
    \begin{align*}
        \hat \rho_k^{\bz,\eta} := \bigotimes_{i \in [N]} \rho_k^{z,\eta,i}.
    \end{align*}    
    Since $\bar \nu^{\bZ} = (\bar \nu^X)^{\otimes N} \otimes (\bar \nu^Y)^{\otimes N}$ is an independent product, after rearranging the coordinates as above, we can write it as $\bar \nu^{\bZ} = (\bar \nu^Z)^{\otimes N}$ where $\bar \nu^Z = \bar \nu^X \otimes \bar \nu^Y$.
    
    Then we can bound the KL divergence by:
    \begin{align*}
        \KL(\rho_k^{\bz,\eta} \,\|\, \bar \nu^{\bZ})
        = \E_{\rho_k^{\bz,\eta}}\left[ \log \frac{\rho_k^{\bz,\eta}}{\bar \nu^{\bZ}} \right]
        &= \E_{\rho_k^{\bz,\eta}}\left[ \log \frac{\rho_k^{\bz,\eta}}{\hat \rho_k^{\bz,\eta}} \right] + \E_{\rho_k^{\bz,\eta}}\left[\log \frac{\hat \rho_k^{\bz,\eta}}{(\bar \nu^{Z})^{\otimes N}} \right] \\
        &= \E_{\rho_k^{\bz,\eta}}\left[ \log \frac{\rho_k^{\bz,\eta}}{\hat \rho_k^{\bz,\eta}} \right] + \sum_{i \in [N]} \E_{\rho_k^{z,\eta,i}} \left[\log \frac{\rho_k^{z,\eta,i}}{\bar \nu^{Z}} \right] \\
        &= \KL\left(\rho_k^{\bz,\eta} \,\|\, \hat \rho_k^{\bz,\eta}\right) + \sum_{i \in [N]} \KL\left(\rho_k^{z,\eta,i} \,\|\, \bar \nu^{Z}\right) \\
        &\ge \sum_{i \in [N]} \KL\left(\rho_k^{z,\eta,i} \,\|\, \bar \nu^{Z}\right)
    \end{align*}
    where the inequality follows by dropping the first term $\KL\left(\rho_k^{\bz,\eta} \,\|\, \hat \rho_k^{\bz,\eta}\right) \ge 0$, which is the multivariate mutual information of $\bz_k^\eta$. 
    Furthermore, since KL divergence is jointly convex in both arguments, we can further bound the above by:
    \begin{align*}
        \KL(\rho_k^{\bz,\eta} \,\|\, \bar \nu^{\bZ})
        &\ge N \cdot \frac{1}{N} \sum_{i \in [N]} \KL\left(\rho_k^{z,\eta,i} \,\|\, \bar \nu^{Z}\right) 
        \,\ge\, N \cdot \KL\left(\rho_k^{z,\eta,\avg} \,\|\, \bar \nu^{Z}\right)
    \end{align*}
    using the definition $\rho_k^{z,\eta,\avg} = \frac{1}{N} \sum_{i \in [N]} \rho_k^{z,\eta,i}$.
    
    Combining this with the bound for $\KL(\rho_k^{\bz,\eta} \,\|\, \bar \nu^{\bZ})$ from Theorem~\ref{Thm:ConvergenceAlgToMF} gives the desired result.
\end{proof}

%%%%%%%%%%%%%%%%%%%%%%%%%%%%%%
\subsection{Proofs for Iteration Complexity of the Finite-Particle Algorithm}

%%%%%%%%%%%%%%%%%%%%%%%%%%%%%%
\subsubsection{Preliminary results}
\label{Sec:BoundInitFI}

Recall $\bar \nu^Z = \bar \nu^X \otimes \bar \nu^Y$ is the stationary mean-field distribution, which is $(\alpha/\reg)$-SLC and $(L/\reg)$-smooth by Lemma~\ref{Lem:Equilibrium}, under Assumption~\ref{As:SCSmooth}.
Recall by Theorem~\ref{Thm:DetMinMaxGF} (in Section~\ref{Sec:ReviewDeterministic}) that there exists a unique pair of equilibrium points $(x^*,y^*) \in \R^{2d}$ that satisfies $\nabla V(x^*,y^*) = 0$.
We first show the following.

\begin{lemma}\label{Lem:BoundVar}
    Assume Assumption~\ref{As:SCSmooth}.
    Then for $\bar Z \sim \bar \nu^Z$ in $\R^{2d}$:
    $$\Var_{\bar \nu^Z}(\bar Z) \le \frac{2\reg d}{\alpha}.$$    
\end{lemma}
\begin{proof}
    Recall that since $\bar \nu^Z$ is $(\alpha/\reg)$-SLC, it also satisfies the $(\alpha/\reg)$-Poincar\'e inequality~\citep{villani2009optimal}, which means for all smooth functions $\phi \colon \R^{2d} \to \R$:
    $$\Var_{\bar \nu^Z}(\phi(\bar Z)) \le \frac{\reg}{\alpha} \, \E_{\bar \nu^Z}\left[\left\|\nabla \phi(\bar Z)\right\|^2\right].$$
    For each unit vector $u \in \R^{2d}$, $\|u\| = 1$, by applying the Poincar\'e inequality to the function $\phi(z) = \langle z, u \rangle$, we get:
    \begin{align*}
        u^\top \Cov_{\bar \nu^Z}(\bar Z) u
        = \Var_{\bar \nu^Z}(\langle \bar Z, u \rangle)
        \le \frac{\reg}{\alpha} \, \E_{\bar \nu^Z}\left[\left\|u\right\|^2\right] = \frac{\reg}{\alpha}. 
    \end{align*}
    This shows that the covariance matrix $\Cov_{\bar \nu^Z}(\bar Z) \in \R^{2d \times 2d}$ satisfies:
    $$\Cov_{\bar \nu^Z}(\bar Z) \preceq \frac{\reg}{\alpha} I.$$
    Taking trace gives us the desired result:
    $\Var_{\bar \nu^Z}(\bar Z) = \Tr\left(\Cov_{\bar \nu^Z}(\bar Z)\right) \le \frac{2\reg d}{\alpha}.$
\end{proof}

\begin{lemma}\label{Lem:BoundDistStationary}
    Assume Assumption~\ref{As:SCSmooth}.
    For $\bar Z \sim \bar \nu^Z$ in $\R^{2d}$, we have:
    $$\E_{\bar \nu^Z}\left[\left\|\bar Z-z^*\right\|^2\right] \le \frac{8\reg dL^2}{\alpha^3}.$$
\end{lemma}
\begin{proof}
    Define the vector field $b^Z \colon \R^{2d} \to \R^{2d}$ by, for all $z = (x,y) \in \R^{2d}$:
    \begin{align*}
        b^Z(x,y) = \begin{pmatrix}
            -\nabla_x V(x,y) \\
            \nabla_y V(x,y)
        \end{pmatrix}.
    \end{align*}
    Observe that $b^Z$ is the case $N=1$ of the vector field $b^{\bZ}$ that we defined in~\eqref{Eq:Defbz} (in $\R^{2dN} = \R^{2d}$).
    Note by definition, $b^Z(z^*) = 0$, since both $\nabla_x V(z^*) = \nabla_y V(z^*) = 0$.
    Recall by Lemma~\ref{Lem:LipschitzMonotone_bz} that $-b^Z$ is $\alpha$-strongly monotone, so for all $z \in \R^{2d}$:
    \begin{align*}
        \alpha \|z-z^*\|^2 \le \langle -b^Z(z)+b^Z(z^*), z-z^* \rangle \le \|b^Z(z)-b^Z(z^*)\| \cdot \|z-z^*\|
    \end{align*}
    where the second inequality is by Cauchy-Schwarz.
    Then we conclude that:
    $$\alpha \|z-z^*\| \le \|b^Z(z)-b^Z(z^*)\| = \|b^Z(z)\|.$$
    Therefore, for $\bar Z \sim \bar \nu^{Z}$:
    \begin{align*}
        \E_{\bar \nu^Z}\left[\left\|\bar Z-z^*\right\|^2\right]
        &\,\le\, \frac{1}{\alpha^2} \E_{\bar \nu^Z}\left[\left\|b^Z(\bar Z)\right\|^2\right] 
        \,\le\, \frac{2L^2 \Var_{\bar \nu^Z}(\bar Z)}{\alpha^2} + \frac{4\reg dL}{\alpha^2} 
        \,\le\, \frac{4\reg dL^2}{\alpha^3} + \frac{4\reg dL}{\alpha^2} 
        \,\le\, \frac{8\reg dL^2}{\alpha^3}
    \end{align*}
    where the second inequality is by Lemma~\ref{Lem:BoundFPMF} for the case $N=1$,
    the third inequality is by the bound on the variance of $\bar \nu^Z$ from Lemma~\ref{Lem:BoundVar},
    and the last inequality is by the bound $1 \le \frac{L}{\alpha}$.    
\end{proof}

%%%%%%%%%%%%%%%%%%%%%%%%%%%%%%
\subsubsection{Bound on the initial relative Fisher information}

We have the following bounds on the distances from the stationary mean-field distribution from a Gaussian starting distribution.

\begin{lemma}\label{Lem:BoundInitKL}
    Assume Assumption~\ref{As:SCSmooth}.
    Let $\rho^X = \N(m^X, \frac{\reg}{L}I)$ and $\rho^Y = \N(m^Y, \frac{\reg}{L}I)$ for arbitrary $m^X, m^Y \in \R^d$, and let $\rho^Z = \rho^X \otimes \rho^Y = \N(m^Z, \frac{\reg}{L}I)$ where $m^Z = (m^X, m^Y) \in \R^{2d}$.
    Then:
    \begin{align*}
        \FI(\rho^Z \,\|\, \bar \nu^Z) &\le \frac{22 dL^4}{\reg \alpha^3} + \frac{2L^2}{\reg^2} \|m^Z-z^*\|^2 \\
        \KL(\rho^Z \,\|\, \bar \nu^Z) &\le \frac{11 dL^4}{\alpha^4} + \frac{L^2}{\reg \alpha} \|m^Z-z^*\|^2  \\
        W_2(\rho^Z, \, \bar \nu^Z)^2 &\le \frac{22 \reg dL^4}{\alpha^5} + \frac{2L^2}{\alpha^2} \|m^Z-z^*\|^2.
    \end{align*}
\end{lemma}
\begin{proof}
    Since $\rho^Z = \rho^X \otimes \rho^Y$ and $\bar \nu^Z = \bar \nu^X \otimes \bar \nu^Y$ are product distributions, we have
    $$\FI(\rho^Z \,\|\, \bar \nu^Z) = \FI(\rho^X \,\|\, \bar \nu^X) + \FI(\rho^Y \,\|\, \bar \nu^Y).$$
    We will bound each term separately.

    Define $g^X = -\log \bar \nu^X$, so $\nabla g^X(x) = \reg^{-1} \E_{\bar \nu^Y}[\nabla_x V(x, \bar Y)]$
    and $\Delta g^X(x) = \reg^{-1} \E_{\bar \nu^Y}[\Delta_x V(x,\bar Y)]$.
    Since we assume $V(x,y)$ is $\alpha$-strongly convex in $x$, we have $\Delta g^X(x) \ge \alpha d/\reg \ge 0$ for all $x \in \R^d$.
    Note also that for $\rho^X = \N(m^X, \frac{\reg}{L}I)$ on $\R^d$, we have
    $$\E_{\rho^X}\left[\left\| \nabla \log \rho^X \right\|^2\right] = \frac{L^2}{\reg^2} \E_{\rho^X}\left[\left\|X-m^X\right\|^2\right] = \frac{Ld}{\reg}.$$
    Then by expanding the square and using integration by parts, we can write:
    \begin{align*}
        \FI(\rho^X \,\|\, \bar \nu^X)
        &= \E_{\rho^X}\left[\left\| \nabla \log \rho^X + \nabla g^X \right\|^2\right] \\
        &= \E_{\rho^X}\left[\left\| \nabla \log \rho^X \right\|^2\right] -2 \, \E_{\rho^X}\left[\Delta g^X\right] + \E_{\rho^X}\left[\left\|\nabla g^X \right\|^2\right] \\
        &\le \frac{Ld}{\reg} + \E_{\rho^X}\left[\left\|\nabla g^X \right\|^2\right].
    \end{align*}
    For the second term above, we can bound:
    \begin{align*}
        &\reg^2 \, \E_{\rho^X}\left[\left\|\nabla g^X \right\|^2\right] \\
        &= \E_{\rho^X}\left[\left\|\E_{\bar \nu^Y}[\nabla_x V(X,\bar Y)] \right\|^2\right] \\
        &\le \E_{\rho^X \otimes \bar \nu^Y}\left[\left\|\nabla_x V(X,\bar Y) \right\|^2\right] \\
        &\le 2\E_{\rho^X \otimes \bar \nu^Y}\left[\left\|\nabla_x V(X,\bar Y) - \nabla_x V(x^*, \bar Y) \right\|^2\right] + 2\E_{\rho^X \otimes \bar \nu^Y}\left[\left\|\nabla_x V(x^*, \bar Y) - \nabla_x V(x^*,y^*) \right\|^2\right] \\
        &\le 2L^2 \E_{\rho^X}\left[\left\|X-x^*\right\|^2\right] + 2L^2 \, \E_{\bar \nu^Y}\left[\left\|\bar Y-y^*\right\|^2\right] \\
        &= 2L^2 \left(\frac{\reg d}{L} + \|m^X-x^*\|^2 \right) + 2L^2 \, \E_{\bar \nu^Y}\left[\left\|\bar Y-y^*\right\|^2\right]
    \end{align*}
    In the above, the first inequality follows by Cauchy-Schwarz.
    In the second inequality, we introduce the additional term $\nabla_x V(x^*,\bar Y)$ and use the inequality $\|a+b\|^2 \le 2\|a\|^2 + 2\|b\|^2$, and also introduce $\nabla_x V(x^*,y^*) = 0$.
    In the third inequality, we use the property that $V$ is $L$-smooth, so $\nabla_x V$ is $L$-Lipschitz.
    The next step follows from the bias-variance decomposition.
    Combining the above calculations, we obtain:
    \begin{align*}
        \FI(\rho^X \,\|\, \bar \nu^X)
        &\le \frac{3Ld}{\reg} + \frac{2L^2}{\reg^2} \|m^X-x^*\|^2 + \frac{2L^2}{\reg^2} \, \E_{\bar \nu^Y}\left[\left\|\bar Y-y^*\right\|^2\right].
    \end{align*}

    By an identical argument, we can also bound:
    \begin{align*}
        \FI(\rho^Y \,\|\, \bar \nu^Y)
        &\le \frac{3L d}{\reg} + \frac{2L^2}{\reg^2} \|m^Y-y^*\|^2 + \frac{2L^2}{\reg^2} \, \E_{\bar \nu^X}\left[\left\|\bar X-x^*\right\|^2\right].
    \end{align*}
    Combining the two bounds above gives:
    \begin{align*}
        \FI(\rho^Z \,\|\, \bar \nu^Z)
        &\le \frac{6L d}{\reg} + \frac{2L^2}{\reg^2} \left(\|m^X-x^*\|^2 + \|m^Y-y^*\|^2\right) \\
        &\qquad + \frac{2L^2}{\reg^2} \left(\E_{\bar \nu^X}\left[\left\|\bar X-x^*\right\|^2\right] + \E_{\bar \nu^Y}\left[\left\|\bar Y-y^*\right\|^2\right]\right) \\
        &= \frac{6Ld}{\reg} + \frac{2L^2}{\reg^2} \|m^Z-z^*\|^2 + \frac{2L^2}{\reg^2} \, \E_{\bar \nu^Z}\left[\left\|\bar Z-z^*\right\|^2\right] \\
        &\le \frac{6Ld}{\reg} + \frac{2L^2}{\reg^2} \|m^Z-z^*\|^2 + \frac{16 dL^4}{\reg \alpha^3} \\
        &\le \frac{22 dL^4}{\reg \alpha^3} + \frac{2L^2}{\reg^2} \|m^Z-z^*\|^2 
    \end{align*}
    where the second inequality follows from Lemma~\ref{Lem:BoundDistStationary},
    and in the last inequality we use the bound $1 \le \frac{L}{\alpha}$ to simplify the result.

    Since $\bar \nu^Z$ is $(\alpha/\reg)$-SLC, it also satisfies $(\alpha/\reg)$-LSI and $(\alpha/\reg)$-TI, so we have:
    $$\KL(\rho^Z \,\|\, \bar \nu^Z) \le \frac{\reg}{2\alpha} \FI(\rho^Z \,\|\, \bar \nu^Z)$$
    and
    $$W_2(\rho^Z, \bar \nu^Z)^2 \le \frac{2\reg}{\alpha} \KL(\rho^Z \,\|\, \bar \nu^Z).$$
    Thus, the bounds for KL divergence and $W_2$ distance follow from the bound for relative Fisher information above.
\end{proof}

%%%%%%%%%%%%%%%%%%%%%%%%%%%%%%
\subsubsection{Proof of Corollary~\ref{Cor:ComplexityAlgToMF} (Iteration Complexity of the Finite-Particle Algorithm)}
\label{Sec:ComplexityAlgToMFProof}

\begin{proof}[Proof of Corollary~\ref{Cor:ComplexityAlgToMF}]
Fix any regularization parameter $\reg > 0$, and any small error threshold $\error > 0$.
We want to run the finite-particle algorithm~\eqref{Eq:ParticleAlgorithm3} with a sufficiently small step size $\eta$ and a sufficiently large number of particles $N$, for a sufficiently large number of iterations $k$, such that the upper bound on the KL divergence for the average particle in Corollary~\ref{Cor:AvgParticleAlgorithm} is less than $\error$.

We choose the parameters to make each term in the bound from Corollary~\ref{Cor:AvgParticleAlgorithm} less than $\frac{1}{3} \error$.
To do so, we can choose the step size to be:
\begin{align*}
    \eta = \frac{\error \, \alpha^3}{7500 \, d L^4}
    \qquad~\Rightarrow~\qquad
    2475 \frac{\eta d L^4}{\alpha^3} \,\le\, 2500 \frac{\eta d L^4}{\alpha^3} = \frac{\error}{3}.
\end{align*}
We assume $\error$ is small enough so that the choice of $\eta$ above satisfies the assumption $\eta \le \frac{\alpha}{64 L^2}$ in Corollary~\ref{Cor:AvgParticleAlgorithm};
this is ensured if $\error \le \frac{7500}{64} \frac{dL^2}{\alpha^2}$.

Recall the bound $\Var_{\bar \nu^{Z}}(\bar Z) \le \frac{2\reg d}{\alpha}$ from Lemma~\ref{Lem:BoundVar}.
We choose the number of particles to be:
\begin{align*}
    N \ge \frac{270 \, dL^4}{\error \, \alpha^4}
    \qquad~\Rightarrow~\qquad
    \frac{45 L^4}{\alpha^3 \reg N} \Var_{\bar \nu^{Z}}(\bar Z)
    \,\le\, \frac{90 \, d L^4}{\alpha^4 N}
    \,\le\, \frac{\error}{3}.
\end{align*}

Next, suppose we run the min-max gradient descent algorithm~\eqref{Eq:MinMaxGD} from $\tilde z_0 = (0,0) \in \R^{2d}$ with step size $\eta_{\GD} = \frac{\alpha}{4L^2}$ for the number of iterations 
$k_{\GD} \ge \frac{4L^2}{\alpha^2} \log \frac{\alpha^3 \|z^*\|^2}{\reg dL^2}$,
so that by Corollary~\ref{Cor:MinMaxGDIter}, we obtain a final point $m^Z := \tilde z_{k_\GD}$ which satisfies the guarantee:
$$\|m^Z - z^*\|^2 \le \frac{\reg dL^2}{\alpha^3}.$$
We consider the Gaussian distribution $\gamma^Z := \N(m^Z, \frac{\reg}{L} I)$.
By the bound from Lemma~\ref{Lem:BoundInitKL}, we have:
\begin{align*}
    \KL(\gamma^Z \,\|\, \bar \nu^Z) &\le \frac{11 dL^4}{\alpha^4} + \frac{L^2}{\reg \alpha} \|m^Z-z^*\|^2
    \,\le\, \frac{12 dL^4}{\alpha^4} \\
    W_2(\gamma^Z, \, \bar \nu^Z)^2 &\le \frac{22 \reg dL^4}{\alpha^5} + \frac{2L^2}{\alpha^2} \|m^Z-z^*\|^2
    \,\le\, \frac{24 \reg dL^4}{\alpha^5}.
\end{align*}
Therefore,
\begin{align*}
    \KL(\gamma^Z \,\|\, \bar \nu^{Z}) + \frac{9 L^2}{\alpha \reg} W_2(\gamma^Z, \, \bar \nu^{Z})^2
    \,\le\, \frac{12 dL^4}{\alpha^4} + \frac{9 L^2}{\alpha \reg} \cdot \frac{24 \reg dL^4}{\alpha^5}
    \,\le\, 228 \, \frac{dL^6}{\alpha^6}
\end{align*}
where in the last inequality we use the bound $1 \le \frac{L}{\alpha}$ to simplify the result, and $12 + 9 \times 24 = 228$.
We use this to initialize the finite-particle algorithm~\eqref{Eq:ParticleAlgorithm3} from the product distribution where each component is $\gamma^Z$:
$$\rho_0^{\bz,\eta} := \left(\gamma^Z\right)^{\otimes N}.$$
This means we start the algorithm from $\bz_0 = (z_0^{1}, \dots, z_0^{N}) \in \R^{2dN}$
where $z_0^{1}, \dots, z_0^{N} \sim \gamma^Z$ are i.i.d.
Since $\bar \nu^{\bZ} = (\bar \nu^Z)^{\otimes N}$ is also a product distribution, the KL divergence and $W_2$ distance split:
\begin{align*}
    \frac{1}{N} \left(\KL(\rho_0^{\bz,\eta} \,\|\, \bar \nu^{\bZ}) + \frac{9 L^2}{\alpha \reg} W_2(\rho_0^{\bz,\eta}, \bar \nu^{\bZ})^2 \right)
    &= \KL(\gamma^Z \,\|\, \bar \nu^{Z}) + \frac{9 L^2}{\alpha \reg} W_2(\gamma^Z, \, \bar \nu^{Z})^2
    \le 228 \, \frac{dL^6}{\alpha^6}.
\end{align*}
Therefore, we can choose the number of iterations $k$ of the finite-particle algorithm~\eqref{Eq:ParticleAlgorithm3} to be:
\begin{align}\label{Eq:CorIter}
    k \,\ge\, \frac{1}{\alpha \eta} \log \frac{3 \cdot 228 \, dL^6}{\error \, \alpha^6} 
    \,=\, \frac{7500 \, d L^4}{\error \, \alpha^4} \log \frac{684 \, dL^6}{\error \, \alpha^6}
\end{align}
so that
$$\frac{e^{-\alpha \eta k}}{N} \left(\KL(\rho_0^{\bz,\eta} \,\|\, \bar \nu^{\bZ}) + \frac{9 L^2}{\alpha \reg} W_2(\rho_0^{\bz,\eta}, \bar \nu^{\bZ})^2 \right) \le \frac{\error}{3}.$$

Combining all of the above, we conclude that if we run the finite-particle algorithm~\eqref{Eq:ParticleAlgorithm3} with the above choice of step size $\eta$ and number of particles $N$, from the initial distribution $\rho_0^{\bz,\eta} = \left(\gamma^Z\right)^{\otimes N}$,
then after $k$ iterations given by~\eqref{Eq:CorIter}, the average particle $z_k^I \sim \rho_k^{z,\eta,\avg}$ satisfies, by Corollary~\ref{Cor:AvgParticleAlgorithm}:
$$\KL(\rho_k^{z,\eta,\avg} \,\|\, \bar \nu^{Z}) \le \frac{\error}{3} + \frac{\error}{3} + \frac{\error}{3} = \error$$
as desired. 
\end{proof}

%%%%%%%%%%%%%%%%%%%%%%
\section{Convergence of Finite-Particle Systems to Their Limiting Distributions}

We study the convergence guarantees of the finite-particle dynamics~\eqref{Eq:ParticleSystem3} and algorithm~\eqref{Eq:ParticleAlgorithm3} to their limiting stationary distributions.

%%%%%%%%%%%%%%%%%%
\subsection{Preliminary Results}

%%%%%%%%%%%%%%%%%%
\subsubsection{Transformation of LSI constant}

We recall the following classical result on how the LSI constant of a probability distribution changes under a pushforward operation by a Lipschitz map.

\begin{lemma}\cite[Remark~7]{chafai2004entropies}\label{Lem:LSIPushforward}
    Suppose $\nu \in \P(\R^D)$ satisfies $\alpha$-LSI for some $\alpha > 0$.
    Let $T \colon \R^D \to \R^D$ be a differentiable map which is $M$-Lipschitz for some $0 < M < \infty$.
    Then the pushforward distribution $ \tilde \nu = T_{\#}\nu$ satisfies $(\alpha/M^2)$-LSI.
\end{lemma}

We also recall the following result on how the LSI constant changes under a convolution.

\begin{lemma}\label{Lem:LSIConvolution}
    Suppose $\nu \in \P(\R^D)$ satisfies $\alpha$-LSI for some $\alpha > 0$.
    For $t > 0$, the probability distribution $\nu_t := \nu \ast \N(0, tI)$ satisfies $\alpha_t$-LSI where $\alpha_t = (\frac{1}{\alpha} + t)^{-1}$.
\end{lemma}
\begin{proof}
    We recall by~\citep[Corollary~3.1]{chafai2004entropies} that if $\rho$ satisfies $\alpha_\rho$-LSI and $\nu$ satisfies $\alpha_\nu$-LSI, then the convolution $\rho \ast \nu$ satisfies LSI with constant $(\frac{1}{\alpha_\rho} + \frac{1}{\alpha_\nu})^{-1}$.
    By assumption, $\nu$ satisfies $\alpha$-LSI.
    Since the Gaussian distribution $\N(0,tI)$ is $(1/t)$-SLC, it satisfies $(1/t)$-LSI.
    Then by the cited result above, the convolution $\nu_t = \nu \ast \N(0, tI)$ satisfies LSI with constant $(\frac{1}{\alpha} + t)^{-1}$, as claimed.
\end{proof}

We recall the following property that KL divergence is preserved under a deterministic map.

\begin{lemma}\label{Lem:KLPreserved}
    Let $T \colon \R^D \to \R^D$ be a deterministic, differentiable bijective map.
    For any probability distributions $\rho, \nu \in \P(\R^D)$:
    $$\KL(T_\# \rho \,\|\, T_\# \nu) = \KL(\rho \,\|\, \nu).$$
\end{lemma}
\begin{proof}
    This follows from a direct computation using the change-of-variable formula for $T_\# \rho$ and $T_\# \nu$.
    Alternatively, this follows from two applications of the data processing inequality from information theory, applied to the channels $T$ and $T^{-1}$:
    \begin{align*}
        \KL(T_\# \rho \,\|\, T_\# \nu)
        &\le \KL(\rho \,\|\, \nu) 
        = \KL((T^{-1})_\# (T_\#\rho) \,\|\, (T^{-1})_\# (T_\#\nu)) 
        \le \KL(T_\# \rho \,\|\, T_\# \nu).        
    \end{align*}
    Hence, both inequalities above must be equality.    
\end{proof}

%%%%%%%%%%%%%%%%%%
\subsubsection{Contraction of the deterministic step of the finite-particle algorithm}

The following shows that the deterministic step in the finite-particle algorithm is a contraction.

\begin{lemma}\label{Lem:Contraction_bz}
    Assume Assumption~\ref{As:SCSmooth}.
    For $\eta > 0$, define the map $G \colon \R^{2dN} \to \R^{2dN}$ by, for $\bz \in \R^{2dN}$:
    $$G(\bz) = \bz + \eta b^{\bZ}(\bz)$$
    where $b^{\bZ} \colon \R^{2dN} \to \R^{2dN}$ is the vector field defined in~\eqref{Eq:Defbz}.
    If $\eta \le \frac{\alpha}{2L^2}$, then $G$ is $M$-Lipschitz, where:
    $$M := \sqrt{1-2\eta \alpha + 4 \eta^2 L^2} \in [0,1].$$
\end{lemma}
\begin{proof}
    For any $\bz, \bar \bz \in \R^{2dN}$, we can compute:
    \begin{align*}
        \left\|G(\bz)-G(\bar\bz)\right\|^2
        &= \left\|\bz-\bar\bz\right\|^2 + 2\eta\left\langle \bz-\bar\bz, b^{\bZ}(\bz)-b^{\bZ}(\bar\bz) \right\rangle + \eta^2 \left\|b^{\bZ}(\bz)-b^{\bZ}(\bar\bz)\right\|^2 \\
        &\le \left\|\bz-\bar\bz\right\|^2 - 2\eta \alpha \left\| \bz-\bar\bz \right\|^2 + 4\eta^2 L^2 \left\|\bz-\bar\bz\right\|^2 \\
        &= (1-2\eta \alpha + 4 \eta^2 L^2) \left\|\bz-\bar\bz\right\|^2.
    \end{align*}
    In the above, in the first step we expand the square.
    In the second step, we use the properties from Lemma~\ref{Lem:LipschitzMonotone_bz} that $b^{\bZ}$ is $(2L)$-Lipschitz, and $-b^{\bZ}$ is $\alpha$-strongly monotone.
    In the last step, we collect the terms.    
    Note that $1-2\eta \alpha + 4 \eta^2 L^2 \in [0,1]$ from our assumption that $\eta \le \frac{\alpha}{2L^2} \le \frac{1}{2\alpha}$.
\end{proof}

%%%%%%%%%%%%%%%%%
\subsubsection{Bounds on the second moment along the finite-particle systems}
\label{Sec:BoundSecondMomentParticle}

In this section, we show that along the finite-particle dynamics~\eqref{Eq:ParticleSystem3} and algorithm~\eqref{Eq:ParticleAlgorithm3}, the distributions remain in $\P(\R^{2dN})$.

Under Assumption~\ref{As:SCSmooth}, recall from Theorem~\ref{Thm:DetMinMaxGF} (in Section~\ref{Sec:ReviewDeterministic}) that there exists a unique equilibrium point $z^* = (x^*,y^*) \in \R^{2d}$ which satisfies $\nabla V(z^*) = 0$.
Define $\bz^* = (x^*,\dots,x^*,y^*,\dots,y^*) \in \R^{2dN}$.
Then by construction, $b^{\bZ}(\bz^*) = 0$, where $b^{\bZ}$ is the vector field defined in~\eqref{Eq:Defbz}.

%%%%%%%
\begin{lemma}\label{Lem:BoundMomentParticleDynamics}
    Assume Assumption~\ref{As:SCSmooth}.
    Suppose $\bZ_t \sim \rho_t^{\bZ}$ evolves following the finite-particle dynamics~\eqref{Eq:ParticleSystem3} in $\R^{2dN}$ from $\bZ_0 \sim \rho_0^{\bZ} \in \P(\R^{2dN})$.
    Then $\rho_t^{\bZ} \in \P(\R^{2dN})$ for all $t \ge 0$.
\end{lemma}
\begin{proof}
    We note that $\rho_t^{\bZ}$ is absolutely continuous with respect to the Lebesgue measure on $\R^{2dN}$ by virtue of the Brownian motion component in the dynamics~\eqref{Eq:ParticleSystem3}.
    We now show that the second moment of $\rho_t^{\bZ}$ remains finite for all $t \ge 0$.
    Since $\bZ_t$ evolves following the dynamics~\eqref{Eq:ParticleSystem3}, its density $\rho_t^{\bZ}$ evolves following the Fokker-Planck equation:
    $$\part{\rho_t^{\bZ}}{t} = -\nabla \cdot \left(\rho_t^{\bZ} \, b^{\bZ}\right) + \reg \, \Delta \rho_t^{\bZ}.$$
    From this, we can compute, using integration by parts:
    \begin{align*}
        \frac{d}{dt} \E_{\rho_t^{\bZ}}\left[\left\|\bZ_t-\bz^*\right\|^2\right]
        &= \int_{\R^{2dN}} \part{\rho_t^{\bZ}(\bz)}{t} \, \|\bz-\bz^*\|^2 \, d\bz \\
        &= \int_{\R^{2dN}} \left(-\nabla \cdot \left(\rho_t^{\bZ} \, b^{\bZ}\right)(\bz) + \reg \, \Delta \rho_t^{\bZ}(\bz)\right) \|\bz-\bz^*\|^2 \, d\bz \\
        &= \int_{\R^{2dN}} \rho_t^{\bZ}(\bz) \left\langle b^{\bZ}(\bz), \nabla \left(\|\bz-\bz^*\|^2\right) \right \rangle \, d\bz + \reg \, \int_{\R^{2dN}} \rho_t^{\bZ}(\bz) \, \Delta \left(\|\bz-\bz^*\|^2\right) \, d\bz \\
        &= 2 \int_{\R^{2dN}} \rho_t^{\bZ}(\bz) \left\langle b^{\bZ}(\bz), \bz-\bz^* \right \rangle \, d\bz + \reg \, \int_{\R^{2dN}} \rho_t^{\bZ}(\bz) \, (4dN) \, d\bz \\
        &= 2 \int_{\R^{2dN}} \rho_t^{\bZ}(\bz) \left\langle b^{\bZ}(\bz) - b^{\bZ}(\bz^*), \bz-\bz^* \right \rangle \, d\bz + 4 \reg dN \\
        &\le -2\alpha \int_{\R^{2dN}} \rho_t^{\bZ}(\bz) \left\|\bz-\bz^* \right\|^2 \, d\bz + 4 \reg dN \\
        &= -2\alpha \, \E_{\rho_t^{\bZ}}\left[\left\|\bZ_t-\bz^*\right\|^2\right] + 4 \reg dN
    \end{align*}
    where the inequality above uses the property that $-b^{\bZ}$ is $\alpha$-strongly monotone, from Lemma~\ref{Lem:LipschitzMonotone_bz}.
    We can write the differential inequality above equivalently as:
    \begin{align*}
        \frac{d}{dt} \left(e^{2\alpha t} \, \E_{\rho_t^{\bZ}}\left[\left\|\bZ_t-\bz^*\right\|^2\right] \right)
        \le e^{2\alpha t} \, 4 \reg dN.
    \end{align*}
    Integrating from $0$ to $t$ and rearranging the result gives:
    \begin{align*}
        \E_{\rho_t^{\bZ}}\left[\left\|\bZ_t-\bz^*\right\|^2\right]
        &\le e^{-2\alpha t} \, \E_{\rho_0^{\bZ}}\left[\left\|\bZ_0-\bz^*\right\|^2\right] + \frac{(1-e^{-2\alpha t})}{\alpha} \, 2 \reg d N \\
        &\le \E_{\rho_0^{\bZ}}\left[\left\|\bZ_0-\bz^*\right\|^2\right] + \frac{2\reg d N}{\alpha}.
    \end{align*}
    Since $\rho_0^{\bZ} \in \P(\R^{2dN})$, $\E_{\rho_0^{\bZ}}\left[\left\|\bZ_0\right\|^2\right] < \infty$, so $\E_{\rho_0^{\bZ}}\left[\left\|\bZ_0-\bz^*\right\|^2\right] \le 2\E_{\rho_0^{\bZ}}\left[\left\|\bZ_0\right\|^2\right] + 2\|\bz^*\|^2 < \infty$.
    Therefore, we also have for all $t \ge 0$:
    \begin{align*}
        \E_{\rho_t^{\bZ}}\left[\left\|\bZ_t\right\|^2\right]
        &\le 2\E_{\rho_t^{\bZ}}\left[\left\|\bZ_t-\bz^*\right\|^2\right] + 2\|\bz^*\|^2 \\
        &\le 2 \, \E_{\rho_0^{\bZ}}\left[\left\|\bZ_0-\bz^*\right\|^2\right] + \frac{4\reg d N}{\alpha} + 2\|\bz^*\|^2 
        \,<\, \infty
    \end{align*}
    which shows that $\rho_t^{\bZ} \in \P(\R^{2dN})$.
\end{proof}

%%%%%%
\begin{lemma}\label{Lem:BoundMomentParticleAlgorithm}
    Assume Assumption~\ref{As:SCSmooth}.
    Suppose $\bz_k \sim \rho_k^{\bz,\eta}$ evolves via the finite-particle algorithm~\eqref{Eq:ParticleAlgorithm3} with step size $0 < \eta < \frac{\alpha}{2L^2}$ from $\bz_0 \sim \rho_0^{\bz,\eta} \in \P(\R^{2dN})$.
    Then $\rho_k^{\bz,\eta} \in \P(\R^{2dN})$ for all $k \ge 0$.
\end{lemma}
\begin{proof}
    We note that $\rho_k^{\bz,\eta}$ is absolutely continuous with respect to the Lebesgue measure on $\R^{2dN}$ since the update rule of the algorithm~\eqref{Eq:ParticleAlgorithm3} adds an independent Gaussian random variable, which corresponds to convolution with the Gaussian distribution.
    We now show that the second moment of $\rho_k^{\bz,\eta}$ remains finite for all $k \ge 0$.

    Define $G \colon \R^{2d} \to \R^{2d}$ by $G(\bz) = \bz + \eta b^{\bZ}(\bz)$.
    By Lemma~\ref{Lem:Contraction_bz}, we know $G$ is $M$-Lipschitz where $M = \sqrt{1-2\eta \alpha + 4 \eta^2 L^2},$
    and note that $0 < M < 1$ from the assumption $\eta < \frac{\alpha}{2L^2} \le \frac{1}{2\alpha}$.
    Note also that by definition, $G(\bz^*) = \bz^*$ since $b^{\bZ}(\bz^*) = 0$.
    We have the update rule from~\eqref{Eq:ParticleAlgorithm3}:
    \begin{align*}
        \bz_{k+1} &= \bz_k + \eta b^{\bZ}(\bz_k) + \sqrt{2\reg\eta} \, \bm{\zeta}_k^{\bz}
        = G(\bz_k) + \sqrt{2\reg\eta} \, \bm{\zeta}_k^{\bz}
    \end{align*}
    where $\bm{\zeta}_k^{\bz} \sim \N(0,I)$ is independent of $\bz_k$.
    Then we can compute:
    \begin{align*}
        \E\left[\left\|\bz_{k+1} - \bz^*\right\|^2\right]
        &= \E\left[\left\|G(\bz_k) - G(\bz^*) + \sqrt{2\reg\eta} \, \bm{\zeta}_k^{\bz}\right\|^2\right] \\
        &= \E\left[\left\|G(\bz_k) - G(\bz^*)\right\|^2\right] + 2\reg\eta \, \E\left[\left\|\bm{\zeta}_k^{\bz}\right\|^2\right] \\
        &\le M^2 \, \E\left[\left\|\bz_k - \bz^*\right\|^2\right] + 4\reg\eta dN
    \end{align*}
    where the inequality follows from the property that $G$ is $M$-Lipschitz.
    Since $0 < M < 1$, we can iterate the recurrence above to obtain:
    \begin{align*}
        \E\left[\left\|\bz_{k} - \bz^*\right\|^2\right]
        &\le M^{2k} \, \E\left[\left\|\bz_{0} - \bz^*\right\|^2\right] + 4\reg\eta dN \sum_{i=0}^{k-1} M^{2i} 
        \,\le\, \E\left[\left\|\bz_{0} - \bz^*\right\|^2\right] + \frac{4\reg\eta dN}{1-M^2}.
    \end{align*}
    Since $\rho_0^{\bz,\eta} \in \P(\R^{2dN})$, $\E_{\rho_0^{\bz,\eta}}\left[\left\|\bz_0\right\|^2\right] < \infty$, so $\E_{\rho_0^{\bz,\eta}}\left[\left\|\bz_0-\bz^*\right\|^2\right] \le 2\E_{\rho_0^{\bz,\eta}}\left[\left\|\bz_0\right\|^2\right] + 2\|\bz^*\|^2 <~\infty$.
    Therefore, we also have for all $k \ge 0$:
    \begin{align*}
        \E_{\rho_k^{\bz,\eta}}\left[\left\|\bz_k\right\|^2\right]
        &\le 2\E_{\rho_k^{\bz,\eta}}\left[\left\|\bz_k-\bz^*\right\|^2\right] + 2\|\bz^*\|^2 \\
        &\le 2 \, \E_{\rho_0^{\bz,\eta}}\left[\left\|\bz_0-\bz^*\right\|^2\right] + \frac{8\reg\eta dN}{1-M^2} + 2\|\bz^*\|^2 
        \,<\, \infty
    \end{align*}
    which shows that $\rho_k^{\bz,\eta} \in \P(\R^{2dN})$.
\end{proof}

%%%%%%%%%%%%%%%%%%%
\subsection{Convergence of Finite-Particle Dynamics to Its Stationary Distribution}
\label{Sec:DynamicsToBiasedLimit}

We show the following exponential convergence rates of the finite-particle dynamics~\eqref{Eq:ParticleSystem3} to its stationary distribution $\rho_\infty^{\bZ}$ in $W_2$ distance and in KL divergence.
We note that the exponential convergence rate in KL divergence guarantee below can be strengthened to hold for all R\'enyi divergence, but we only present the result for KL divergence here for simplicity.

\begin{theorem}\label{Thm:DynamicsToBiasedLimit}
    Assume Assumption~\ref{As:SCSmooth}.
    There exists a unique stationary distribution $\rho_\infty^{\bZ} \in \P(\R^{2dN})$ of the finite-particle dynamics~\eqref{Eq:ParticleSystem3}, and it satisfies $(\alpha/\reg)$-LSI.
    Furthermore, suppose $\bZ_t \sim \rho_t^{\bZ}$ evolves following the finite-particle dynamics~\eqref{Eq:ParticleSystem3} in $\R^{2dN}$ from $\rho_0^{\bZ} \in \P(\R^{2dN})$.
    For all $t \ge 0$:
    \begin{align*}
        W_2(\rho_t^{\bZ}, \rho_\infty^{\bZ})^2 &\le e^{-2\alpha t} \, W_2(\rho_0^{\bZ}, \rho_\infty^{\bZ})^2 \\
        \KL(\rho_t^{\bZ} \,\|\, \rho_\infty^{\bZ}) &\le e^{-2\alpha t} \, \KL(\rho_0^{\bZ} \,\|\, \rho_\infty^{\bZ}).
    \end{align*}
\end{theorem}
\begin{proof}
    We showed in Lemma~\ref{Lem:BoundMomentParticleDynamics} that since $\rho_0^{\bZ} \in \P(\R^{2dN})$, we have $\rho_t^{\bZ} \in \P(\R^{2dN})$ for all $t \ge 0$. 

    \medskip
    \noindent
    \textbf{(1) Exponential contraction in $W_2$ distance:}
    Suppose we run two copies of the finite-particle dynamics~\eqref{Eq:ParticleAlgorithm3} from $\bZ_0 \sim \rho_0^{\bZ}$ and $\tilde \bZ_0 \sim \tilde \rho_0^{\bZ}$, where $(\bZ_0, \tilde \bZ_0)$ has the joint distribution which is the optimal $W_2$ coupling between $\rho_0^{\bZ}$ and $\tilde \rho_0^{\bZ}$,
    to get $\bZ_t \sim \rho_t^{\bZ}$ and $\tilde \bZ_t \sim \tilde \rho_t^{\bZ}$, for $t \ge 0$.
    We can write the two stochastic processes using synchronous coupling:
    \begin{align*}
        d\bZ_t &= b^{\bZ}(\bZ_t) \, dt + \sqrt{2\reg} \, dW_t^{\bZ} \\
        d\tilde \bZ_t &= b^{\bZ}(\tilde \bZ_t) \, dt + \sqrt{2\reg} \, dW_t^{\bZ}
    \end{align*}
    we use the same standard Brownian motion $(W_t^{\bZ})_{t \ge 0}$ in $\R^{2dN}$.
    Then we have:
    \begin{align*}
        d(\bZ_t-\tilde \bZ_t) = \left(b^{\bZ}(\bZ_t)-b^{\bZ}(\tilde \bZ_t)\right) dt
    \end{align*}
    where the Brownian motion terms cancel because they are equal in the two processes.
    From this, we get:
    \begin{align*}
        d\left\|\bZ_t-\tilde \bZ_t\right\|^2 = 2\left\langle \bZ_t-\tilde \bZ_t, \, b^{\bZ}(\bZ_t)- b^{\bZ}(\tilde \bZ_t) \right \rangle dt.
    \end{align*}
    Taking expectation over the joint distribution, we obtain:
    \begin{align*}
        \frac{d}{dt} \E\left[ \left\|\bZ_t-\tilde \bZ_t\right\|^2 \right]
        &= 2 \E\left[\left\langle \bZ_t-\tilde \bZ_t, \, b^{\bZ}(\bZ_t)-b^{\bZ}(\tilde \bZ_t) \right \rangle \right]
        \,\le\, -2\alpha \,\E\left[ \left\|\bZ_t-\tilde \bZ_t\right\|^2 \right]
    \end{align*}
    where in the last step we use the property from Lemma~\ref{Lem:LipschitzMonotone_bz} that $-b^{\bZ}$ is $\alpha$-strongly monotone.
    Integrating the differential inequality above from $0$ to $t$, and using the fact that $(\bZ_0, \tilde \bZ_0)$ has the optimal $W_2$ coupling gives:
    \begin{align*}
        \E\left[ \left\|\bZ_t-\tilde \bZ_t\right\|^2 \right]
        \,\le\, e^{-2\alpha t} \, \E\left[ \left\|\bZ_0-\tilde \bZ_0\right\|^2 \right]
        \,=\, e^{-2\alpha t} \, W_2\left(\rho_0^{\bZ}, \, \tilde \rho_0^{\bZ}\right)^2.
    \end{align*}
    Using the definition of the $W_2$ distance as the infimum over all coupling, we conclude that
    \begin{align*}
        W_2\left(\rho_t^{\bZ}, \, \tilde \rho_t^{\bZ}\right)^2
        \,\le\, e^{-2\alpha t} \, W_2\left(\rho_0^{\bZ}, \, \tilde \rho_0^{\bZ}\right)^2.
    \end{align*}
    This implies $\lim_{t \to \infty} W_2\left(\rho_t^{\bZ}, \, \tilde \rho_t^{\bZ}\right)^2 = 0$.
    Since this holds for any initial distributions $\rho_0^{\bZ}, \, \tilde \rho_0^{\bZ}$, this shows that there must be a stationary distribution $\rho_\infty^{\bZ}$ of the dynamics, and it is unique by the contraction property above.
    Then plugging in $\tilde \rho_t^{\bZ} = \tilde \rho_0^{\bZ} = \rho_\infty^{\bZ}$ yields the desired exponential convergence guarantee in $W_2$ distance.

    \medskip
    \noindent
    \textbf{(2) Isoperimetry of the stationary distribution:}
    As $\eta \to 0$, the finite-particle algorithm~\eqref{Eq:ParticleAlgorithm3} with step size $\eta$ recovers the continuous-time finite-particle dynamics~\eqref{Eq:ParticleSystem3}.
    In particular, as $\eta \to 0$, the stationary distribution $\rho_\infty^{\bz,\eta}$ of the finite-particle algorithm recovers the stationary distribution $\rho_\infty^{\bZ}$ of the finite-particle dynamics.
    We show in Theorem~\ref{Thm:AlgorithmToBiasedLimit} below that $\rho_\infty^{\bz,\eta}$ satisfies $\alpha_\eta = \frac{(\alpha-2\eta L^2)}{\reg}$-LSI, which implies $\rho_\infty^{\bZ} = \lim_{\eta \to 0} \rho_\infty^{\bz,\eta}$ satisfies LSI with constant $\lim_{\eta \to 0} \alpha_\eta = \alpha/\reg$, as claimed.

    \medskip
    \noindent
    \textbf{(3) Exponential convergence in KL divergence:}
    We run two copies of the finite-particle dynamics~\eqref{Eq:ParticleAlgorithm3} from $\bZ_0 \sim \rho_0^{\bZ}$ and from the stationary distribution $\tilde \bZ_0 \sim \tilde \rho_0^{\bZ} = \rho_\infty^{\bZ}$, 
    to get $\bZ_t \sim \rho_t^{\bZ}$ and $\tilde \bZ_t \sim \tilde \rho_t^{\bZ} = \rho_\infty^{\bZ}$, for $t \ge 0$.
    Then $\rho_t^{\bZ}$ and $\tilde \rho_t^{\bZ}$ evolve following the Fokker-Planck equations:
    \begin{align*}
        \part{\rho_t^{\bZ}}{t}
        &= -\nabla \cdot\left(\rho_t^{\bZ} \, b^{\bZ} \right) + \reg \, \Delta \rho_t^{\bZ} \\
        \part{\tilde \rho_t^{\bZ}}{t}
        &= -\nabla \cdot\left(\tilde \rho_t^{\bZ} \, b^{\bZ} \right) + \reg \, \Delta \tilde \rho_t^{\bZ}.
    \end{align*}
    Using the identity from Lemma~\ref{Lem:ddtKLDiv}, we can compute:
    \begin{align*}
        \frac{d}{dt} \KL(\rho_t^{\bZ} \,\|\, \tilde \rho_t^{\bZ})
        &= -\reg \, \FI(\rho_t^{\bZ} \,\|\, \tilde \rho_t^{\bZ}) 
        \,\le\, -2\alpha \, \KL(\rho_t^{\bZ} \,\|\, \tilde \rho_t^{\bZ})
    \end{align*}
    where in the second step we use the fact that $\tilde \rho_t^{\bZ} = \rho_\infty^{\bZ}$ satisfies $(\alpha/\reg)$-LSI.
    Integrating the differential inequality above from $0$ to $t$ yields:
    \begin{align*}
        \KL(\rho_t^{\bZ} \,\|\, \tilde \rho_t^{\bZ})
        \,\le\, e^{-2\alpha t} \, \KL(\rho_0^{\bZ} \,\|\, \tilde \rho_0^{\bZ}).
    \end{align*}
    Substituting $\tilde \rho_t^{\bZ} = \tilde \rho_0^{\bZ} = \rho_\infty^{\bZ}$ yields the desired exponential convergence rate in KL divergence.
\end{proof}

%%%%%%%%%%%%%%%%%
\subsection{Convergence of Finite-Particle Algorithm to Its Stationary Distribution}
\label{Sec:AlgorithmToBiasedLimit}

We show the following exponential convergence rates of the finite-particle algorithm~\eqref{Eq:ParticleAlgorithm3} to its stationary limiting distribution $\rho_\infty^{\bz,\eta}$.
We note that the convergence in KL divergence below can be strengthened to be an exponential convergence in all R\'enyi divergence, similar to the result for the Unadjusted Langevin Algorithm shown in~\citep{VW19}; for simplicity, here we only present the convergence results for $W_2$ distance and KL divergence.

We also note that similar to the standard discretization of the Langevin dynamics as the Unadjusted Langevin Algorithm which is biased, here the stationary distribution $\rho_\infty^{\bz,\eta}$ of the finite-particle algorithm~\eqref{Eq:ParticleAlgorithm3} is not equal to the stationary distribution $\rho_\infty^{\bZ}$ of the continuous-time finite-particle dynamics~\eqref{Eq:ParticleSystem3}.
We can also characterize the biased convergence of the finite-particle algorithm~\eqref{Eq:ParticleAlgorithm3} to the continuous-time stationary distribution $\rho_\infty^{\bZ}$, but we skip this part because in principle we do not actually care about the convergence to $\rho_\infty^{\bZ}$, but only about the biased convergence to the stationary mean-field distribution $\bar \nu^{\bZ}$, which we show in Theorem~\ref{Thm:ConvergenceAlgToMF} in the main text of this paper.

\begin{theorem}\label{Thm:AlgorithmToBiasedLimit}
    Assume Assumption~\ref{As:SCSmooth} and $0 < \eta < \frac{\alpha}{2L^2}$. 
    Then there exists a unique stationary distribution $\rho_\infty^{\bz,\eta} \in \P(\R^{2dN})$ of the finite-particle algorithm~\eqref{Eq:ParticleAlgorithm3} with step size $\eta$, and $\rho_\infty^{\bz,\eta}$ satisfies $\alpha_\eta$-LSI where
    $\alpha_\eta = (\alpha-2\eta L^2)/\reg$.
    Furthermore, suppose $\bz_k \sim \rho_k^{\bz,\eta}$ evolves following the finite-particle algorithm~\eqref{Eq:ParticleAlgorithm3} with step size $0 < \eta < \frac{\alpha}{2L^2}$ from $\rho_0^{\bz,\eta} \in \P(\R^{2dN})$.
    Define $M = \sqrt{1-2\eta \alpha + 4 \eta^2 L^2} \in (0,1)$.
    Then for all $k \ge 0$:
    \begin{align*}
        W_2(\rho_k^{\bz,\eta}, \,\rho_\infty^{\bz,\eta})^2 
        &\le M^{2k}  \, W_2(\rho_0^{\bz,\eta}, \, \rho_\infty^{\bz,\eta})^2 \\
        \KL(\rho_k^{\bz,\eta} \,\|\, \rho_\infty^{\bz,\eta}) 
        &\le M^{2k}  \, \KL(\rho_0^{\bz,\eta} \,\|\, \rho_\infty^{\bz,\eta}).
    \end{align*}
\end{theorem}
\begin{proof}
    We showed in Lemma~\ref{Lem:BoundMomentParticleAlgorithm} that since $\rho_0^{\bz,\eta} \in \P(\R^{2dN})$, we have $\rho_k^{\bz,\eta} \in \P(\R^{2dN})$ for all $k \ge 0$. 

    Define $G \colon \R^{2d} \to \R^{2d}$ by $G(\bz) = \bz + \eta b^{\bZ}(\bz)$.
    By Lemma~\ref{Lem:Contraction_bz}, we know $G$ is $M$-Lipschitz where 
    $M = \sqrt{1-2\eta \alpha + 4 \eta^2 L^2},$
    and note that $0 < M < 1$ from the assumption $\eta < \frac{\alpha}{2L^2} \le \frac{1}{2\alpha}$.

    \medskip
    \noindent
    \textbf{(1) Exponential contraction in $W_2$ distance:}
    Suppose we run two copies of the finite-particle algorithm~\eqref{Eq:ParticleAlgorithm3} with step size $0 < \eta < \frac{\alpha}{2L^2}$ from $\bz_0 \sim \rho_0^{\bz,\eta}$ and $\tilde \bz_0 \sim \tilde \rho_0^{\bz,\eta}$, where $(\bz_0, \tilde \bz_0)$ has the joint distribution which is the optimal $W_2$ coupling between $\rho_0^{\bz,\eta}$ and $\tilde \rho_0^{\bz,\eta}$,
    to get $\bz_k \sim \rho_k^{\bz,\eta}$ and $\tilde \bz_k \sim \tilde \rho_k^{\bz,\eta}$, for $k \ge 1$.
    We can write the update of the algorithm~\eqref{Eq:ParticleAlgorithm3} using synchronous coupling as:
    \begin{subequations}\label{Eq:ParticleAlgorithmUpdate}
    \begin{align}
        \bz_{k+1} &= G(\bz_k) + \sqrt{2\reg\eta} \, \bm{\zeta}_k, \\
        \tilde \bz_{k+1} &= G(\tilde \bz_k) + \sqrt{2\reg\eta} \, \bm{\zeta}_k
    \end{align}        
    \end{subequations}
    where we use the same Gaussian noise $\bm{\zeta}_k \sim \N(0,I)$ for both updates, independent of $\bz_k$ and $\tilde \bz_k$.
    Then we can compute:
    \begin{align*}
        \|\bz_{k+1}-\tilde\bz_{k+1}\|^2
        &= \|G(\bz_{k})-G(\tilde\bz_{k})\|^2 
        \,\le\, M^2 \|\bz_{k}-\tilde\bz_{k}\|^2
    \end{align*}
    where in the last step we have used the property that $G$ is $M$-Lipschitz, from Lemma~\ref{Lem:Contraction_bz}.
    Taking expectation over the joint distribution, this gives:
    \begin{align*}
        \E\left[\|\bz_{k+1}-\tilde\bz_{k+1}\|^2\right] \,\le\, M^2 \E\left[\|\bz_{k}-\tilde\bz_{k}\|^2\right].
    \end{align*}
    Unrolling the recursion and using the fact that $(\bz_0,\tilde \bz_0)$ has the optimal $W_2$ coupling, we get:
    \begin{align*}
        \E\left[\|\bz_{k}-\tilde\bz_{k}\|^2\right] \,\le\, M^{2k} \, \E\left[\|\bz_{0}-\tilde\bz_{0}\|^2\right]
        \,=\, M^{2k} \, W_2(\rho_0^{\bz,\eta}, \, \tilde \rho_0^{\bz,\eta})^2.
    \end{align*}
    Using the definition of the $W_2$ distance as the infimum over all coupling, this implies:
    \begin{align*}
        W_2(\rho_k^{\bz,\eta}, \, \tilde \rho_k^{\bz,\eta})^2
        \,\le\, M^{2k} \, W_2(\rho_0^{\bz,\eta}, \, \tilde \rho_0^{\bz,\eta})^2.
    \end{align*}
    Since $0 < M < 1$, this implies $\lim_{k \to \infty} W_2(\rho_k^{\bz,\eta}, \, \tilde \rho_k^{\bz,\eta})^2 = 0$.
    Since this holds for any initial distributions $\rho_0^{\bz,\eta}, \, \tilde \rho_0^{\bz,\eta}$, this shows that there must be a stationary distribution $\rho_\infty^{\bz,\eta}$ of the algorithm, and it is unique by the contraction property above.
    Then plugging in $\tilde \rho_k^{\bz,\eta} = \tilde \rho_0^{\bz,\eta} = \rho_\infty^{\bz,\eta}$ gives us the desired exponential convergence guarantee in $W_2$ distance.
    
    \medskip
    \noindent
    \textbf{(2) Isoperimetry of the stationary distribution:}
    We write each step of the finite-particle algorithm~\eqref{Eq:ParticleAlgorithm3} as a composition of a pushforward of the deterministic map $G$, followed by a Gaussian convolution, so at each iteration $k \ge 0$, the next distribution is given by:
    \begin{align}\label{Eq:ParticleAlgorithmUpdateDistribution}
        \rho_{k+1}^{\bz,\eta} = (G_\# \rho_k^{\bz,\eta}) \ast \N(0, 2\reg\eta I).
    \end{align}
    
    Suppose we run the algorithm~\eqref{Eq:ParticleAlgorithm3} from $\rho_0^{\bz,\eta}$ which satisfies $\beta_0$-LSI for some $0 < \beta_0 < \infty$
    (for example, a Gaussian distribution).
    Then inductively from the iteration above, for each $k \ge 1$, $\rho_k^{\bz,\eta}$ satisfies $\beta_k$-LSI, where the constants $(\beta_k)_{k \ge 0}$ satisfy the recurrence:
    \begin{align}\label{Eq:BetaRecurrence}
        \beta_{k+1} = \frac{1}{\frac{M^2}{\beta_k} + 2\reg\eta} = \frac{\beta_k}{M^2+2\reg\eta \beta_k}.
    \end{align}
    Define $\alpha_\eta$ by:
    $$\alpha_\eta := \frac{1-M^2}{2\reg\eta} = \frac{\alpha-2\eta L^2}{\reg}$$
    and observe that $\beta_k = \alpha_\eta$ is a fixed point of the recurrence~\eqref{Eq:BetaRecurrence}.
    Furthermore, we can rewrite the recurrence~\eqref{Eq:BetaRecurrence} as:
    \begin{align*}
        \frac{1}{\beta_{k+1}} - \frac{1}{\alpha_\eta} 
        &=
        \frac{M^2}{\beta_k} + \frac{2\reg\eta \alpha_\eta - 1}{\alpha_\eta} 
        \,=\,
        \frac{M^2}{\beta_k} - \frac{M^2}{\alpha_\eta} 
        \,=\, M^2 \left(\frac{1}{\beta_{k}} - \frac{1}{\alpha_\eta}\right).
    \end{align*}
    Unrolling the recurrence, we get:
    $$\frac{1}{\beta_{k}} = \frac{1}{\alpha_\eta} + M^{2k} \left(\frac{1}{\beta_{0}} - \frac{1}{\alpha_\eta}\right).$$
    As $M^2 < 1$, this shows that $1/\beta_{k}$ converges to $1/\alpha_\eta$ exponentially fast.
    Therefore, the stationary distribution $\rho_{\infty}^{\bz,\eta} = \lim_{k \to \infty} \rho_k^{\bz,\eta}$ satisfies LSI with constant $\lim_{k \to \infty} \beta_k = \alpha_\eta$, as desired.

    \medskip
    \noindent
    \textbf{(3) Exponential convergence in KL divergence:}
    We run two copies of the finite-particle algorithm~\eqref{Eq:ParticleAlgorithm3} with step size $0 < \eta < \frac{\alpha}{2L^2}$ from $\bz_0 \sim \rho_0^{\bz,\eta}$ and from the stationary distribution $\tilde \bz_0 \sim \tilde \rho_0^{\bz,\eta} = \rho_\infty^{\bz,\eta}$, to get $\bz_k \sim \rho_k^{\bz,\eta}$ and $\tilde \bz_k \sim \tilde \rho_k^{\bz,\eta} = \rho_\infty^{\bz,\eta}$.
    We write one step of the algorithm update as a pushforward followed by a Gaussian convolution, as in~\eqref{Eq:ParticleAlgorithmUpdateDistribution}.
    We define the half-steps:
    \begin{align*}
        \rho_{k+\frac{1}{2}}^{\bz,\eta} &= G_\# \rho_k^{\bz,\eta},
        \qquad\qquad
        \tilde \rho_{k+\frac{1}{2}}^{\bz,\eta} = G_\# \tilde \rho_k^{\bz,\eta}.   
    \end{align*}
    Since $G$ is a contraction ($M$-Lipschitz with $M < 1$), it is bijective, and so we know by Lemma~\ref{Lem:KLPreserved}:
    \begin{align*}
        \KL\left(\rho_{k+\frac{1}{2}}^{\bz,\eta} \,\big\|\, \tilde \rho_{k+\frac{1}{2}}^{\bz,\eta}\right)
        \,=\, \KL\left(\rho_{k}^{\bz,\eta} \,\|\, \tilde \rho_{k}^{\bz,\eta}\right).
    \end{align*}
    Furthermore, since $\tilde \rho_k^{\bz,\eta} = \rho_\infty^{\bz,\eta}$ satisfies $\alpha_\eta$-LSI, and $G$ is $M$-Lipschitz, we know by Lemma~\ref{Lem:LSIPushforward} that $\tilde \rho_{k+\frac{1}{2}}^{\bz,\eta}$ satisfies $(\alpha_\eta/M^2)$-LSI.
    Then we interpret the next half-step, the Gaussian convolution:
    \begin{align*}
        \rho_{k+1}^{\bz,\eta} &= \rho_{k+\frac{1}{2}}^{\bz,\eta} \ast \N(0, 2\reg\eta I),
        \qquad\qquad
        \tilde \rho_{k+1}^{\bz,\eta} = \tilde \rho_{k+\frac{1}{2}}^{\bz,\eta} \ast \N(0, 2\reg\eta I),   
    \end{align*}
    as the solutions of the heat flow:
    \begin{align*}
        \part{\rho_t}{t} &= \Delta \rho_t,
        \qquad\qquad
        \part{\nu_t}{t} = \Delta \nu_t
    \end{align*}
    starting from $\rho_0 := \rho_{k+\frac{1}{2}}^{\bz,\eta}$ and $\nu_0 := \tilde \rho_{k+\frac{1}{2}}^{\bz,\eta}$,
    for time $t = \reg\eta$, to get
    $\rho_{\reg\eta} = \rho_{k+1}^{\bz,\eta}$ and $\nu_{\reg\eta} = \tilde \rho_{k+1}^{\bz,\eta}$.
    Note since $\nu_0 = \tilde \rho_{k+\frac{1}{2}}^{\bz,\eta}$ satisfies $\gamma_0$-LSI with $\gamma_0 := (\alpha_\eta/M^2)$, by Lemma~\ref{Lem:LSIConvolution}, we know $\nu_t = \nu_0 \ast \N(0,2tI)$ satisfies $\gamma_t$-LSI where:
    $$\gamma_t = \frac{1}{\frac{1}{\gamma_0} + 2t} = \frac{\gamma_0}{1+2\gamma_0\,t} = \frac{1}{2} \frac{d}{dt} \log(1+2\gamma_0 \, t).$$
    Then by the formula from Lemma~\ref{Lem:ddtKLDiv}, we can compute:
    \begin{align*}
        \frac{d}{dt} \KL(\rho_t \,\|\, \nu_t) &= -\FI(\rho_t \,\|\, \nu_t) 
        \,\le\, -2\gamma_t \, \KL(\rho_t \,\|\, \nu_t).
    \end{align*}
    Integrating from $t=0$ to $t=\reg\eta$ gives:
    \begin{align*}
        \KL(\rho_{\reg\eta} \,\|\, \nu_{\reg\eta})
        &\le \exp\left(-2\int_0^{\reg\eta} \gamma_t \, dt \right) \, \KL(\rho_{0} \,\|\, \nu_{0})
        \,=\, \frac{\KL(\rho_{0} \,\|\, \nu_{0})}{1+2\gamma_0 \reg\eta}.
    \end{align*}
    Note that by definition,
    \begin{align*}
        1+2\gamma_0 \reg\eta
        \,=\, 1+\frac{2\alpha_\eta \reg\eta}{M^2} 
        \,=\, 1+ \frac{1-M^2}{M^2}
        \,=\, \frac{1}{M^2}.
    \end{align*}
    Then substituting back the definitions of $\rho_{\reg\eta},\rho_0$, and $\nu_{\reg\eta},\nu_0$ gives us:
    \begin{align*}
        \KL\left(\rho_{k+1}^{\bz,\eta} \,\|\, \tilde \rho_{k+1}^{\bz,\eta}\right)
        &\le \frac{\KL\left(\rho_{k+\frac{1}{2}}^{\bz,\eta} \,\big\|\, \tilde \rho_{k+\frac{1}{2}}^{\bz,\eta}\right)}{1+2\gamma_0 \reg\eta} 
        \,=\, 
        M^2 \,\KL\left(\rho_{k+\frac{1}{2}}^{\bz,\eta} \,\big\|\, \tilde \rho_{k+\frac{1}{2}}^{\bz,\eta}\right)
        \,=\, M^2 \,\KL\left(\rho_{k}^{\bz,\eta} \,\|\, \tilde \rho_{k}^{\bz,\eta}\right).        
    \end{align*}
    Iterating this recursion and substituting $\tilde \rho_{k+1}^{\bz,\eta} = \tilde \rho_{k}^{\bz,\eta} = \rho_\infty^{\bz,\eta}$ gives:
    \begin{align*}
        \KL\left(\rho_{k}^{\bz,\eta} \,\|\, \rho_\infty^{\bz,\eta} \right)
        \,\le\, M^{2k} \, \KL\left(\rho_{0}^{\bz,\eta} \,\|\, \rho_\infty^{\bz,\eta} \right)
    \end{align*}
    as desired.
\end{proof}

%%%%%%%%%%%%%%%%%%%%%%%%%
\section{Deterministic Zero-Sum Game}
\label{Sec:ReviewDeterministic}

We review the properties of the deterministic finite-dimensional zero-sum game:
\begin{align}\label{Eq:GameDet}
    \min_{x \in \R^d} \, \max_{y \in \R^d} \, V(x,y)
\end{align}
where $V \colon \R^d \times \R^d \to \R$ satisfies Assumption~\ref{As:SCSmooth}.
In particular, for each $x \in \R^d$, we have $\max_{y \in \R^d} V(x,y) < \infty$, and similarly, for each $y \in \R^d$, we have $\min_{x \in \R^d} V(x,y) > -\infty$,
so all the quantities we discuss below are finite.

Recall we say that $(x^*,y^*) \in \R^{2d}$ is an \textit{equilibrium point} (or a \textit{Nash equilibrium}) for the game~\eqref{Eq:GameDet} if for all $x,y \in \R^d$, the following holds:
\begin{align}\label{Eq:DetEquilibrium}
    V(x^*,y) \le V(x^*,y^*) \le V(x,y^*).
\end{align}
We recall the \textit{duality gap} $\DG \colon \R^d \times \R^d \to \R$ of the game~\eqref{Eq:GameDet} is defined by:
$$\DG(x,y) = \max_{y' \in \R^d} V(x,y') - \min_{x' \in \R^d} V(x',y).$$
Note that $\DG(x,y) \ge 0$ for all $(x,y) \in \R^{2d}$, and $\DG(x,y) = 0$ if and only if $(x,y) = (x^*,y^*)$ is an equilibrium point.
Under Assumption~\ref{As:SCSmooth}, there exists a unique equilibrium point $(x^*,y^*)$; furthermore, the duality gap is bounded by the squared gradient norm, see Theorem~\ref{Thm:DetMinMaxGF} below.

In this section, we review the exponential convergence guarantees of the continuous-time \textit{min-max gradient flow} in Section~\ref{Sec:ReviewDetGF}, and the discrete-time \textit{min-max gradient descent} in Section~\ref{Sec:ReviewDetGD}.

In Section~\ref{Sec:ReviewZeroEntropy}, we review the zero-sum game~\eqref{Eq:Game0} in the space of distributions without entropy regularization.
Under the same assumptions on $V$, we show that the unique equilibrium distribution of the game~\eqref{Eq:Game0} is a pure equilibrium, i.e., a point mass $(\delta_{x^*}, \delta_{y^*})$, where $(x^*,y^*)$ is the equilibrium point of the deterministic game from this section.

%%%%%%%%%%%%%%%%%%%%%%%%%%%%
\subsection{Convergence of Min-Max Gradient Flow}
\label{Sec:ReviewDetGF}

We consider the \textit{min-max gradient flow} which is the dynamics for $(X_t,Y_t)_{t \ge 0}$ in $\R^{2d}$ which evolves following:
\begin{align*}
    \dot X_t &= -\nabla_x V(X_t,Y_t) \\
    \dot Y_t &= \nabla_y V(X_t,Y_t). 
\end{align*}
In terms of the joint variable $Z_t = (X_t,Y_t) \in \R^{2d}$, we can write this as:
\begin{align}\label{Eq:MinMaxGF}
    \dot Z_t = b^Z(Z_t)
\end{align}
where we have defined the vector field $b^Z \colon \R^{2d} \to \R^{2d}$ by, for all $z = (x,y) \in \R^{2d}$:
\begin{align}\label{Eq:DetDefbz}
    b^Z(x,y) = \begin{pmatrix}
        -\nabla_x V(x,y) \\
        \nabla_y V(x,y)
    \end{pmatrix}.
\end{align}
Observe that $b^Z$ is the case $N=1$ of the vector field $b^{\bZ}$ that we defined in~\eqref{Eq:Defbz} (in $\R^{2dN} = \R^{2d}$).
In particular, the guarantees that we proved for $b^{\bZ}$ also apply to $b^Z$ when we specialize to $N=1$.

Below, for $z = (x,y) \in \R^{2d}$, we also write $\nabla V(z) \equiv \nabla V(x,y)$, and $\DG(z) \equiv \DG(x,y)$.

\begin{theorem}\label{Thm:DetMinMaxGF}
    Assume Assumption~\ref{As:SCSmooth}.
    Then we have the following properties:
    \begin{enumerate}
        \item There exists a unique equilibrium point $z^* = (x^*,y^*) \in \R^{2d}$, and it satisfies $\nabla V(z^*) = 0$.
        \item For all $z = (x,y) \in \R^{2d}$, the duality gap is bounded by the squared gradient norm:
        \begin{align}\label{Eq:DetDGBound}
            \DG(z) \le \frac{1}{2\alpha} \|\nabla V(z)\|^2.
        \end{align}
        \item Suppose $(Z_t)_{t \ge 0}$ evolves following the min-max gradient flow~\eqref{Eq:MinMaxGF} in $\R^{2d}$.
        For all $t \ge 0$:
        \begin{subequations}
        \begin{align}
            \|Z_t - z^*\|^2 &\le e^{-2\alpha t} \, \|Z_0 - z^*\|^2 \label{Eq:ConvGFDist} \\
            2\alpha \, \DG(Z_t) \le \|\nabla V(Z_t)\|^2 
            &\le e^{-2\alpha t} \, \|\nabla V(Z_0)\|^2 \label{Eq:ConvGFGrad}
        \end{align}            
        \end{subequations}
    \end{enumerate}
\end{theorem}
\begin{proof}
\textbf{(1) Bound in duality gap:}
For each $y \in \R^d$, since $x \mapsto V(x,y)$ is $\alpha$-strongly convex by assumption, it also satisfies $\alpha$-gradient domination, i.e., 
$$V(x,y) - \min_{x' \in \R^d} V(x',y) \le \frac{1}{2\alpha} \|\nabla_x V(x,y)\|^2.$$
Similarly, for each $x \in \R^d$, since $y \mapsto -V(x,y)$ is $\alpha$-strongly convex by assumption, it also satisfies $\alpha$-gradient domination, i.e., 
$$\max_{y' \in \R^d} V(x,y') - V(x,y) \le \frac{1}{2\alpha} \|\nabla_y V(x,y)\|^2.$$
Summing the two inequalities above gives:
\begin{align*}
    \DG(x,y) &= \max_{y' \in \R^d} V(x,y') - \min_{x' \in \R^d} V(x',y) \\
    &\le \frac{1}{2\alpha} \left(\|\nabla_x V(x,y)\|^2 + \|\nabla_y V(x,y)\|^2\right)
    = \frac{1}{2\alpha} \|\nabla V(x,y)\|^2
\end{align*}
as claimed in~\eqref{Eq:DetDGBound}.

\medskip
\noindent
\textbf{(2) Exponential convergence in distance and existence of stationary point:} 
Suppose we run two copies of the min-max gradient flow~\eqref{Eq:MinMaxGF} from $Z_0, Z_0' \in \R^{2d}$, to get $Z_t, Z_t' \in \R^{2d}$ for $t \ge 0$,
so they satisfy: $\dot Z_t = b^Z(Z_t)$ and $\dot Z_t' = b^Z(Z_t')$.
Then we can compute:
\begin{align*}
    \frac{d}{dt} \|Z_t-Z_t'\|^2 &= 2 \langle Z_t-Z_t', b(Z_t)-b(Z_t') \rangle 
    \le -2\alpha \|Z_t-Z_t'\|^2
\end{align*}
where the inequality follows from the property that $-b^Z$ is $\alpha$-strongly monotone, by Lemma~\ref{Lem:LipschitzMonotone_bz}.
Integrating the differential inequality above gives:
\begin{align*}
    \|Z_t-Z_t'\|^2 \le e^{-2\alpha t} \, \|Z_0-Z_0'\|^2.
\end{align*}
Therefore, $\lim_{t \to \infty} \|Z_t-Z_t'\|^2 = 0$.
Since this holds for any initial points $Z_0,Z_0'$, this means there must be a stationary point $z^*$, and furthermore, this stationary point $z^*$ is unique by the contraction property above.
Plugging in $Z_t' = Z_0' = z^*$ to the guarantee above gives the desired exponential convergence rate in distance~\eqref{Eq:ConvGFDist}.

\medskip
\noindent
\textbf{(3) Stationary point is an equilibrium point:}
Since $z^* = (x^*,y^*)$ is stationary for the min-max gradient flow~\eqref{Eq:MinMaxGF}, it makes the vector field vanish: $b^Z(z^*) = 0$, which means $\nabla_x V(x^*,y^*) = 0$ and $\nabla_y V(x^*,y^*) = 0$, so indeed $\nabla V(x^*, y^*) = 0$.

Furthermore, since $x \mapsto V(x,y^*)$
is strongly convex by assumption, $\nabla_x V(x^*,y^*) = 0$ means $x^* = \arg\min_{x \in \R^d} V(x,y^*)$, so $V(x^*,y^*) \le V(x,y^*)$ for all $x \in \R^d$.
Similarly, since $y \mapsto V(x^*,y)$ is strongly concave by assumption, $\nabla_y V(x^*,y^*) = 0$ means $y^* = \arg\max_{y \in \R^d} V(x^*,y)$, so $V(x^*,y^*) \ge V(x^*,y)$ for all $y \in \R^d$.
This shows that $z^* = (x^*,y^*)$ is an equilibrium point as defined in~\eqref{Eq:DetEquilibrium}.

\medskip
\noindent
\textbf{(4) Exponential convergence in gradient norm:}
Observe that $\|b^Z(Z_t)\|^2 = \|\nabla_x V(X_t,Y_t)\|^2 + \|\nabla_y V(X_t,Y_t)\|^2 = \|\nabla V(X_t,Y_t)\|^2$,
so it suffices to prove the exponential convergence of $\|b^Z(Z_t)\|$.
We can compute:
\begin{align*}
    \frac{d}{dt} \|b^Z(Z_t)\|^2
    &= 2 \langle b^Z(Z_t), \nabla b^Z(Z_t) \, \dot Z_t \rangle \\
    &= 2 \langle b^Z(Z_t), \nabla b^Z(Z_t) \, b^Z(Z_t) \rangle \\
    &= 2 \langle b^Z(Z_t), (\nabla b^Z(Z_t))_\sym \, b^Z(Z_t) \rangle
\end{align*}
where the last step follows since $u^\top A u = u^\top (A_\sym) u$ for all $u \in \R^{2d}$ and $A \in \R^{2d \times 2d}$, where $A_\sym = \frac{1}{2} (A+A^\top)$.
Recall from Lemma~\ref{Lem:LipschitzMonotone_bz} part (2), we have shown that $(\nabla b^Z(Z_t))_\sym \preceq -\alpha I$ (which is equivalent to the property that $-b^Z$ is $\alpha$-strongly monotone).
Then from the computation above, we can bound:
\begin{align*}
    \frac{d}{dt} \|b^Z(Z_t)\|^2
    &\le -2\alpha \, \| b^Z(Z_t) \|^2.
\end{align*}
Integrating this differential inequality gives:
\begin{align*}
    \|b^Z(Z_t)\|^2 \le e^{-2\alpha t} \, \|b^Z(Z_0)\|^2
\end{align*}
as claimed in the second inequality in~\eqref{Eq:ConvGFGrad}.
Combining this with the bound for the duality gap~\eqref{Eq:DetDGBound} yields:
\begin{align*}
    2\alpha \, \DG(Z_t) \le \|\nabla V(Z_t)\|^2 \le e^{-2\alpha t} \, \|\nabla V(Z_0)\|^2
\end{align*}
as desired.
\end{proof}

%%%%%%%%%%%%%%%%%%%%%%%%%%%%
\subsection{Convergence of Min-Max Gradient Descent}
\label{Sec:ReviewDetGD}

We consider the \textit{min-max gradient descent} with step size $\eta > 0$, which maintains the iterates $(x_k,y_k)_{k \ge 0}$ in $\R^{2d}$ which evolves following the update:
\begin{align*}
    x_{k+1} &= x_{k} - \eta \nabla_x V(x_{k},y_{k}) \\
    y_{k+1} &= y_{k} + \eta \nabla_y V(x_{k},y_{k}). 
\end{align*}
In terms of the joint variable $z_k = (x_{k},y_{k}) \in \R^{2d}$, we can write this as:
\begin{align}\label{Eq:MinMaxGD}
    z_{k+1} = z_k + \eta b^Z(z_k)
\end{align}
where $b^Z$ is the vector field defined in~\eqref{Eq:DetDefbz}.

\begin{theorem}\label{Thm:DetMinMaxGD}
    Assume Assumption~\ref{As:SCSmooth}.
    Let $z^* = (x^*,y^*) \in \R^{2d}$ be the unique equilibrium point.
    Suppose $(z_k)_{k \ge 0}$ evolves following the min-max gradient descent~\eqref{Eq:MinMaxGD} in $\R^{2d}$ with step size $\eta > 0$.
    If $\eta \le \frac{\alpha}{4L^2}$, then for all $k \ge 0$:
    \begin{align}\label{Eq:ConvGDDist}
        \|z_k - z^*\|^2 &\le e^{-\alpha \eta k} \, \|z_0 - z^*\|^2. 
    \end{align}            
    Furthermore, if $\eta \le \frac{\alpha}{16L^2}$, then we also have for all $k \ge 0$:
    \begin{align}
        2\alpha \, \DG(z_k) \,\le\, \|\nabla V(z_k)\|^2 
        &\le e^{-\alpha \eta k} \, \|\nabla V(z_0)\|^2 \label{Eq:ConvGDGrad}.
    \end{align}            
\end{theorem}
\begin{proof}
    \textbf{(1) Exponential convergence in distance:} Let $G \colon \R^{2d} \to \R^{2d}$ be $G(\bz) = \bz + \eta b^Z(\bz)$.
    Recall we show in Lemma~\ref{Lem:Contraction_bz} (with $N=1$) that $G$ is $M$-Lipschitz, where
    \begin{align*}
        M := \sqrt{1-2\eta \alpha + 4 \eta^2 L^2} \le \sqrt{1-\eta \alpha} \le e^{-\frac{1}{2} \eta \alpha}
    \end{align*}
    where the first bound above follows from the assumption $\eta \le \frac{\alpha}{4L^2}$, and the second from the inequality $1-c \le e^{-c}$ for $c \ge 0$.
    We can write the min-max gradient descent update~\eqref{Eq:MinMaxGD} as $z_{k+1} = G(z_k)$, and note $z^* = G(z^*)$ is a fixed point.
    Then we can compute:
    \begin{align*}
        \|z_{k+1} - z^*\|^2
        &= \|G(z_{k}) - G(z^*)\|^2
        \le M^2 \, \|z_{k} - z^*\|^2
    \end{align*}
    where the inequality follows from the property that $G$ is $M$-Lipschitz.
    Iterating this bound gives:
    \begin{align*}
        \|z_{k} - z^*\|^2
        \,\le\, M^{2k} \, \|z_{0} - z^*\|^2
        \,\le\, e^{-\alpha \eta k} \, \|z_{0} - z^*\|^2
    \end{align*}
    where the last inequality follows from the bound for $M$ above.

    \medskip
    \noindent
    \textbf{(2) Exponential convergence in gradient norm:} We consider a continuous-time interpolation of one step of the min-max gradient descent~\eqref{Eq:MinMaxGD} as:
    $$Z_t = Z_0 + t b^Z(Z_0)$$
    so that if $Z_0 = z_k$, then $Z_\eta = Z_0 + \eta b^Z(Z_0) = z_k + \eta b^Z(z_k) = z_{k+1}$.
    We can first bound:
    \begin{align*}
        \|b^Z(Z_t)-b^Z(Z_0)\|^2
        &\le 4L^2 \, \|Z_t-Z_0\|^2 \\
        &=\, 4t^2L^2 \, \|b^Z(Z_0)\|^2 \\
        &\le 8t^2L^2 \, \|b^Z(Z_t)-b^Z(Z_0)\|^2 + 8t^2L^2 \, \|b^Z(Z_t)\|^2 \\
        &\le \frac{1}{2} \, \|b^Z(Z_t)-b^Z(Z_0)\|^2 + 8t^2L^2 \, \|b^Z(Z_t)\|^2
    \end{align*}
    where the first inequality follows from the property that $b^Z$ is $(2L)$-Lipschitz from Lemma~\ref{Lem:LipschitzMonotone_bz}.
    In the second inequality we introduce the term $b^Z(Z_t)$ and use the inequality $\|a+b\|^2 \le 2\|a\|^2 + 2\|b\|^2$.
    In the third inequality we use the assumption $t \le \eta \le \frac{\alpha}{16L^2} \le \frac{1}{4L}$, so $8t^2L^2 \le \frac{1}{2}$.
    Rearranging the above and taking square-root gives us:
    \begin{align}\label{Eq:CalcGD1}
        \|b^Z(Z_t)-b^Z(Z_0)\|
        &\le 4tL \, \|b^Z(Z_t)\|.
    \end{align}
    Next, we can compute along the continuous-time interpolation, where $\dot Z_t = b^Z(Z_0)$, for $0 \le t \le \eta$:
    \begin{align*}
        \frac{d}{dt} \|b^Z(Z_t)\|^2
        &= 2 \langle b^Z(Z_t), \nabla b^Z(Z_t) \, b^Z(Z_0) \rangle \\
        &= 2 \langle b^Z(Z_t), \nabla b^Z(Z_t) \, b^Z(Z_t) \rangle + 2 \langle b^Z(Z_t), \nabla b^Z(Z_t) \, (b^Z(Z_0)-b^Z(Z_t)) \rangle \\
        &= 2 \langle b^Z(Z_t), (\nabla b^Z(Z_t))_\sym \, b^Z(Z_t) \rangle + 2 \langle b^Z(Z_t), \nabla b^Z(Z_t) \, (b^Z(Z_0)-b^Z(Z_t)) \rangle \\
        &\le -2\alpha \|b^Z(Z_t)\|^2 + 4L \|b^Z(Z_t)\| \cdot \|b^Z(Z_0)-b^Z(Z_t)\| \\
        &\le -2\alpha \|b^Z(Z_t)\|^2 + 16tL^2 \|b^Z(Z_t)\|^2 \\
        &\le -\alpha \|b^Z(Z_t)\|^2.
    \end{align*}
    In the first inequality above, we use the properties from Lemma~\ref{Lem:LipschitzMonotone_bz} that $(\nabla b^Z(Z_t))_\sym \preceq -\alpha I$ (equivalent to $-b^Z$ is $\alpha$-strongly monotone), and $b^Z$ is $(2L)$-Lipschitz.
    In the second inequality, we use the bound on $\|b^Z(Z_t)-b^Z(Z_0)\|$ from~\eqref{Eq:CalcGD1}.
    In the third inequality, we use the assumption $t \le \eta \le \frac{\alpha}{16L^2}$, so $16 t L^2 \le \alpha$.
    Integrating the differential inequality above from $t=0$ to $t=\eta$ gives:
    \begin{align*}
        \|b^Z(z_{k+1})\|^2 
        = \|b^Z(Z_\eta)\|^2
        \le e^{-\alpha \eta} \, \|b^Z(Z_0)\|^2
        = e^{-\alpha \eta} \, \|b^Z(z_k)\|^2.
    \end{align*}
    Iterating the recurrence above and recalling $\|b^Z(z)\| = \|\nabla V(z)\|$ gives us:
    \begin{align*}
        \|\nabla V(z_k)\|^2 = \|b^Z(z_{k})\|^2 
        \,\le\, e^{-\alpha \eta k} \, \|b^Z(z_0)\|^2
        = e^{-\alpha \eta k} \, \|\nabla V(z_0)\|^2
    \end{align*}
    as claimed in the second inequality in~\eqref{Eq:ConvGDGrad}.
    Finally, combining this bound with the bound for the duality gap~\eqref{Eq:DetDGBound} yields:
    \begin{align*}
        2\alpha \, \DG(z_k) \le \|\nabla V(z_k)\|^2 \le e^{-\alpha \eta k} \, \|\nabla V(z_0)\|^2
    \end{align*}
    as desired.
\end{proof}

%%%%%%%%%%%%%%%%%%%%%
\subsubsection{Corollary on iteration complexity of min-max gradient descent}

\begin{corollary}\label{Cor:MinMaxGDIter}
    Given any $\error > 0$, if we run the min-max gradient descent algorithm~\eqref{Eq:MinMaxGD}
    from $z_0 = (0,0) \in \R^{2d}$ with step size $\eta = \frac{\alpha}{4L^2}$, then we have
    $\|z_k - z^*\|^2 \le \error$
    for all
    $$k \ge \frac{4 L^2}{\alpha^2} \log \left( \frac{\|z^*\|^2}{\error} \right).$$
\end{corollary}
\begin{proof}
    This follows from the bound~\eqref{Eq:ConvGDDist} from Theorem~\ref{Thm:DetMinMaxGD}:
    \begin{align*}
        \|z_k - z^*\|^2
        \le \exp\left(-\alpha \eta k \right) \|z^*\|^2 
        = \exp\left(-\frac{\alpha^2 k}{4 L^2}\right)\|z^*\|^2 
        \le \error
    \end{align*}
    where the last inequality follows from our choice of $k$.
\end{proof}

%%%%%%%%%%%%%%%%%%%%%%%%%
\section{Zero-Sum Game in the Space of Distributions Without Regularization}
\label{Sec:ReviewZeroEntropy}

Consider the zero-sum game in the space of probability distributions without entropy regularization:
\begin{align}\label{Eq:GameZeroEntropy}
    \min_{\rho^X \in \P(\R^d)} \, \max_{\rho^Y \in \P(\R^d)} \; \E_{\rho^X \otimes \rho^Y}[V]
\end{align}
where $V \colon \R^d \times \R^d \to \R$ satisfies Assumption~\ref{As:SCSmooth}.
This is the same game~\eqref{Eq:Game0} as in Section~\ref{Sec:Introduction}.

We say that a pair of probability distributions $(\bar \nu^X, \bar \nu^Y) \in \P(\R^d) \times \P(\R^d)$ is an \textit{equilibrium distribution} for the game~\eqref{Eq:GameZeroEntropy} if the following holds for all $(\rho^X, \rho^Y) \in \P(\R^d) \times \P(\R^d)$:
\begin{align}\label{Eq:EquilibriumZeroEntropy}
    \E_{\bar \nu^X \otimes \rho^Y}[V] \le \E_{\bar \nu^X \otimes \bar \nu^Y}[V] \le \E_{\rho^X \otimes \bar \nu^Y}[V].
\end{align}
Under Assumption~\ref{As:SCSmooth}, there exists a unique equilibrium distribution, which is the point mass $(\delta_{x^*},\delta_{y^*})$, where $(x^*,y^*)$ is the equilibrium point of the deterministic game~\eqref{Eq:GameDet}, see Theorem~\ref{Thm:ZeroEntropy} below.

We consider the \textit{mean-field min-max gradient flow} which is the dynamics for random variables $Z_t = (X_t,Y_t) \sim \rho_t^X \otimes \rho_t^Y = \rho_t^Z$ in $\R^{2d}$ which evolves via:
\begin{subequations}\label{Eq:MinMaxGFZeroEntropy}
    \begin{align}
        \dot X_t &= -\E_{\rho_t^Y}[\nabla_x V(X_t, Y_t)] \\
        \dot Y_t &= \E_{\rho_t^X}[\nabla_x V(X_t, Y_t)].
    \end{align}
\end{subequations}

We also consider the \textit{mean-field min-max gradient descent} algorithm with step size $\eta > 0$ which maintains random variables $z_k = (x_k,y_k) \sim \rho_k^{x,\eta} \otimes \rho_k^{y,\eta} = \rho_k^{z,\eta}$ in $\R^{2d}$ with the update rule:
\begin{subequations}\label{Eq:MinMaxGDZeroEntropy}
    \begin{align}
        x_{k+1} &= x_k - \eta \, \E_{\rho_k^{y,\eta}}[\nabla_x V(x_k, y_k)] \\
        y_{k+1} &= y_k + \eta \, \E_{\rho_k^{x,\eta}}[\nabla_y V(x_k, y_k)].
    \end{align}
\end{subequations}

We have the following convergence guarantees.

\begin{theorem}\label{Thm:ZeroEntropy}
    Assume Assumption~\ref{As:SCSmooth}.
    Let $z^* = (x^*,y^*) \in \R^{2d}$ be the unique equilibrium point for the deterministic game~\eqref{Eq:GameDet}.
    Then:
    \begin{enumerate}
        \item There exists a unique equilibrium distribution for the game~\eqref{Eq:GameZeroEntropy}, which is the point mass distribution: $(\delta_{x^*},\delta_{y^*})$.
        \item Suppose $Z_t = (X_t,Y_t) \sim \rho_t^X \otimes \rho_t^Y = \rho_t^Z$ evolves via the mean-field min-max gradient flow~\eqref{Eq:MinMaxGFZeroEntropy} in $\R^{2d}$,
        and let $\delta_{z^*} = \delta_{x^*} \otimes \delta_{y^*}$.
        For all $t \ge 0$:
        \begin{align}\label{Eq:ConvGFZeroEntropy}
            W_2(\rho_t^Z, \delta_{z^*})^2 \le e^{-2\alpha t} \, W_2(\rho_0^Z, \delta_{z^*})^2.
        \end{align}
        \item Suppose $z_k = (x_k,y_k) \sim \rho_k^{x,\eta} \otimes \rho_k^{y,\eta} = \rho_k^{z,\eta}$ evolves following the mean-field min-max gradient descent~\eqref{Eq:MinMaxGDZeroEntropy} in $\R^{2d}$ with step size $0 < \eta \le \frac{\alpha}{4L^2}$.
        Then for all $k \ge 0$:
        \begin{align}\label{Eq:ConvGDZeroEntropy}
            W_2(\rho_k^{z,\eta}, \delta_{z^*})^2 \le e^{-\alpha \eta k} \, W_2(\rho_0^{z,\eta}, \delta_{z^*})^2.
        \end{align}
    \end{enumerate}    
\end{theorem}
\begin{proof}
    \textbf{(1)~Equilibrium distribution:}
    Recall since $(x^*,y^*) \in \R^{2d}$ is an equilibrium point for the deterministic game~\eqref{Eq:GameDet}, it satisfies the property~\eqref{Eq:DetEquilibrium} that for all $x,y \in \R^d$:
    \begin{align*}
        V(x^*,y) \le V(x^*,y^*) \le V(x,y^*).
    \end{align*}
    Then for any probability distributions $\rho^X, \rho^Y \in \P(\R^d)$, with $X \sim \rho^X$ and $Y \sim \rho^Y$, we have:
    \begin{align*}
        \E_{\rho^Y}[V(x^*,Y)] \le V(x^*,y^*) \le \E_{\rho^X}[V(X,y^*)].
    \end{align*}
    We can write this equivalently as the condition~\eqref{Eq:EquilibriumZeroEntropy}:
    \begin{align*}
        \E_{\delta_{x^*} \otimes \rho^Y}[V] \le \E_{\delta_{x^*} \otimes \delta_{y^*}}[V] \le \E_{\rho^X \otimes \delta_{y^*}}[V].
    \end{align*}
    This shows that $(\delta_{x^*}, \delta_{y^*})$ is an equilibrium distribution of the game~\eqref{Eq:GameDet}.

    Furthermore, note that any equilibrium distribution of the game~\eqref{Eq:GameDet} is a stationary distribution for the mean-field min-max gradient flow~\eqref{Eq:MinMaxGFZeroEntropy}.
    Then the uniqueness of the equilibrium follows from the convergence guarantee~\eqref{Eq:ConvGFZeroEntropy} that we will prove below, which shows that along the mean-field min-max gradient flow~\eqref{Eq:MinMaxGFZeroEntropy}, any starting distribution $\rho_0^Z$ converges to $\delta_{z^*} = \delta_{x^*} \otimes \delta_{y^*}$. 

    \medskip
    \noindent
    \textbf{(2)~Convergence of mean-field min-max gradient flow:} Let $b^Z \colon \R^{2d} \to \R^{2d}$ be the vector field defined in~\eqref{Eq:DetDefbz}. 
    Recall from Lemma~\ref{Lem:LipschitzMonotone_bz} that $-b^Z$ is $\alpha$-strongly monotone.
    Recall that $b^Z(z^*) = 0$.
        
    Let $Z_t = (X_t,Y_t) \sim \rho_t^X \otimes \rho_t^Y = \rho_t^Z$ evolve via the mean-field min-max gradient flow~\eqref{Eq:MinMaxGFZeroEntropy} in $\R^{2d}$.
    Then we can compute:
    \begin{align*}
        \frac{d}{dt} \|X_t-x^*\|^2
        &= 2 \langle X_t - x^*, \dot X_t \rangle \\
        &= -2 \left\langle X_t - x^*, \E_{\rho_t^Y}[\nabla_x V(X_t, Y_t)] \right\rangle \\
        &= -2 \E_{\rho_t^Y}\big[\left\langle X_t - x^*, \nabla_x V(X_t, Y_t) \right\rangle \big].
    \end{align*}
    Taking expectation over $X_t \sim \rho_t^X$ and writing $\rho_t^Z = \rho_t^X \otimes \rho_t^Y$ gives:
    \begin{align*}
        \frac{d}{dt} \E_{\rho_t^X}\left[\|X_t-x^*\|^2\right]
        &= -2 \E_{\rho_t^Z}\big[\left\langle X_t - x^*, \nabla_x V(X_t, Y_t) \right\rangle \big].
    \end{align*}
    Similarly, we can also compute:
    \begin{align*}
        \frac{d}{dt} \E_{\rho_t^Y}\left[\|Y_t-y^*\|^2\right]
        &= 2 \E_{\rho_t^Z}\big[\left\langle Y_t - y^*, \nabla_y V(X_t, Y_t) \right\rangle \big].
    \end{align*}
    Adding the two identities above gives:
    \begin{align*}
        \frac{d}{dt} \E_{\rho_t^Z}\left[\left\|Z_t-z^*\right\|^2\right]
        &= \frac{d}{dt} \E_{\rho_t^X}\left[\|X_t-x^*\|^2\right] + \frac{d}{dt} \E_{\rho_t^Y}\left[\|Y_t-y^*\|^2\right] \\
        &= 2 \E_{\rho_t^Z}\big[\left(\left\langle X_t - x^*, -\nabla_x V(X_t, Y_t) \right\rangle + \left\langle Y_t - y^*, \nabla_y V(X_t, Y_t) \right\rangle \right) \big] \\
        &= 2 \E_{\rho_t^Z}\big[\left\langle Z_t - z^*, b^Z(Z_t) \right\rangle \big] \\
        &= 2 \E_{\rho_t^Z}\big[\left\langle Z_t - z^*, b^Z(Z_t) - b^Z(z^*) \right\rangle \big] \\
        &\le -2\alpha \E_{\rho_t^Z}\left[\left\| Z_t - z^* \right\|^2 \right]
    \end{align*}
    where the inequality follows from the property that $-b^Z$ is $\alpha$-strongly monotone.
    Integrating the differential inequality above from $0$ to $t$, and noting the special formula for the $W_2$ distance to a point mass $\delta_{z^*}$ yields the desired convergence guarantee:
    \begin{align*}
        W_2(\rho_t^Z, \delta_{z^*})^2
        = \E_{\rho_t^Z}\left[\left\|Z_t-z^*\right\|^2\right]
        \le e^{-2\alpha t} \, \E_{\rho_0^Z}\left[\left\|Z_0-z^*\right\|^2\right]
        = e^{-2\alpha t} \, W_2(\rho_0^Z, \delta_{z^*})^2.
    \end{align*}

    \medskip
    \noindent
    \textbf{(3)~Convergence of mean-field min-max gradient descent:} 
    Let $G \colon \R^{2d} \to \R^{2d}$ be $G(z) = z + \eta b^Z(z)$.
    Note that $G(z^*) = z^*$.
    Recall from Lemma~\ref{Lem:Contraction_bz} that $G$ is $M$-Lipschitz, where
    $$M := \sqrt{1-2\eta \alpha + 4 \eta^2 L^2} \le \sqrt{1-\eta \alpha} \le e^{-\frac{1}{2} \eta \alpha}$$
    where the bound above follows from the assumption $\eta \le \frac{\alpha}{4L^2}$ and the inequality $1-c \le e^{-c}$.

    From the update of mean-field min-max gradient descent~\eqref{Eq:MinMaxGDZeroEntropy}, we can compute:
    \begin{align*}
        \E_{\rho_k^{x,\eta}} \left[\|x_{k+1}-x^*\|^2\right]
        &= \E_{\rho_k^{x,\eta}} \left[\left\|x_k - \eta \, \E_{\rho_k^{y,\eta}}[\nabla_x V(x_k, y_k)] - x^* \right\|^2\right] \\
        &\le \E_{\rho_k^{z,\eta}} \left[\left\|x_k - \eta \, \nabla_x V(x_k, y_k) - x^* \right\|^2\right]
    \end{align*}
    where the inequality follows from Cauchy-Schwarz and using $\rho_k^{z,\eta} = \rho_k^{x,\eta} \otimes \rho_k^{y,\eta}$.
    Similarly, we can compute:
    \begin{align*}
        \E_{\rho_k^{y,\eta}} \left[\|y_{k+1}-y^*\|^2\right]
        &= \E_{\rho_k^{y,\eta}} \left[\left\|y_k + \eta \, \E_{\rho_k^{x,\eta}}[\nabla_y V(x_k, y_k)] - y^* \right\|^2\right] \\
        &\le \E_{\rho_k^{z,\eta}} \left[\left\|y_k + \eta \, \nabla_y V(x_k, y_k) - y^* \right\|^2\right].
    \end{align*}
    Adding the two bounds above, and using $G(z)= z + \eta b^Z(z)$ and $G(z^*) = z^*$, we can bound:
    \begin{align*}
        \E_{\rho_k^{z,\eta}} \left[\|z_{k+1}-z^*\|^2\right]
        &\le 
        % \E_{\rho_k^{z,\eta}} \left[\left\|z_k + \eta \, b^Z(z_k) - z^* \right\|^2\right] 
        \E_{\rho_k^{z,\eta}} \left[\left\|G(z_k) - G(z^*) \right\|^2\right] 
        \le e^{-\alpha \eta} \, \E_{\rho_k^{z,\eta}} \left[\left\|z_k - z^* \right\|^2\right]
    \end{align*}
    where the last inequality follows from the property that $G$ is $M$-Lipschitz with $M \le e^{-\frac{1}{2} \alpha \eta}$.
    Iterating the bound above gives:
    \begin{align*}
        W_2(\rho_k^{z,\eta}, \delta_{z^*})^2 = \E_{\rho_k^{z,\eta}} \left[\left\|z_k - z^* \right\|^2\right] \le e^{-\alpha \eta k} \, \E_{\rho_0^{z,\eta}} \left[\left\|z_0 - z^* \right\|^2\right]
        = e^{-\alpha \eta k} \, 
        W_2(\rho_0^{z,\eta}, \delta_{z^*})^2
    \end{align*}
    as desired.
\end{proof}

\end{document}